\documentclass[11pt,letterpaper]{article}
\usepackage[left=1in, right=1in, top=1in, bottom=1in]{geometry}

\usepackage[utf8]{inputenc}
\usepackage[inline]{enumitem}
\usepackage[square,numbers,sort&compress]{natbib}
\usepackage[font=small]{caption}
\usepackage[ruled,noend]{algorithm2e}
\setlist[itemize]{noitemsep, topsep=5pt}
\setlist[enumerate]{noitemsep, topsep=5pt}

\usepackage{adjustbox}
\usepackage{algpseudocode}
\usepackage{amssymb}
\usepackage{amsfonts}
\usepackage{amsmath}
\usepackage{amsthm}
\usepackage{bm}
\usepackage{booktabs}
\usepackage[inline]{enumitem}
\usepackage{float}
\usepackage{footnote}
\usepackage{graphicx}
\usepackage{hyperref}
\usepackage{cleveref}
\usepackage{makecell} %
\usepackage{mathpazo}
\usepackage{mathtools}
\usepackage{multirow}
\usepackage{nicefrac}
\usepackage{placeins}
\usepackage{scalerel}
\usepackage{soul}
\usepackage{subcaption}
\usepackage{tikz}
\usepackage{varwidth}
\usepackage{verbatim}
\usepackage{xcolor}
\usepackage{xfrac}
\usepackage{xparse}
\usetikzlibrary{positioning}

\renewcommand{\phi}{\varphi}
\newcommand{\set}[1]{\left\{#1\right\}}

\newcommand{\sqbr}[1]{\left[ #1 \right]}
\newcommand{\brackets}[1]{\left( #1 \right)}
\definecolor{darkcyan}{rgb}{0.0, 0.7, 0.7}

\newcommand{\swir}{stagewise-IR\xspace}
\newcommand{\oir}{overall-IR\xspace}
\newcommand{\Swir}{Stagewise-IR\xspace}

\newcommand{\Oir}{Overall-IR\xspace}
\newcommand{\samono}{strongly-allocation-monotone\xspace}
\newcommand{\Samono}{Strongly-Allocation-Monotone\xspace}
\newcommand{\amono}{allocation-monotone\xspace}
\newcommand{\Amono}{Allocation-Monotone\xspace}
\newcommand{\wamono}{weakly-allocation-monotone\xspace}
\newcommand{\Wamono}{Weakly-Allocation-Monotone\xspace}

\newcommand{\rmono}{reward-monotone\xspace}

\newcommand{\fr}{fixed-rate\xspace}
\newcommand{\FR}{Fixed-Rate\xspace}
\newcommand{\swac}{SWAC\xspace}
\newcommand{\swacs}{SWACs\xspace}

\newcommand{\finauc}{finite-menu \swac}
\newcommand{\finaucs}{finite-menu \swacs}

\NewDocumentCommand{\xhat}{ O{} m g }{
    \IfNoValueTF{#2}{
        {\hat{\alloc}(#1)}
    }{
        {\hat{\alloc}^{#1}(#2)}
    }
}
\NewDocumentCommand{\x}{ O{} m g }{
    \IfNoValueTF{#2}{
        {\alloc(#1)}
    }{
        {\alloc^{#1}(#2)}
    }
}
\NewDocumentCommand{\alloc}{ O{} }{
    \ifblank{#1}{
        {X}
    }{
        {X^{#1}}
    }
}
\NewDocumentCommand{\phat}{ O{} m g }{
    \IfNoValueTF{#2}{
        {\hat{\pay}(#1)}
    }{
        {\hat{\pay}^{#1}(#2)}
    }
}
\NewDocumentCommand{\p}{ O{} m g }{
    \IfNoValueTF{#2}{
        {\pay(#1)}
    }{
        {\pay^{#1}(#2)}
    }
}
\NewDocumentCommand{\pay}{ O{} }{
    \ifblank{#1}{
        {P}
    }{
        {P^{#1}}
    }
}
\NewDocumentCommand{\spay}{ O{} }{
    \ifblank{#1}{
        {p}
    }{
        {p^{#1}}
    }
}
\NewDocumentCommand{\fhat}{ O{} m g }{
    \IfNoValueTF{#2}{
        {\hat{\filter}(#1)}
    }{
        {\hat{\filter}^{#1}(#2)}
    }
}
\NewDocumentCommand{\f}{ O{}  m g }{
    \IfNoValueTF{#2}{
        {\filter(#1)}
    }{
        {\filter^{#1}(#2)}
    }
}
\NewDocumentCommand{\filter}{ O{}}{
    \ifblank{#1}{
        {f}
    }{
        {f^{#1}}
    }
}

\newcommand{\reals}{\mathbb{R}}
\newcommand{\nat}{\mathbb{N}}
\DeclareMathOperator*{\E}{\bb{E}}
\DeclareMathOperator{\history}{\mathcal{H}}

\newcommand{\given}{\medspace | \medspace}

\newcommand \sett[2] { \left.\left\{{#1}\medspace\right\vert\medspace{#2}\right\}}

\newcommand{\Uni}{\mathsf{Uni}}
\renewcommand{\mathbf}{\bm}

\DeclareMathOperator{\then}{\Longrightarrow}
\DeclareMathOperator{\preals}{\reals_{\geq0}}
 
\newcommand\numberthis{\addtocounter{equation}{1}\tag{\theequation}}

\newcommand{\psqrrange}[1]{[#1]_+}
\newcommand{\sqrrange}[1]{[#1]}
\newcommand{\bb}[1]{\mathbb{#1}}
\DeclareMathOperator{\A}{\auc{A}}
\DeclareMathOperator{\B}{\mathcal{B}}
\DeclareMathOperator{\C}{\mathcal{C}}
\DeclareMathOperator{\D}{\mathcal{D}}

\DeclareMathOperator{\I}{\mathcal{I}}
\DeclareMathOperator{\M}{\mech{M}}
\renewcommand{\S}{\mathcal{S}}
\DeclareMathOperator{\T}{\mathcal{T}}
\DeclareMathOperator{\U}{\mathcal{U}}
\DeclareMathOperator{\V}{\mathcal{V}}
\newcommand{\vstar}{{v^*}}
\newcommand{\xstar}{{x^*}}
\newcommand{\g}{g}

\NewDocumentCommand{\mopt}{ O{} }{
    \ifblank{#1}{
        {\mech{M}^{\textsc{opt}}}
    }{
        {}
    }
}
\NewDocumentCommand{\trew}{ o m g}{ %
    \IfNoValueTF{#1}{ 
        {\textsc{Rew}\brackets{#2}}
    }{
        {\textsc{Rew}^{#1}\brackets{#2}}
    }
}
\NewDocumentCommand{\rew}{ o m m g}{ %
    \IfNoValueTF{#1}{ 
        {\textsc{Rew}\brackets{#2,#3}}
    }{
        {\textsc{Rew}^{#1}\brackets{#2,#3}}
    }
}
\NewDocumentCommand{\util}{ o m m g}{
    \IfNoValueTF{#1}{ 
        {u\brackets{#2,#3}}
    }{
        {u^{#1}\brackets{#2,#3}}
    }
}
\NewDocumentCommand{\utili}{ o m m m g}{
    \IfNoValueTF{#1}{
        {u_{#2}\brackets{#3,#4}}
    }{
        {u^{#1}_{#2}\brackets{#3,#4}}
    }
}
\NewDocumentCommand{\utilformula}{ o m m g}{
    \IfNoValueTF{#1}{
        {\x{#3}#2-\p{#3}}
    }{
        {\x[#1]{#3}#2-\p[#1]{#3}}
    }
}
\NewDocumentCommand{\strat}{ O{} m g}{ %
    \IfNoValueTF{#1}{
        \sigma\brackets{{#2}}
    }{
       \sigma^{#1}\brackets{#2}
    }
}
\NewDocumentCommand{\bid}{ O{} m g}{ %
    \IfNoValueTF{#1}{
        \mathsf{b}\brackets{#2}
    }{
       \mathsf{b}^{#1}\brackets{#2}
    }
}
\DeclareMathOperator{\remrt}{\textsc{RemRt}}
\DeclareMathOperator{\remlft}{\textsc{RemLft}}
\DeclareMathOperator*{\mylim}{\lim\limits_{i\rightarrow\infty}}

\newcommand{\irn}{\Bar{\varphi}}
\newcommand{\argmax}[1]{\arg\underset{#1}{\max}\ }

\newcommand{\sumof}[3]{\sum_{{#1}={#2}}^{#3}\limits}
\newcommand{\insumi}[2]{\sum_{i={#1}}^{#2}}
\newcommand{\insumof}[3]{\sum_{{#1}={#2}}^{#3}}
\NewDocumentCommand{\ind}{ O{} m g }{
    \IfNoValueTF{#2}{
        \bb{I}_{#1}
    }{
        \bb{I}_{#2}^{#1}
    }
}

\theoremstyle{plain}
\newtheorem{theorem}{Theorem}[section]
\newtheorem{claim}{Claim}
\newtheorem{lemma}[theorem]{Lemma}
\newtheorem{corollary}[theorem]{Corollary}
\newtheorem{proposition}[theorem]{Proposition}

\newtheorem{definition}{Definition}

\newtheorem*{lemma*}{Lemma}
\newtheorem*{theorem*}{Theorem}
\newtheorem*{proposition*}{Proposition}
\newtheorem{observation}{Observation}

\newtheorem{example}[theorem]{Example}

\renewcommand{\include}{\input}

\newcommand{\auc}[1]{\mathsf{#1}}
\newcommand{\mech}[1]{\mathsf{#1}}

\title{Dynamic Rental Games with\\Stagewise Individual Rationality
} 

\author{
Batya Berzack\thanks{Tel Aviv University, \url{batyaberzack@mail.tau.ac.il}} 
\and
Rotem Oshman\thanks{Tel Aviv University, \url{roshman@tau.ac.il}}
\and
Inbal Talgam-Cohen\thanks{Tel Aviv University, \url{inbaltalgam@gmail.com}}
}

\date{}

\begin{document}

\begin{titlepage}
\maketitle
\begin{abstract}

We study \emph{rental games}---a single-parameter dynamic mechanism design problem, in which a designer rents out an indivisible asset over $n$ days. Each day, an agent arrives with a private valuation per day of rental, drawn from that day's (known) distribution. The designer can either rent out the asset to the current agent for any number of remaining days, charging them a (possibly different) payment per day, or turn the agent away. Agents who arrive when the asset is not available are turned away. A defining feature of our dynamic model is that agents are \emph{\swir} (individually rational), meaning they reject any rental agreement that results in temporary negative utility, even if their final utility is positive. We ask whether and under which economic objectives it is useful for the designer to exploit the \swir nature of the agents. 

We show that an optimal rental mechanism can be modeled as a sequence of dynamic auctions with seller costs. However, the \swir behavior of the agents makes these auctions quite different from classical single-parameter auctions: Myerson's Lemma does not apply, and indeed we show that truthful mechanisms are not necessarily monotone, and payments do not necessarily follow Myerson's unique payment rule. We develop alternative characterizations of optimal mechanisms under several classes of economic objectives, including generalizations of welfare, revenue and consumer surplus. These characterizations allow us to use Myerson's unique payment rule in several cases, and for the other cases we develop optimal mechanisms from scratch. 
Our work shows that rental games raise interesting questions even in the single-parameter regime.

\end{abstract}

\thispagestyle{empty} 
\end{titlepage}
\pagenumbering{arabic}
\bibliographystyle{apalike}
\sloppy

\section{Introduction}\label{sec:intro}

Renting rather than buying is a primary form of consumption in many markets --- from housing, through books, to formal wear.
In today's sharing economy, there is a further shift towards the share and reuse of assets, often facilitated by computational platforms (e.g.,  
Airbnb or Rent the Runway).
This raises a natural and timely question of \emph{rental mechanism design}. 

The inherent temporal nature of rental puts it within the realm of dynamic mechanism design; but while dynamic \emph{selling} mechanisms have attracted much attention recently, dynamic \emph{rental} mechanisms are relatively unexplored. In rental, there is a \emph{durable} asset, and an important consideration is opportunity cost --- ``renting differs from selling in that there exist new trade opportunities even after a positive transaction'' \cite{beccuti2020optimality}.
In this work we study a single-parameter 
rental game, focusing on connections to selling mechanisms as well as on new design challenges.

\paragraph{Our rental game model.} 

Imagine owning a valuable asset that can be rented out over a fixed time horizon of $n$ days.
Each day $i\in[n]$, a new agent arrives, wishing to rent the asset for as long as possible.
The agent has a single-dimensional private \emph{valuation}, representing how much they value a day of rent. 
The valuation is drawn from a known prior distribution $\D_i$ (independently of other agents' valuations).%
\footnote{Knowing the distribution for each day resembles the prophet inequality setting, and in fact 
our model can be viewed as a generalization of a model by~\cite{AbelsPS23} on prophet inequalities over time. However in~\cite{AbelsPS23}, agent valuations are \emph{known} to the designer, 
as well as other key differences as we discuss in Section~\ref{sec:related}.}
We can either enter into a rental agreement with the agent, or turn them away. 
An agreement specifies the length of the rental period and a (possibly non-uniform) payment per rental day. 
Importantly, the rental agreement is \emph{irrevocable}, creating an opportunity cost; for example, if we rent out the item for three days to agent $1$, then the agents arriving on days $2$ and $3$ must be turned away (see Example~\ref{ex:intro-example}).
Our goal is to maximize some reward function $g$,
such as revenue or social welfare, in total over the entire $n$ days.
What is the optimal rental mechanism that interacts with the agents and chooses the rental agreements? 

\paragraph{Stagewise individual rationality (IR).}
The answer to this question depends on what we require from the mechanism. A standard requirement is truthfulness, and another is individual rationality (IR). Interestingly, IR has more than one interpretation for mechanisms that take place over multiple stages. 
Like several recent works on dynamic mechanism design, we choose to focus on the requirement of \emph{ex post} IR, also referred to as \emph{stagewise}-IR or \emph{limited liability} of the agent \cite{MirrokniLTZ16,MirrokniLTZ18,BravermanSW21,0001DMS18,AshlagiDH16,BalseiroML17}.
Under \swir, no agent will enter into an agreement that temporarily results in negative cumulative utility, even if eventually their cumulative utility becomes positive. In particular, such agents reject large upfront payments that exceed their valuation. 
For example, an agent with valuation $v>0$ will not agree to transfer an upfront payment of $2v$ on the first day to rent the item for 3 days, because their utility on the first day would be negative ($v - 2v < 0$), even though their utility at the end of the rental period would be positive ($3v - 2v > 0$).

In~\cite{MirrokniLTZ18} it is noted that the requirement of \swir is in fact self-imposed by the agents; we thus extend the term and refer to \swir \emph{agents} in addition to \swir \emph{mechanisms}. We are interested in the optimal rental mechanism for \swir agents, and whether or not the designer will take advantage of the agents' self-imposed limitation.%
\footnote{The alternative to \swir\ is \oir, where agents are concerned only with their final utility, determined by the total payment they are charged (and the number of rental days). Since per-day payments no longer matter, this makes the mechanism design problem significantly simpler, and we show in Appendix~\ref{appendix:oir} that it can be solved by generalizing Myersonian auction theory.}

\paragraph{Objectives (a.k.a.~reward functions of the mechanism designer).}

A standard objective for mechanism design is maximizing revenue. Some rental mechanisms are indeed for profit, but some have different objectives: consider for example a public library, a community center renting out gear, public or on-campus housing, etc. These are examples of rental mechanisms maximizing social welfare, or even the renters' overall utility --- known as \emph{consumer surplus}. The latter objective is especially important in \emph{money burning} scenarios~\cite{HartlineR08}, where the payments made by the agents go to waste and so should be subtracted from their value for renting the asset. For this reason, there is growing recent interest in consumer surplus~\cite{ezra2024optimalmechanismsconsumersurplus, ganesh2023combinatorial, goldner2026multidimensional}.

\paragraph{From rental game to stagewise auctions (\swacs).}
The rental game is essentially a sequence of dynamic multi-unit auctions where each unit corresponds to a rental day. Since the agent in the rental game is \swir, the auction must also be \swir. 
We establish a reduction from the online rental game to an offline problem: designing a predetermined sequence of dynamic multi-unit auctions. We show that w.l.o.g., each such auction is of a simple form: the buyer submits a bid, which deterministically determines the number of days rented and required payment for each day (i.e., there is no benefit to interacting with the buyer beyond the initial bid, or to randomization). To capture the opportunity cost incurred from entering into an agreement with the current agent and turning some of the future agents away, the auction will include a \emph{seller's cost} that increases with the length of the allocation, and depends on the valuation distributions of the agents that will be turned away. The cost is subtracted from the seller's reward from the stagewise auction. We name these auctions \emph{stagewise auctions with seller cost}, or \swacs for short.
The above reduction leads to the problem of designing optimal \swacs in various settings, motivating the analysis of \swacs in general. Our results apply to arbitrary cost functions, and not only the opportunity cost for the rental. In subsequent work~\cite{berzack2026optimal}, this general formulation is used to extend the rental model to multiple buyers arriving each day, while allowing the asset to be rented to at most one buyer at a time.

\paragraph{Challenges in designing \swacs.}

In classic single-parameter auction design, Myerson's theory shows an equivalence between monotone allocation rules and truthfulness of the auction. In \swacs, due to their temporal nature and the \swir requirement, Myerson's theory~\cite{myerson1981optimal} no longer holds as is. This gives rise to interesting phenomena, and chief among these is that a truthful \swac\ need not be monotone. Also, the total payments in a truthful \swac need not follow Myerson's unique payment rule, even if it is monotone.

\begin{example}[Truthful but non-monotone \swac]
    \label{ex:intro-example}
    Consider a \swac\ 
    where the agent's valuation is uniform over the range $[0,8]$.
    The auction proceeds as follows:
    \begin{enumerate}
        \item For bids below $4$, the agent rents the asset for six days, and pays $2$
        each day, for a total of $12$.
         \item For bids $4$ and above, the agent rents the asset for five days, and pays $4$ upfront on the first day and zero on the remaining days.
    \end{enumerate}
    Under the first offer, the overall utility at the end of the rental period for an agent with valuation $v$ is $6v - 12$,
    and under the second offer it is $5v - 4$.
    Thus, all agents in our distribution would derive higher overall utility from the second offer  compared to the first. However, \swir\ agents with $v<4$ would not accept the second offer,
    because their utility on the first day would be negative
    ($v - 4 < 0$).
    Only agents with valuation at least $4$ can accept the second offer; the high payment on the first day \emph{filters out} agents with lower valuations.
    
    It is not hard to see that this mechanism is
    \emph{truthful}\footnote{
    Agents with valuations below $4$ cannot bid $\geq 4$, and all other bids yield the same result as bidding truthfully. Agents with valuations $\geq 4$ prefer the second offer, and this can be achieved by bidding truthfully.
    }.
    However, it is
    \emph{non-monotone in the agent valuation}
    with respect to both the allocation (i.e., the number of rental days)
    and the revenue of the designer:
    agents with valuation $3$ receive an allocation of six days and net the designer a revenue of $12$,
    while agents with valuation $4$ receive an allocation of five days and net the designer a revenue of $4$. 
\end{example}

The broken equivalence between monotonicity and truthfulness makes the design of optimal \swacs for different objectives a challenging task.
A key question that arises in our work is whether and when an \emph{optimal} \swac is monotone, and how monotonicity impacts the structure and payments of optimal \swacs.

\subsection{Our Results}
\label{sec:results}

We consider four families of reward functions,
i.e., objectives for the renter.
Consider a day of rental, let $v$ be the valuation of the agent who rents the asset that day, and let $p$ be the agent's payment for the day. 
Let $\alpha,\beta>0$ be strictly positive scalars. The classes are:
\begin{itemize} 
    \item \emph{Welfare-like} reward: a non-decreasing function $f(v)$ of the valuation alone;
    \item \emph{Revenue-like} reward: a linear function $\beta p$ of the payment alone;
    \item \emph{Positive tradeoff} reward, which has the form $g(v, p) = \alpha v + \beta p$;
    \item \emph{Negative tradeoff} reward, which has the form $g(v, p) = \alpha v - \beta p$ with $\alpha\geq\beta$. %
\end{itemize}
In all cases, the designer aims to maximize the total sum of rewards over all days. We note that positive tradeoff captures interpolations of welfare and revenue, while negative tradeoff covers consumer surplus,
$g(v,p) = v-p$,  among other objectives.
\footnote{We treat welfare-like and revenue-like reward separately from positive and negative tradeoff, even though they could be viewed as a special case, because our results for these first two classes are more general.}

\begin{table}[h!]
\centering
\begin{adjustbox}{max width=\textwidth}
\begin{tabular}{|l|c|c|c|c|}
\hline
 & \makecell{\textbf{Welfare-Like}\\\textbf{($f(v)$)}}
 & \makecell{\textbf{Revenue-Like}\\\textbf{($\beta p$)}}
 & \makecell{\textbf{Positive Tradeoff}*\\ \textbf{($\alpha v+\beta p$)}} 
 & \makecell{\textbf{Negative Tradeoff}* \\E.g. Consumer Surplus\\ \textbf{($\alpha v-\beta p$)}} \\ 
\hline

\makecell{Are Optimal Auctions\\Allocation-Monotone?\dag}  
& \multirow{3}{*}{\makecell{Yes\\(Lemma~\ref{lemma:welfare-like})}} 
& \makecell{Yes\\(Lemma~\ref{lemma:rev-like-wamono})} 
& \makecell{Yes\\(Lemma~\ref{lemma:pos-neg-tradeoff-samono})}
& \makecell{Yes\\(Lemmas~\ref{lemma:pos-neg-tradeoff-samono},\ref{lemma:consumer-surplus-amono})} \\ 
\cline{1-1}\cline{3-5}

\makecell{Are Optimal Auctions\\Reward-Monotone?\dag}  
& 
& \makecell{Yes\\(Lemma~\ref{lemma:rev-like-rmono})} 
& \makecell{Yes\\(Lemma~\ref{lemma:pos-tradeoff-rmono})} 
& \makecell{No\\(Lemma~\ref{lemma:consumer_surplus_filter})} \\ 
\cline{1-1}\cline{3-5}
 
\makecell{Are Fixed-Rate Payment\\Schedules W.L.O.G?}  
& 
& \makecell{Yes\\(Corollary~\ref{corollary:rev-like-fr-wlog})} 
& \makecell{Yes\\(Lemma~\ref{lemma:positive-tradeoff-fr-wlog})} 
& \makecell{No\\(Corollary~\ref{cor:neg-tadeoff-fr-not-wlog})} \\ 
\hline

\makecell{Optimal Rental\\Mechanism}  
& \multicolumn{3}{c|}{Algorithm~\ref{alg:fr-rental}} 
& Algorithm~\ref{alg:consumer-surplus-auc}\ddag \\ 
\hline
\end{tabular}
\end{adjustbox}

\caption[Informal Summary of Our Results]{Informal Summary of Our Results \\
* For these reward functions, some of the allocation-monotonicity results are restricted to finite-menu auctions, where the bidders are offered a \emph{finite} menu of options to choose from. 
The other results for these reward functions hold for any \swac\ that is \amono.\\
\dag Except for a possible set of measure zero.\\
\ddag This algorithm is restricted to i.i.d.\ distributions.
}
\label{tab:cases}
\end{table}

Our results for each class of reward functions are summarized in Table~\ref{tab:cases}. 
The first three rows apply to \swacs,
and the fourth row
refers to
the resulting rental mechanisms,
which can be computed by polynomial-time algorithms.
The results for welfare- and revenue-like reward
apply to any \swac and any sequence of agent distributions that have a pdf;
some of the results for positive and negative tradeoff apply to \emph{\finaucs}, where the designer offers a finite menu of options,
and the optimal mechanism for negative tradeoff is further restricted to i.i.d.\ agent distributions.
The techniques we use to establish these results are described in detail in our technical overview in Sections~\ref{subsec:rental-to-swacs},~\ref{sec:tech-overview-mono} and~\ref{sec:tech-overview-rental-mechs}.
We briefly describe each result in the table.

The first two rows in Table~\ref{tab:cases} summarize our results on monotonicity with respect to the allocation (\emph{allocation-monotonicity}) and the designer's reward (\emph{reward-monotonicity}).
As demonstrated in Example~\ref{ex:intro-example}, a truthful \swac\ need not be allocation-monotone, i.e., a higher valuation may result in 
fewer allocated units. 
This already shows that Myerson's Lemma~\cite{myerson1981optimal} does not apply, so optimal \swacs cannot be derived using standard techniques.
Example~\ref{ex:intro-example} also shows that a truthful \swac\ can be non-monotone in the reward, i.e., a higher valuation
may result in a lower total 
revenue
for the designer.
One of our key results is that for many classes of reward functions, an \emph{optimal} \swac\ recovers these monotonicity properties. 

The next two rows in Table~\ref{tab:cases} refer to
payment schedules and the translation from \swacs back to rental mechanisms: 
one benefit of monotonicity is that as we show,
for welfare-like, revenue-like, and positive tradeoff reward,
we can
transition to \emph{\fr payment schedules}, which %
charge the agent the same payment every day.
Once this is done, we can recover enough of Myerson's lemma to derive an optimal fixed-rate \swac, and thus, an optimal rental mechanism. Notably, the resulting rental mechanisms have a simple and practical format: At each horizon, a menu of fixed-rate rental agreements. 
Simplicity of the mechanism is a crucial consideration in real-life mechanism design~\cite{hartline2009simple}.

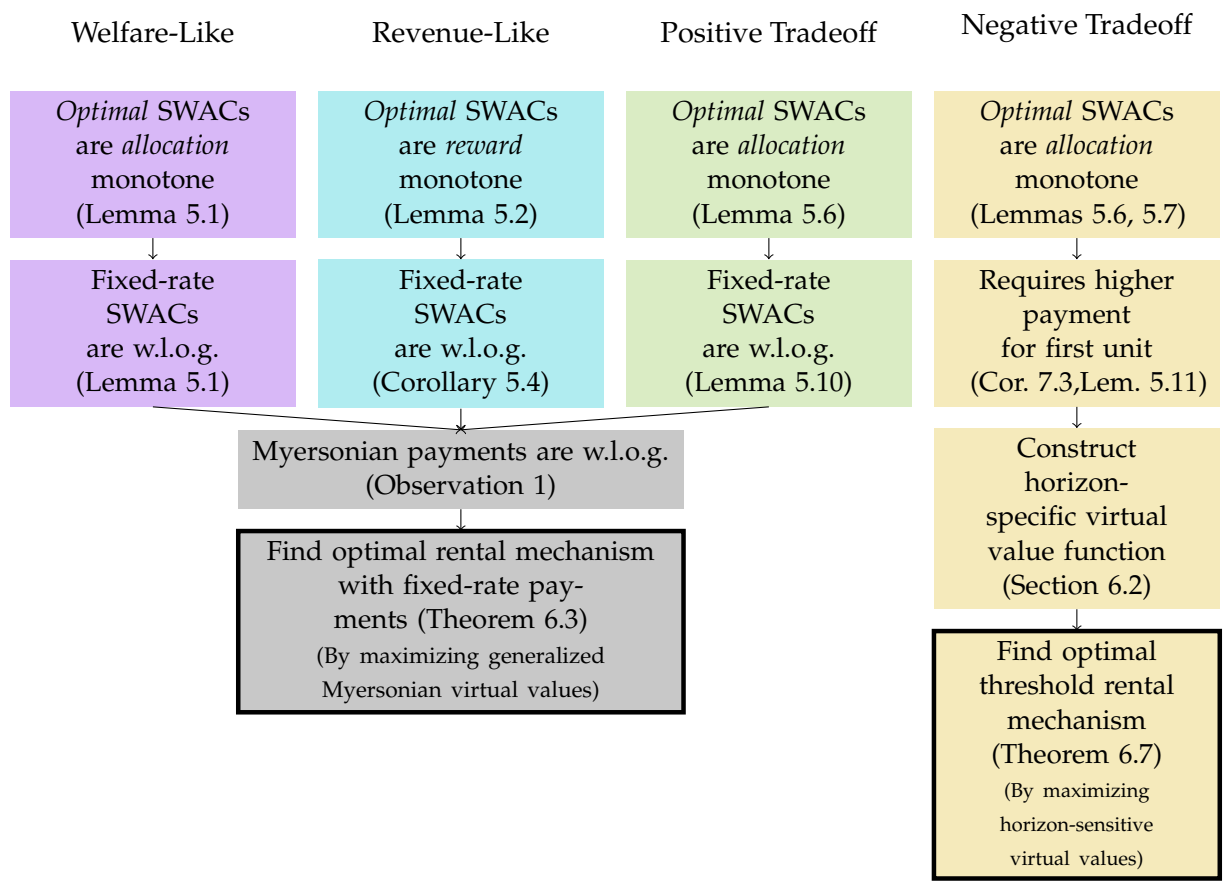
\begin{figure}[h]
\centering
\begin{tikzpicture}[x=0.75pt,y=0.75pt,yscale=-1,xscale=1, %
    node distance=8pt,
    purplebox/.style={fill={rgb, 255:red, 216; green, 184; blue, 247 }, fill opacity=1, minimum width=4cm, minimum height=0.8cm, text width=85pt, align=center,scale=0.95}, 
    greenbox/.style={fill={rgb, 255:red, 176; green, 236; blue, 240 }, fill opacity=1, minimum width=4cm, minimum height=0.8cm, text width=85pt, align=center,scale=0.95}, 
    bluebox/.style={fill={rgb, 255:red, 218; green, 236; blue, 196 }, fill opacity=1, minimum width=4cm, minimum height=0.8cm, text width=85pt, align=center,scale=0.95}, 
    yellowbox/.style={fill={rgb, 255:red, 245; green, 234; blue, 187 }, fill opacity=1, minimum width=4cm, minimum height=0.8cm, text width=85pt, align=center,scale=0.95}, 
    greybox/.style={fill={rgb, 255:red, 200; green,200; blue, 200 }, fill opacity=1, minimum width=2.5cm, minimum height=0.8cm, text width=170pt, align=center,scale=0.95}, 
    doubleborder/.style={draw, line width=0.2mm,  double=black, double distance=0.2mm} %
    ]

\node (A1) [purplebox] {\textit{Optimal} \swacs\ are \textit{allocation} monotone\\(Lemma~\ref{lemma:welfare-like})}; 
\node (B1) [right=of A1, greenbox] {\textit{Optimal} \swacs\ are \textit{reward} monotone\\(Lemma~\ref{lemma:rev-like-rmono})};
\node (C1) [right=of B1, bluebox] {\textit{Optimal} \swacs\ are \textit{allocation} monotone\\(Lemma~\ref{lemma:pos-neg-tradeoff-samono})}; 
\node (D1) [right=of C1, yellowbox] {\textit{Optimal} \swacs\ are \textit{allocation} monotone\\(Lemmas~\ref{lemma:pos-neg-tradeoff-samono},~\ref{lemma:consumer-surplus-amono})};

\node [above=0.5cm of A1, align=center] {Welfare-Like};
\node [above=0.5cm of B1, align=center] {Revenue-Like};
\node [above=0.5cm of C1, align=center] {Positive Tradeoff};
\node [above=0.5cm of D1, align=center] {Negative Tradeoff};

\node (A2) [below=of A1, purplebox] {Fixed-rate\\\swacs\\are w.l.o.g.\\(Lemma~\ref{lemma:welfare-like})};
\node (B2) [below=of B1, greenbox] {Fixed-rate\\\swacs\\are w.l.o.g.\\(Corollary~\ref{corollary:rev-like-fr-wlog})}; 
\node (C2) [below=of C1, bluebox] {Fixed-rate\\\swacs\\are w.l.o.g.\\(Lemma~\ref{lemma:positive-tradeoff-fr-wlog})};
\node (D2) [below=of D1, yellowbox] {Requires higher payment for first unit\\(Cor.~\ref{cor:neg-tadeoff-fr-not-wlog},Lem.~\ref{lemma:neg-tradeoff-threshold-auction-truthful-and-optimal})}; 

\node (ABC1) [below=of B2, greybox] {Myersonian payments are w.l.o.g.\\(Observation~\ref{ob:fr-means-myerson})};
\node (D3) [below=of D2, yellowbox] {Construct horizon-specific virtual value function\\(Section~\ref{sec:threshold-mech})}; 

\node (ABC2) [below=of ABC1, greybox, doubleborder] {Find optimal rental mechanism\\with  fixed-rate payments (Theorem~\ref{thm:optima-fr-rental-mech})\\{\footnotesize (By maximizing generalized Myersonian virtual values)}};
\node (D4) [below=of D3, yellowbox, doubleborder] {Find optimal \\threshold rental mechanism (Theorem~\ref{theorem:threshold-mech})\\{\scriptsize (By maximizing }\\{\scriptsize horizon-sensitive}\\{\scriptsize  virtual values)}};

 \draw[->]       (A1.south) --(A2.north);
 \draw[->]       (B1.south) --(B2.north);
 \draw[->]       (C1.south) --(C2.north);
 \draw[->]       (D1.south) --(D2.north);

 \draw[->]       (A2.south) --(ABC1.north);
 \draw[->]       (B2.south) --(ABC1.north);
 \draw[->]       (C2.south) --(ABC1.north);
 \draw[->]       (D2.south) --(D3.north);

 \draw[->]       (ABC1.south) --(ABC2.north);
 \draw[->]       (D3.south) --(D4.north);

\end{tikzpicture}
\caption[Flow of proofs for welfare-like, revenue-like, positive tradeoff and negative tradeoff reward]{Flow of proofs for welfare-like, revenue-like, positive tradeoff and negative tradeoff reward. For each, we give monotonicity results (for either allocation or reward), 
establish the form of the optimal mechanism (w.l.o.g.),
and finally give an optimal rental mechanism.
Although some intermediate results appear to be the same across different types, each result is proven separately and leverages characteristics specific to the reward.}
\label{fig:flowchart}
\end{figure}

\tikzset{every picture/.style={line width=0.75pt}}
\usetikzlibrary{positioning}

The optimality of \fr payment schedules is no longer true for the class of negative tradeoff reward functions (e.g., consumer surplus $g(v,p) = v-p$). %
For such reward functions, the designer is better off selecting only agents with high valuations for long rentals, but charging them low payments.
We show that this screening effect cannot be achieved by mechanisms that charge fixed-rate payments,
and is necessary for optimality.
Nevertheless, we show that there is still an optimal mechanism with a simple structure:
we charge the entire cost of the rental on the first day, and zero on the remaining days.
Despite the simple structure, designing and analyzing mechanisms for negative tradeoff requires some new ideas---interestingly, we obtain it by optimizing a generalized horizon-dependent notion of Myersonian virtual values (see Section~\ref{sec:threshold-mech}).
One take-away message is that when the designer's objective is negative tradeoff, %
utilizing \swir is no longer ``taking advantage'' of the agent but rather is beneficial for both designer and agent.

As outlined above,
our first step is to reduce the rental mechanism design problem to a predetermined sequence of stagewise auctions with seller cost (\swacs),
and this is explained in Section~\ref{subsec:rental-to-swacs}.
Next, for each reward class, we establish some structural properties, and use them to derive an optimal \swac for that class.
The flow of the proof is illustrated in Figure~\ref{fig:flowchart}.

\paragraph{Paper organization.}
We begin with a discussion of related work in Section~\ref{sec:related}.
In Section~\ref{sec:model}, we formally define the rental game, the rental mechanism, and stagewise auctions.
In Section~\ref{subsec:rental-to-swacs} we present our reduction from rental mechanisms to \swacs.
In Section~\ref{sec:properties}, we establish key properties of \swacs\ that are used throughout the paper and provide some intuition.
In Section~\ref{sec:tech-overview-mono}, we explore structural results on optimal \swacs\ for different classes of reward functions. In Section~\ref{sec:tech-overview-rental-mechs}, we derive optimal rental mechanisms using these results on \swacs.
In Section~\ref{sec:fr-gap}, we illustrate through examples how non-\fr\ pricing can significantly impact the designer's reward.
Finally, in Section~\ref{sec:our-techniques} we summarize our techniques, and in Section~\ref{sec:future} discuss potential applications and future research directions.
For brevity, some technical proofs are deferred to the appendix.

\subsection{Related Work}
\label{sec:related}

\paragraph{Renting a Durable Good.} Renting is a fundamental mode of consumption, appearing in markets such as housing, books, and fashion. For a mechanism design perspective on the sharing economy, see \cite{OrenR21}. While selling mechanisms for durable goods have been extensively studied in economics \cite{doval2024optimal}, rental markets introduce a temporal component: assets remain available for future transactions, making opportunity cost a key factor \cite{beccuti2020optimality,cremer2003rental}. These differences place rental mechanisms within the broader realm of dynamic mechanism design, but unlike dynamic selling, rental mechanisms remain relatively unexplored. Our work contributes to filling this gap by formalizing a single-parameter rental game, highlighting connections to selling mechanisms while addressing novel design challenges.

\paragraph{Prophet Inequalities Over Time.}
Prophet inequalities have long been known to have deep connections to auction design~\cite{HajiaghayiKS07}, including applications to sequential posted-price mechanisms \cite{KleinbergW12}.
\cite{AbelsPS23} take a different approach: in ~\cite{AbelsPS23}, 
instead of selling goods to agents, they lend them over time.  In their model, each agent has an i.i.d.~valuation drawn from a known distribution, and the decision-maker selects an optimal lending strategy to maximize the expectation of the sum of valuations taken in each step (i.e., welfare). A crucial distinction is that their model does not involve payments, as valuations are revealed upon arrival. Their work provides a single-threshold algorithm that determines whether to lend to an agent and for how long.

In contrast, our model considers valuations drawn from different distributions across days, necessitating a more complex allocation structure (there can be multiple thresholds in the allocation rule). We design mechanisms that incorporate payments to ensure truthful behavior, which is essential when valuations are private. Furthermore, our work generalizes beyond welfare maximization by considering multiple economic objectives, including revenue and consumer surplus. Our results generalize their approach by addressing these challenges while maintaining incentive compatibility. 
 Recent works extend \cite{AbelsPS23} to a graph setting \cite{CristiO24} and improve approximation bounds \cite{perezsalazar2024}.
 
\paragraph{Stagewise Individual Rationality.} \Swir\ is a leading requirement in the literature on dynamic auctions (with repeated sales rather than a rental)~\cite{MirrokniLTZ16,MirrokniLTZ18,BravermanSW21}, and is sometimes referred to as \emph{per-round} individual rationality or \emph{ex-post IR}\footnote{Although ex-post IR usually does not relate to utility per stage, this is how it is defined in~\cite{MirrokniLTZ19}}.
 There are many economic reasons why agents might be \swir: for example, they might be unable to go into debt, or they might not have sufficient trust established with the designer to take on negative utility in the belief that the designer will keep to an agreement that will eventually result in positive utility.
~\cite{0001DMS18,AshlagiDH16,BalseiroML17}  study a slightly stronger notion of IR, in which the agent's utility in each stage, standalone, must be nonnegative\footnote{All of our results continue to hold if we require this stronger form of IR instead of \swir.} 
(see also \cite{PapadimitriouPP16,collina2024repeated}).
\cite{BravermanSW21} consider repeated auctions with adversarial valuations and per-round IR, assuming knowledge of total buyer valuations—a key distinction from our model.
\cite{0001DMS18} approximate revenue-optimal auctions under heterogeneous buyer behaviors, whereas we focus on rental mechanisms, where seller opportunity cost plays a central role, and explore optimal designs for multiple objectives beyond revenue, including welfare and consumer surplus.
 \cite{MirrokniLTZ16} develop bank account mechanisms for repeated combinatorial auctions with a form of IR that is not far from \swir. 
\cite{BalseiroML17} introduce martingale utility constraints, which we do not impose, allowing greater flexibility in the design of mechanisms.  \cite{krahmer2016optimality} study multi-unit contracts under ex-post participation constraints, though their model differs as buyers do not always know their valuations. 
Related sequential vs. simultaneous sales mechanisms appear in \cite{hausch1986multi,PapadimitriouPP16,AshlagiDH16}. 
\cite{BergemannCW17} focus on revenue-maximizing sequential auctions under ex-post IR constraints, which also allows for some sort of screening, but different from ours; their screening is used to offer different auction terms, depending on the type that the agent reports, and they have no allocation cost. They analyze when static contracts are optimal and when dynamic ones (i.e. with screening) are optimal, depending on the valuation distribution.

\paragraph{Dynamic auctions.} Dynamic auctions are essentially repeated auctions that are optimized over time, adapting payments based on past events \cite{bergemann2019dynamic,bergemann2010dynamic}. \cite{MirrokniLTZ19} study dynamic auctions with multiple buyers and sequentially arriving items.
 \cite{0001DMS18}  examine revenue-maximizing repeated auctions with a single buyer, assuming limited trust in the seller. In contrast, we assume full commitment, enabling exact optimal mechanisms.
\cite{DevanurPS15} consider a repeated auction to a single buyer with a fixed value for the item, drawn from a known distribution, which is similar to our dynamic auction, but does not include a seller cost function. Another major difference is that they consider only cases where the seller does not have full commitment over future prices, which greatly limits the designer's capabilities. Additionally, they consider a Perfect Bayesian equilibrium (which is why they don't consider full commitment), whereas we consider dominant strategy equilibrium.  \cite{MirrokniLTZ18} focus on non-clairvoyant mechanisms, restricting future distributional knowledge, whereas our model assumes full knowledge, which is key for mechanism design.

\paragraph{General seller costs.} A general cost environment is one in which the designer must pay a service cost $c(x)$ for an allocation $x$, and is defined by~\cite[Chapter 3]{hartline2013mechanism}, and considered for example by~\cite{krahmer2016optimality}, who consider some form of a repeated auction, but in which the buyer's valuation has only 2 possible values.

\paragraph{Additional related work.}
\cite{FeldmanGL15} studies combinatorial online auctions where buyers with valuations drawn from known priors arrive sequentially and purchase their utility-maximizing bundle at initially defined posted prices. Their work focuses on achieving an approximation algorithm for social welfare maximization. While~\cite{FeldmanGL15} explores simple posted price mechanisms for combinatorial goods, our work focuses on dynamic mechanisms for the rental of a single asset under stagewise individual rationality, considering a broader range of reward functions

\section{Our Model}\label{sec:model}

In this section we formally define the rental game model,
as well as stagewise auctions with seller cost (\swacs), which model the designer's interaction with a single agent. 
Throughout,
we let $\sqrrange{n} =\set{0,\ldots,n}$ and $\psqrrange{n} = \set{1,\ldots,n}$. 

\subsection{The Rental Game}\label{subsec:rental}

Fix a domain $\V$ of possible nonnegative agent valuations.
An \emph{$(n, \boldsymbol{\D}, \g)$-rental game}
is defined by a length $n \geq 0$, which is the timeframe over which the game is played;
a sequence of distributions
of agent valuations,
$\boldsymbol{\D} = \D_1,\ldots,\D_n$ supported over $\V\subseteq\preals$, one for each day; 
and a \emph{reward function} $\g$,
whose cumulative value over time the designer aims to maximize.

The game consists of \emph{days},
on each of which a new agent $i$ arrives with an independent valuation $v \sim \D_i$, interacts with the designer, and is either allocated the asset for some number of days under some payment schedule, or turned away.
The \emph{payment schedule}
is a sequence of payments
$p = (p_1,\ldots,p_x)$,
where $x$ is the length of the allocation.
We refer to the number of remaining days at some point in the game as the \emph{horizon}.

Each distribution $\D_i$ is assumed to have a probability density function (pdf) and to be supported on a single continuous interval (not necessarily the same interval for all distributions).
We also make the following assumptions on the reward function $g$:
\begin{enumerate*}
    \item It is normalized so that $g(0,0) = 0$;
    \item It is non-decreasing in its first argument, the valuation;
    \item For all valuations $v \in \V$,
    the function $g(v, \cdot) : \preals \rightarrow \preals$ is linear (in the payment).
\end{enumerate*}
The last constraint captures the fact that the designer is motivated only by the \emph{total} payment,
not the payment \emph{schedule},
so that for any two payment schedules $(p_1,\ldots,p_x)$ and $(p_1', \ldots, p_{x}')$ for the same allocation $x$,
if $\sum_{i = 1}^x p_i = \sum_{i = 1}^{x} p_i'$,
then $\sum_{i = 1}^x g(v, p_i) = \sum_{i = 1}^{x} g(v, p_i')$ for any valuation $v$.

Next we define \emph{stagewise auctions},
the mechanism by which the designer negotiates with one individual agent;
then, in 
Section~\ref{sec:rental_mech_definitions}, we 
use stagewise auctions to define \emph{rental mechanisms},
the overall mechanism by which the designer interacts with all $n$ agents over the $n$ days of the game.

\subsection{Stagewise Auctions with Seller Cost}
\label{subsection:auction}

\paragraph{Stagewise auction settings.}
A \emph{$(n,\D,\g,c)$-stagewise auction setting with seller cost} (\emph{\swac setting} for short)
features a single \emph{seller} (a.k.a. designer) and a single \emph{buyer} (a.k.a. agent).
The seller wishes to sell up to $n$ units of a single item,
\emph{over time},
one unit per day.
The buyer has some valuation $v$ per unit, drawn from a distribution $\D$
supported on $\V\subseteq \mathbb{R}_{\ge 0}$.
If $\V$ is unbounded from above, we abuse notation by writing $\sup\V=\infty$.
The \swac\ setting is publicly known, but the valuation $v$ is known only to the buyer.

The seller has a \emph{reward function} $g(\cdot,\cdot)$
mapping a valuation~$v$ and a payment $p$ to the seller's reward from selling
a single unit (on a single day) at price $p$ to an agent with valuation $v$. The reward is additive,
so that the seller's reward from selling $k$ units at prices $p_1,\ldots,p_k$ to an agent with valuation $v$ is given by $\sum_{i = 1}^k g(v, p_i)$.
However, 
the seller also
incurs a \emph{cost} $c:\sqbr{n}\rightarrow\preals$,
mapping the number of units sold to the overall production cost that the seller pays. 
The cost function is normalized so that $c(0)=0$, and it is non-decreasing in the number of units sold.
We distinguish between the \emph{gross reward} $\sum_{i = 1}^k g(v,p_i)$ of the designer from selling $k$ units, and their \emph{net reward}, which is the gross reward minus the cost of producing the $k$ units (see below).
For brevity, in the sequel,
we use the term \emph{reward} to refer to the designer's net reward.

\paragraph{Stagewise auctions.}
In its most general form, the negotiation between the seller and the buyer in a stagewise auction may feature both randomness and daily interaction: each day,
the seller negotiates with the buyer and decides whether to sell the buyer
one unit that day, and at what cost.
The next day they negotiate for the next unit, and so on,
until eventually either all $n$ units are sold or the seller decides not to sell the buyer any more units and the buyer leaves.
To simplify the exposition, we focus on a simplified and deterministic
version where the buyer submits a single bid on the first day,
and the seller then decides how many units
to sell in total (still one unit per day), and what payment to charge on each day.
We prove in Appendix~\ref{app:generelizations-and-definitions} that this simplified form is without loss of generality.

Formally, given a \swac\ setting $(n,\D,\g,c)$,
a \emph{stagewise auction with seller cost (\swac)} $\A:\B\rightarrow \left(\preals\right)^\ast$
for this setting
is a (deterministic) mapping that takes a bid $b\in \B$, where $\B$ is a bidding set defined by the seller, and outputs a
\emph{payment schedule} $\A(b) = (p_1, \ldots,p_x)$,
where $x \leq n$ is the the number of units sold to the buyer.
Each $p_i \in \preals$ is the payment charged for the $i$-th unit on day $i$.
The buyer 
is assumed (w.l.o.g.) to be deterministic,
and
ties in the buyer's utility are broken in favor of the seller.
Our goal is to design a \swac\ that maximizes in expectation the seller's \emph{net reward}, which is the total reward over all units sold, minus the cost of the allocation.

\paragraph{Notation and terminology.}

Given a stagewise auction $\A$ and an buyer $v$ that makes bid $b$, we use the following notation and terminology (omitting $\A$ where clear from the context):
\begin{itemize}
    \item The \emph{allocation}, $\x[\A]{b}$, is the number of units allocated to the buyer.
    \item The \emph{total payment} of the buyer is denoted $\p[\A]{b} = \sum_{i = 1}^{\x[\A]{b}} p_i$, where $(p_1,\ldots,p_{\x[\A]{b}})$ 
    are the payments charged by the seller.
    \item The \emph{cumulative utility} of the buyer
    from the first $i \leq \x{b}$ 
    units they are allocated is denoted
    $\utili[\A]{i}{v}{b} = i \cdot v - \insumof{j}{1}{i} p_j$.
    \item The \emph{overall utility} of the buyer is $\util{v}{b}=\utili{\x{b}}{v}{b}$.
\item The seller's net reward is denoted by
$\rew[\A]{v}{b} = \insumi{1}{\x{b}}\g(v,p_i)-c(\x{b})$.
\end{itemize}

To simplify the notation, we let $\trew[\A]{v}$ denote the seller's net reward when an buyer with valuation $v$ bids optimally (i.e., in a manner that maximizes the buyer's utility).
The seller aims to maximize $\E[\A]=\E_{v\sim\D}[\trew[\A]{v}]$. We say that a \swac is \emph{truthful} if bidding truthfully is a weakly-dominant strategy for an agent with any valuation.
Finally, for ease of exposition, we may refer to an $(n,\D,\g,c)$-stagewise auction as an $n$-auction, when the context is clear.

\paragraph{Stagewise individual rationality (\swir).}
The key feature of our model is that agents are not willing to temporarily incur negative utility at any point in time.
Formally, a \swir agent with valuation $v$
will not make bid $b \in \B$
if there exists a day $1 \leq i \leq \x[\A]{b}$
on which $\utili[\A]{i}{v}{b} < 0$.

\paragraph{Filters.}
The filter 
represents the highest cumulative average payment encountered at any day and is given by $\f{b}=\max_{\ell\leq\x{b}}
        \insumi{1}{\ell}p_i/\ell$.
A \swir\ agent will not submit a bid resulting in a negative cumulative utility at any day, hence they can only submit a bid $b$ if $\f{b}\leq v$. 
For example, if the payments induced by some bid are $(3,4,2)$, 
the cumulative average payments at each timestep are: $
\frac{3}{1}, \frac{3+4}{2} = 3.5$ and $ \frac{3+4+2}{3} = 3$.
Thus, the \emph{filter for this bid} is \( \max \{ 3, 3.5, 3 \} = 3.5 \), which occurs after the first two days. 
Consequently, an agent with valuation \( v < 3.5 \) cannot submit this bid.
In essence, the filter served as a ``barrier to entry'' for making bid $b$:
it is the smallest valuation of any \swir\ agent that can afford
to place bid $b$,
without going into negative utility at any point. As a design tool, 
where truthfulness is concerned,
they are useful in preventing
{overbidding},
not {underbidding}. Observe that a payment schedule can, w.l.o.g., consist of a high payment on the first day—thereby defining the filter—with the remaining payments spread evenly over the subsequent days.

\paragraph{Finite-menu stagewise auctions.}
Our monotonicity results for positive and negative tradeoff are restricted to
\swacs that have a finite number of options:
there is a finite number of payment schedules
that the seller offers.
W.l.o.g., such \swacs can be represented as follows:
we say that a \swac\ is \emph{finite-menu}
if it is truthful
and its allocation function has the form
$\I_1 \mapsto x_1,\ldots, \I_k \mapsto x_k$,
where
$\I_1 = [ t_1, t_2 ), \ldots, \I_k = [t_k, t_{k+1}]$ are adjacent non-overlapping intervals,
with $\inf \V = t_1 < t_2 < \ldots < t_k < t_{k+1} = \sup \V$,
and $x_i$ is the allocation for any bid in the interval $\I_i$.
The payment schedule associated with all bids in an interval $\I_i$
is the same;
we denote by $p_i$ the total payment for a bid in $\I_i$.
Although technically agents submit bids in $\V$ (the auction is truthful),
since all bids in a given interval lead to the same outcome,
it is convenient to think of a bid as the interval $\I_i$ to which it belongs,
and accordingly we sometimes use the notation $\util{v}{\I_i}$.

We say that a \finauc\ is \emph{finite-menu optimal} (or \emph{FM-optimal}, for short)
if it is optimal within the family of all \finaucs.

\subsection{Rental Mechanisms}\label{sec:rental_mech_definitions}
 A \emph{rental mechanism} for the $(n, \boldsymbol{\D},\g)$-rental game is 
 a mapping $\mech{M}$ that
 takes the current history $\history$ of the game,
 and if the asset is currently available,
 returns a \swac\ $\M(\history)$
 that is used to negotiate with the current day's agent.
 The number of units in the \swac\ is equal to the
 number of days remaining in the game,
 but the other parameters (the distribution, reward function and cost function) can be chosen arbitrarily by the designer.
   The allocation returned by the \swac\ determines the number of days the current agent rents the asset, and its payments specify the sequence of daily rental fees charged to the agent.

Let $v_1,\ldots,v_n$ be the valuations of the $n$ agents that arrive over the $n$ days,
and suppose that under the mechanism $\M$
the designer rents the item to agent $j_i$ on day $i$ and charges payment $p_i$.
If the asset is not rented out to any agent on day $i$, then $j_i = \bot$ and $p_i = 0$.
The designer's \emph{reward}
is given by 
$\trew[\mech{M}]{v_1,\dots,v_n} = \sum_{1 \leq i \leq n : j_i \neq \bot} g( v_{j_i}, p_i)$.
The designer’s objective is to maximize their expected reward,
$
\E\sqbr{\trew[\mech{M}]{v_1,\dots,v_n}},
$ 
sometimes denoted $\E[\M]$ for short.

Throughout the remainder of the paper, we purposefully use some of the \swac notation of Section~\ref{subsection:auction} in the context of rental mechanisms as well, since a rental mechanism consists of a sequence of \swacs.

\section{From Rental Games to Stagewise Auctions with Cost}\label{subsec:rental-to-swacs}

In this section we  explain how we reduce the rental mechanism design problem from an \emph{online} problem to an \emph{offline} one by showing that the designer can fix in advance a sequence: $\A_1,\ldots,\A_n$ of optimal stagewise auctions (\swacs),
where $\A_i$ is to be used to negotiate with the agent that arrives on day $i$ if the asset is available that day. We remark that this reduction holds for \emph{any} reward function, even if it does not satisfy the assumptions in Section~\ref{subsec:rental}. 
The full results for this section appear in Appendix~\ref{app:rental-to-swacs}.

The first step in our reduction is to show that w.l.o.g.,
an optimal rental mechanism is a \emph{fixed} (i.e., predetermined) sequence of \swacs:
because agent valuations are independent of one another,
an optimal rental mechanism can be \emph{history-independent},
meaning that 
although we do not know in advance whether the asset will be available to rent on a given day and what we did with the asset prior to that day,
we can decide in advance what the mechanism should do \emph{if} the asset is available that day,
regardless of the history up to that day.
Next we show that w.l.o.g., the \swac\
for the $i$-th day uses a specific seller cost function,
which captures
the designer's \emph{opportunity cost} ---
the expected loss of future reward incurred by renting the asset to the current agent.
If the designer rents out the asset for $x\geq1$ days at horizon $n$,
they forgo the expected reward from days $n-1,\dots,n-(x-1)$, since the asset is unavailable during this period. 
However,
the asset becomes available again at horizon $n - x$, 
so the designer can gain the expected reward from that point onward.
The \emph{opportunity cost} is defined as the difference between these two expected rewards:

\begin{definition}
An $(n,\boldsymbol{\D},\g)\textnormal{-over-time cost function}$ is defined as:
\begin{align*}
    c^{\boldsymbol{\D}}_{n,\g}(x)
    &=
    \begin{cases}
        R^{\boldsymbol{\D}_{\text{last }n-1},\g}_{n-1}-R^{\boldsymbol{\D}_{\text{last }n-x},\g}_{n-x},&\text{if }x\geq 1\\
        0,&\text{otherwise}
    \end{cases} %
    \qquad
    =
    \qquad
    R^{\boldsymbol{\D_{\text{last }n-1}},\g}_{n-1}-R^{\boldsymbol{\D}_{\text{last }n-\max\{x,1\}},\g}_{n-\max\{x,1\}},
\end{align*}
where $R^{\boldsymbol{\D},\g}_m$ denotes the expected reward of an optimal $(m,\boldsymbol{\D},\g)$-rental.
\label{def:cost-function}
\end{definition}
For clarity, we may omit superscripts and subscripts when the context is unambiguous.

The reduction from rental games to \swacs\ is summarized by the following theorem:

\begin{theorem}\label{theorem:rental_as_auctions}

    For all $h\in\psqrrange{n}$, let $\A_h$ be an optimal $\brackets{h,\D_h,\g,c_{h,\g}^{\boldsymbol{\D}}}$-stagewise auction with over-time cost. Then, the rental mechanism $\brackets{\A_n,\dots,\A_1}$ is optimal for the $(n,\boldsymbol{\D},\g)$-rental game.

\end{theorem}

The remainder of this technical overview describes how to find an optimal \swac\ under welfare-like, revenue-like, positive tradeoff or negative tradeoff reward.
However, our results are not restricted to the over-time cost from Definition~\ref{def:cost-function}:
they apply under any seller cost function, as long as the reward function
of the \swac\ is from one of the four classes.

\section{Useful Properties of Stagewise Auctions}\label{sec:properties}

The following propositions are used throughout,
and help establish some intuition for the structure of \swacs
and the way agents behave in them. These properties help establish key characteristics of truthful \swacs and illustrate how payments, allocations, and filtering interact.

\begin{proposition}
    \label{prop:decreasing-payments}
    In any truthful \swac,
    if two bids $w < v$
    receive the same allocation, $\x{w} = \x{v}$,
    then 
    the higher-valuation agent pays less,
    that is,
    $\p{w} \geq \p{v}$.
\end{proposition}
\begin{proof}
    Suppose $\p{w} \neq \p{v}$ (otherwise we are done).
    Agent $v$ can withstand any payment schedule that agent $w$ can (as $v > w$),
    but since the \swac is truthful,
    we know that $v$ chooses not to bid $w$.
    The allocations for the two bids are the same, so the only reason agent $v$ prefers to bid $v$ rather than $w$ is that
    they pay less,
    that is,
    $\p{w} > \p{v}$.
\end{proof}

\begin{proposition}
    \label{prop:benefit_swap}
Consider a (not necessarily truthful) stagewise auction $\A$,
let $w < v$ be valuations,
    and let $b, b'$ be two bids such that $\x{b} > \x{b'}$.
    Then the benefit that agent $w$
    derives by switching 
    from bid $b$ to bid $b'$
    exceeds that of agent $v$,
    that is,
        $\util{w}{b'} - \util{w}{b}
        >
        \util{v}{b'} - \util{v}{b}
        $.
    \end{proposition}
\begin{proof}
    For any valuation $z$,
    \begin{equation*}
        \util{z}{b'} - \util{z}{b}
        =
        \left( \x{b'} - \x{b} \right) \cdot z - \left( \p{b'} - \p{b} \right).
    \end{equation*}
This is strictly 
decreasing
in $z$, and the claim follows.
\end{proof}

Observe that the marginal utility of allocating an additional unit at some fixed price is higher for a ``stronger'' agent than for a ``weaker'' one. This implies that, in general, the stronger agent is more willing to accept a larger allocation than the weaker agent. Thus, if a weaker agent $w$ receives more units than a stronger agent $v$ in a truthful auction, it is because the weaker agent was filtered out from bidding similarly to the stronger agent. Formally, $\p{v} \geq \f{v} > w$ and $\x{v}>1$, as the filter would not be possible otherwise. This situation is formalized in the following proposition for truthful \swacs. For non-truthful \swacs the proof is very similar.

\begin{proposition}
\label{prop:un_monotonicity_means_filter}
    Let $\A$ be a deterministic and truthful \swac. For valuations $w<v$, if $\x{w}>\x{v}$, then agent $w$ would strictly benefit from bidding $v$ but is filtered out, that is, $\util{w}{v}>\util{w}{w}$ and $\p{v}\geq\f{v}>w$.
\end{proposition}
\begin{proof}
Let $\A$ be a deterministic and truthful \swac. Suppose that there are valuations $w<v$ for which $\x{w}>\x{v}$. Since $v$ chooses not to bid $w$ even though they can, it follows that $\util{v}{v}-\util{v}{w}$. It follows immediately from Prop.~\ref{prop:benefit_swap} that $\util{w}{v}>\util{w}{w}$. Since $w$ is truthful in $\A$ and is utility-maximizing, we know they are filtered out, i.e. $\f{v}>w$.
\end{proof}

\section{Monotonicity and Fixed-Rate Payments}
\label{sec:tech-overview-mono}

The two structural properties that play an important role in finding an optimal \swac
are \emph{monotonicity} in the reward and/or the allocation (depending on the reward function),%
\footnote{
As we pointed out in Section~\ref{sec:intro},
whereas all truthful classical auctions are monotone~\cite{myerson1981optimal}, for \swacs\ monotonicity is more subtle, and does not hold for any truthful \swac\ but rather only for optimal \swacs.}
and whether or not an optimal \swac 
is \emph{fixed rate} w.l.o.g.,
that is,
whether
we may assume that
the designer
offers only agreements where the agents pay the same amount on every day.
Fixed-rate payments \emph{imply} monotonicity:
it is not hard to show that Myerson's lemma applies in this case.
However, the converse is not necessarily true.

First, we give some formal definitions.
\begin{definition}
    A stagewise auction $\A$ is called \textnormal{\fr}
    if for every bid $b$,
    there exists a price $p$
    such that $\A[b] = (p,\ldots,p)$,
    meaning all units carry the same price $p$.
\end{definition}

 We consider several notions of monotonicity:

\begin{definition}[Allocation Monotonicity]\label{def:allocation-monotonicity}
    A \swac\ is:
    \begin{itemize}
        \item \textnormal{\Samono} if for all $w,v\in \V$ such that $w<v$, it holds that $\x{w}\leq \x{v}$.
        \item \textnormal{\Amono} if there is a set $\S$ with $\Pr_{\D}[\S]=0$, and for all $w,v\in\V\setminus\S$, if $w<v$ then $\x{w}\leq \x{v}$.
        \item \textnormal{\Wamono} if there is a set $\S$ with $\Pr_{\D}[\S]=0$ such that for all $w,v\in\V\setminus\S$, if $w<v$ and $\trew{w}\neq \trew{v}$ then $\x{w}\leq \x{v}$.
    \end{itemize}
\end{definition}
Observe that these definitions are hierarchical: a \samono\ \swac\ is also \amono, and an \amono\ \swac\ is also \wamono.

\begin{definition}[Reward Monotonicity]\label{def:rmono}
    A \swac\ is \textnormal{\rmono} if there is a set $\S$ with $\Pr_{\D}[\S]=0$ such that for all $w,v\in\V\setminus\S$, if $w<v$ then $\trew{w}\leq\trew{v}$.
\end{definition}

We now turn to the individual reward classes.

\subsection{Welfare-Like Reward: Monotone, and \FR W.L.O.G}\label{sec:welfare-like-mono}

The first class of reward functions we deal with is welfare-like reward, in which the reward is independent of the payment.
In this case, the designer's reward increases with the agent's valuation, so clearly a monotone non-decreasing allocation rule is optimal. A monotone allocation rule is clearly implementable by total payments equal to the unique payment rule of Myerson~\cite{myerson1981optimal} with no filters, i.e. \fr payments. Since in welfare-like reward the designer doesn't care about the payments made, this is also optimal, and it follows that \fr \swacs are w.l.o.g. for welfare-like reward.

For the sake of completeness, the remainder of this section is dedicated to formally proving this intuition.

\begin{lemma}\label{lemma:welfare-like}
    For a welfare-like reward function:
    \begin{enumerate}
        \item There exists a \fr\ and \samono\ optimal \swac.\label{item:first}
        \item All truthful optimal \swacs\ are \wamono. \label{item:second}
        \item If either the reward or the cost function are strictly increasing, all truthful \swacs\ are \amono.\label{item:third}
        \item All truthful optimal \swacs\ are \rmono.\label{item:fourth}
    \end{enumerate}
\end{lemma}

\begin{proof}
    Let $\g$ be a reward function matching the conditions of the lemma. Since $\g$ is payment-independent, in this proof we abuse notation and write $\g(v)$ instead of $\g(v,\cdot)$.
    
    Consider an agent $v\in\V$. Since the designer's reward depends only on the number of units sold to the agent, an allocation that would maximize the designer's reward is in
    \begin{align}
        \mathcal{X}_v=\argmax{x\in\sqrrange{n}}\set{x\cdot\g(v)-c(x)}.
    \end{align}
    
    Note that we don't yet know this allocation is implementable, meaning that we don't know we can achieve this allocation by a truthful auction. 

    For any valuations $w<v$ and allocations $x_w\in\mathcal{X}_w$ and $x_v\in\mathcal{X}_v$, it holds that:
    \begin{align}
        &\trew{w}=x_w\cdot\g(w)-c(x_w)\geq x_v\cdot \g(w)-c(x_v)\label{eq:2}\\
        &\trew{v}= x_v\cdot \g(v)-c(x_v)\geq x_w\cdot\g(v)-c(x_w)\label{eq:1}
    \end{align}

    Rearranging the above equations, we get that 
    \begin{align}
        (x_w-x_v)\g(w)\geq c(x_w)-c(x_v) \geq (x_w-x_v)\g(v)\label{eq:7}.
    \end{align}

    Suppose $x_w>x_v$. Since $\g$ is non-decreasing in the valuation, Eq.~(\ref{eq:7}) yields $\g(w)=\g(v)$. 
    
    Define an allocation rule $\alloc'(v)=\min\mathcal{X}_v$, where we choose the largest allocation that maximizes the designer's reward for agent $v$. If $\alloc'(w)>\alloc'(v)$ we would have $\g(w)=\g(v)$, therefore $\mathcal{X}_w=\mathcal{X}_v$ and in particular $\alloc'(w)=\alloc'(v)$, which is a contradiction. Thus $\alloc'$ is a monotone (weakly) increasing allocation rule, and due to Myerson it can be implemented as a truthful \fr \swac; 
    given an agent, we ask them to reveal their bid, then sell $\alloc'(v)$ units at the total price of the unique payment rule, such that the payments are \fr over all days. Since it is a \fr \swac, \swir has no affect here and thus Myerson guarantees that this is truthful, and this proves point~\ref{item:first} from the theorem.

    Let $\A$ be an optimal and truthful $(n,\D,\g,c)$-stagewise auction.  Define $\S=\sett{v\in\V}{\x{v}\notin\mathcal{X}_v}$. For all $v\in\S$ we know that $\trew[\A]{v}$ is less than the designer's reward from $v$ in an optimal auction, and for all $v\notin\S$ we know $\trew[\A]{v}$ is as high as possible. Thus since $\A$ is optimal we know $\Pr[\S]=0$. 
    \begin{itemize}
        \item We will prove that $\A$ is \wamono, for point~\ref{item:second} of the proof. Let $w,v\in\V\setminus\S$ such that $w<v$ and $\trew[\A]{w}\neq \trew[\A]{v}$. We need to prove $\x{w}\leq\x{v}$.

    Suppose $\x{w}>\x{v}$. In this case, as we proved above, $\g(w)=\g(v)$. Thus, since $\trew[\A]{w}\neq \trew[\A]{v}$, the inequality~(\ref{eq:2}) is strict, and subsequently one of the inequalities in eq.~(\ref{eq:7}) is also strict, which is a contradiction to $\g(w)=\g(v)$ and completes the proof of point~\ref{item:second}.
    \item 
    We now prove specific cases in which $\A$ is \amono for point~\ref{item:third}. Suppose $w,v\in\V\setminus\S$ but $\trew[\A]{w}=\trew[\A]{v}$ (the case of inequality is complete, due to point~\ref{item:second}). In this case equation~(\ref{eq:1}) yields
    \begin{align*}
        \trew[\A]{w}=\x{w}\g(w)-c(\x{w})=\trew[\A]{v}\geq\x{w}\g(v)-c(\x{w}),
    \end{align*}
    giving either $\x{w}=0$, which is enough for \amono, or $\g(w)=\g(v)$. If $\g$ is strictly increasing, we are done, so suppose $c$ is strictly increasing. In this case one of the inequalities from eq.~(\ref{eq:7}) is strict, proving that $\x{w}<\x{v}$, and completing the proof of point~\ref{item:third}.
    \item Consider $w,v\in\V\setminus\S$. From equation~(\ref{eq:1}) and due to the monotonicity of $\g$, we have 
    \begin{align*}
        \trew[\A]{v}&=\x{v}\cdot \g(v)-c(\x{v})
        \geq \x{w}\cdot\g(v)-c(\x{w})
        \\&\geq \x{w}\cdot\g(w)-c(\x{w})=\trew[\A]{w}
    \end{align*}
    and we are done.

    \end{itemize}

\end{proof}

\subsection{Revenue-Like Reward: Monotone, and \FR W.L.O.G}\label{sec:revenue-like-mono}

Recall that a revenue-like reward function
has the form $\g(v,p)=\alpha p$ for some constant $\alpha>0$.
For revenue-like reward, we show that a truthful and optimal \swac\ is \rmono and \wamono. 
A \rmono \swac
allows us to apply a ``flattening'' process to the payment schedule of all bids,
transforming any optimal \swac\ into an equivalent \swac\ with \fr payments.
In this section we write $\g(p)$ instead of $\g(\cdot,p)$, for the sake of simplicity. Additionally, when we write $\rew[\A]{\cdot}{b}$ we refer to the designer's reward from bid $b$, which does not depend on the agent's valuation.

The following lemmas actually holds for arbitrary agent distributions, not just continuous ones.

\begin{lemma}\label{lemma:rev-like-rmono}
     A truthful and optimal \swac with respect to revenue-like reward is \rmono.
\end{lemma}

\begin{lemma}\label{lemma:rev-like-wamono}
     A truthful and optimal \swac with respect to revenue-like reward is \wamono.
\end{lemma}

For the proofs of Lemmas~\ref{lemma:rev-like-rmono} and~\ref{lemma:rev-like-wamono} we first need to establish several preliminary results. 
For both proofs we would like to consider a truthful and optimal \swac $\A$ that is \emph{non-monotone} (with respect to either reward or allocation), and then based on that,
     obtain a new \swac $\A'$ with higher expected reward than $\A$, contradicting the optimality of $\A$.
     Our first step in both cases is to identify a ``non-negligible'' violation of monotonicity:
    for example, it is not enough to find two specific valuations $w < v$
    such that $\trew[\A]{w} > \trew[\A]{v}$,
    because even if we could modify the auction so that the reward from valuation $v$ strictly increases,
    this still does not change the expected revenue --- no single point does.
    It is highly nontrivial to find a concrete and actionable violation of monotonicity;
    our only assumption about $\A$ is that
    there is no valuation set of measure 0 whose removal would make $\A$ (reward- or allocation-)monotone.

The following claim shows that, given any set of valuations $\C$ in a truthful \swac, you can find a decreasing sequence of valuations all allocated the \emph{same} amount, whose payments (and thus rewards) steadily \emph{increase} and converge to the supremum of rewards attained on $\C$.
This sequence is crucial: in the next claim, we use its limit as a reference point to construct a transformation that ensures low-reward agents adjust their bids in a way that strictly increases the designer’s overall reward.

\begin{claim}\label{claim:sequence_to_supremum}
    Let $\A$ be a truthful \swac with revenue-like reward.
    
    Let $\C\subseteq\V$, and define $B=\sup\set{\trew[\A]{w}:w\in\C}$.

    There is some sequence $\set{v_i}_{i=1}^{\infty}\subseteq\C$, value $\vstar\in\textnormal{cl}\C$\footnote{\textnormal{cl}$\C$ denotes the closure of set $\C$} and $\xstar\in[n]$ for which:
    \begin{enumerate}
        \item $(v_i)$ is (weakly) decreasing\label{item:1_sequence_to_supremum}
        \item $\brackets{\p{v_i}}$ and $\brackets{\trew[\A]{v_i}}$ are (weakly) increasing\label{item:2_sequence_to_supremum}
        \item $\x{v_i}=\xstar\ \ \forall i$\label{item:3_sequence_to_supremum}
        \item $\mylim v_i=\vstar$\label{item:4_sequence_to_supremum}
        \item $\mylim \trew[\A]{v_i}=B$\label{item:5_sequence_to_supremum}
    \end{enumerate}
\end{claim}
\begin{proof}
    Consider some $\A,\C$ and $B$ as defined in the claim. If there is some $w\in\C$ for which $\trew[\A]{w}=B$, we can define $v_i=w$ for all $i$, $\vstar=w$ and $\xstar=\x{w}$; this sequence fulfills all the required properties.

    Suppose $B$ is not reached by any $w\in\C$. Since $B$ is a supremum (and not a maximum), there is a sequence of distinct elements $\set{y_i}_{i=1}^{\infty}\subseteq\C$ such that $\set{\trew[\A]{y_i}}_{i=1}^\infty$ is strictly increasing, and $\mylim \trew[\A]{y_i}=B$. We will find a subsequence of $\brackets{y_i}$ that will match the require properties.

    For each allocation $k\in\sqrrange{n}$, define $\set{y_i^k}_{i=1}^\infty\subseteq\set{y_i}_{i=1}^{\infty}$ to be the subsequence of the elements that are allocated $k$ units in $\A$ (while preserving the original order from $\brackets{y_i}$). The matching sequence $\set{\trew[\A]{y_i^k}}_{i=1}^\infty$ is a monotone subsequence of $\set{\trew[\A]{y_i}}_{i=1}^\infty$. All monotone subsequences of a converging sequence also converge, so $\mylim \trew[\A]{y_i^k}$ is defined for all $k\in\sqrrange{n}$. Since there are only $n+1$ possible allocations, these subsequences $\brackets{\trew[\A]{y_i^k}}$ cover the entire sequence $\brackets{\trew[\A]{y_i}}$, hence $B$ is
    \begin{align*}
        B=\mylim \trew[\A]{y_i}=\max\set{\mylim \trew[\A]{y_i^k}\mid k\in\sqrrange{n}}.
    \end{align*}

    Specifically, there is some allocation $\xstar$ for which $\mylim \trew[\A]{y_i^\xstar}=B$, so define the sequence $v_i=y_i^\xstar$ for all $i$, and we will show that it fulfills all the properties from the claim. By the choice of $(v_i)$, for every element $v_i$ in the sequence, $\x{v_i}=\xstar$ and $\mylim \trew[\A]{v_i}=B$. Thus we showed that $\brackets{v_i}$ fulfills properties~\ref{item:3_sequence_to_supremum} and~\ref{item:5_sequence_to_supremum}.

    Now consider some $v_i$ and $v_{i+1}$. Since $\x{v_i}=\x{v_{i+1}}=\xstar$ but $\trew[\A]{v_i}<\trew[\A]{v_{i+1}}$, we have
    \begin{align*}
        \trew[\A]{v_i}=\xstar\cdot\g(\p{v_i})-c(\xstar)<\trew[\A]{v_{i+1}}=\xstar\cdot\g(\p{v_{i+1}})-c(\xstar).
    \end{align*}
    Since $\g$ is non-decreasing, this proves $\p{v_i}<\p{v_{i+1}}$, thus property~\ref{item:2_sequence_to_supremum} is fulfilled. Additionally, due to Prop.~\ref{prop:decreasing-payments}, $v_i>v_{i+1}$, meaning $(v_i)$ is decreasing, which is property~\ref{item:1_sequence_to_supremum}. Finally, since $(v_i)$ is decreasing and bounded by 0, it converges to some valuation $\vstar\in\reals$, and since $v_i\in\C$ for all $i$, it holds that $\vstar\in\textnormal{cl}\C$, which fulfills the final property,~\ref{item:4_sequence_to_supremum}.
\end{proof}

Building on the sequence from Claim~\ref{claim:sequence_to_supremum}, we design a general \swac transformation that we use for the monotonicity proofs in this section.
This transformation is applied in different ways, but in both cases the \swac obtained differs from the original \swac in a similar manner: a set of agents with non-zero probability that violate monotonicity change their bid in a way that
strictly increases the designer's reward from them,
while other agents do not yield worse reward. Hence the transformation creates 
a \swac with higher expected reward.

\begin{claim}\label{claim:rev_create_a_tag}
     Let $\A$ be a truthful \swac with respect to revenue-like reward.
     Let $\C\subseteq\V$, and define $B=\sup \trew[\A]{\C}$. Let $\vstar$ be the limit of the sequence from Claim~\ref{claim:sequence_to_supremum} over $\C$, with allocation $\xstar$ for all agents with valuations in the sequence.

     There is a \swac $\A'$ in which for all $w\in\V$:  
     \begin{enumerate}
         \item If $\vstar\in\V$, then $\trew[\A']{\vstar}= B$.\label{item:1_rev_create_a_tag}
         \item If $w<\vstar$, then either $\trew[\A']{w}=\trew[\A]{w}$ or $\trew[\A']{w}=B$.\label{item:2_rev_create_a_tag}
         \item If $w<\vstar$ and $\x[\A]{w}> \xstar$, then $\trew[\A']{w}=B$.\label{item:3_rev_create_a_tag}
         \item If $w>\vstar$, then $\trew[\A']{w}\geq \trew[\A]{w}$.\label{item:4_rev_create_a_tag}
         \item If $w>\vstar$ and $\trew[\A]{w}<B$, then $\trew[\A']{w}>\trew[\A]{w}$.\label{item:5_rev_create_a_tag}
     \end{enumerate}
\end{claim}
\begin{proof}

Let $\A$ be a truthful \swac w.r.t. revenue-like reward $\g(v,p)=\alpha p$, let $\C\subseteq\V$, and define $B=\sup \trew[\A]{\C}$. Let $\brackets{v_i}$ be the sequence from Claim~\ref{claim:sequence_to_supremum} with its limit $\vstar$ and allocation $\x[\A]{v_i}=\xstar$. 

Define a set $S^\ast\coloneq\sett{v\in\V}{v>\vstar\textnormal{ and }\trew[\A]{v}<B}$.
The main effect we want from $\A'$ is to strictly increase the designer's reward from agents in $S^\ast$, while making sure no other agent yields worse reward.

     To cause agents in $S^*$ to change their bid and give us improved reward,
     we remove $S^*$
     from the bidding space ---
     but that is not enough,
     as agents in $S^*$ may change their bid to a valuation that yields \emph{lower} reward.
     To prevent this,
     we
     must also block off two types of bids:
    \begin{enumerate}
        \item Bids $w > \vstar$ that yield lower reward than $B$:
           \begin{equation*}
            \remrt=\sett{w\in\V}{w>\vstar\textnormal{ but }\trew[\A]{w}< B}.
        \end{equation*}
        This includes the set $S^*$.
        \item Bids $w < \vstar$ that receive a higher allocation than $\xstar$:
        \begin{equation*}
            \remlft=\sett{w\in\V}{w<\vstar\textnormal{ but }\x[\A]{w} > \xstar}.
        \end{equation*}
        These bids must be removed to prevent agents in $\remrt$ (including $S^*$) from bidding values lower than $v^*$, as we show below. 
    \end{enumerate}

   Formally, we define $\A'$ as follows:

   The bid space of $\A'$ will be $\brackets{\V\cup\set{\vstar}}\setminus\brackets{\remlft\cup\remrt}$, and:
    \begin{itemize}
        \item For bid $\vstar$, the designer of $\A'$ will sell $\xstar$ units at a total price of $\mylim\p[\A]{v_i}$, with \fr payments.
        \item For all other bids, $\A'$ will offer exactly the same allocation and payment schedule as $\A$ for the same bid.
    \end{itemize}

Note that indeed $\mylim \p[\A]{v_i}$ is defined, since the sequence $\set{\p[\A]{v_i}}_{i=1}^{\infty}$ is increasing and bounded from above: every element $\p[\A]{v_i}$ is bounded by $\xstar v_i$ due to IR, therefore the whole sequence is bounded by  $\xstar v_1$.

In Figure~\ref{figure:remRtLft} we can see an example of such a transformation from $\A$ to $\A'$, assuming $v_i=\vstar$ for all $i$, for ease of exposition.

\begin{figure}[H]
\centering
\begin{subfigure}[t]{0.49\textwidth}
    \centering
    \includegraphics[width=\textwidth]{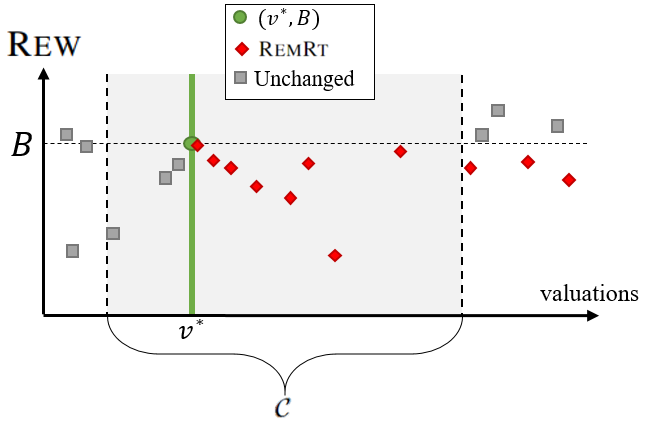}
    \caption{Part 1: Removing points in $\remrt$ from the bidding set, that is, the points $>\vstar$ with a reward less than $B=\trew[\A]{\vstar}$}\label{subfig:1-remRtLft}
\end{subfigure}
\hfill
\begin{subfigure}[t]{0.49\textwidth}
    \centering
    \includegraphics[width=\textwidth]{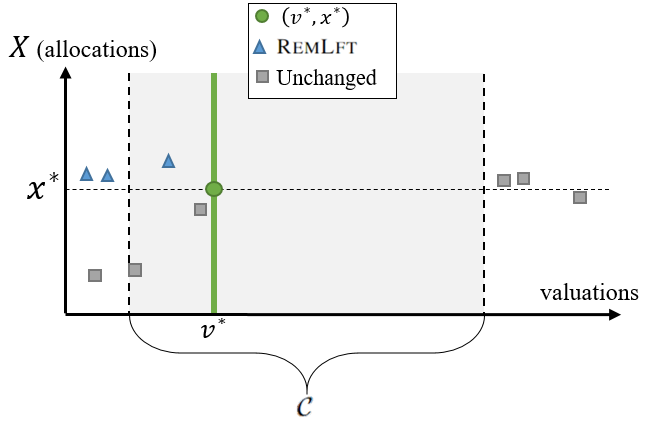}
    \caption{Part 2: Removing points in $\remlft$ from the bidding set, that is, the points $<\vstar$ with an allocation larger than $\xstar$}\label{subfig:2-remRtLft}
\end{subfigure}
\caption[An illustration of the transformation used for proofs for revenue-like reward]{An illustration of the transformation from $\A$ to $\A'$. We show for the points in $\V$, the reward the designer yields from them in $\A$ in Subfigure~\ref{subfig:1-remRtLft}, and their allocation in $\A$, in Subfigure~\ref{subfig:2-remRtLft}. On these figures we point out which valuations will be omitted from the bidding set of $\A'$.}
\label{figure:remRtLft}
\end{figure}

What is the designer's reward in $\A'$? Besides for $\vstar$, agents who don't change their bid between $\A$ and $\A'$ still yield the same designer reward, because their outcome is the same. Thus we restrict the analysis to agents who change their bid. 
There are only a few reasons an agent may change their bid: if their bid in $\A$ was removed; if their bid was changed, in the case of agent $\vstar$; or if they have a new and better option in $\A'$ --- this could only be bidding $\vstar$ as no other bid was added or changed.

Before analyzing these cases, observe that 
\begin{align*}
            \rew[\A']{\cdot}{\vstar}&=\alpha{\p[\A']{\vstar}}-c\brackets{\x[\A']{\vstar}}=\alpha\brackets{\mylim{\p[\A]{v_i}}}-c(\xstar)\\
        &=\mylim\brackets{\alpha\p[\A]{v_i}-c(\x[\A]{v_i})}=\mylim \trew[\A]{v_i}=B.\numberthis\label{eq:reward-at-sup}
\end{align*}

We now analyze each individual case, to show how the reward from each agent changes between $\A$ and $\A'$.
\begin{itemize}
    \item Agent $\vstar$ will be truthful in $\A'$ (assuming $\vstar\in\V$), as we show in each case:
    \begin{itemize}
            \item If $\x[\A']{\vstar}=\x[\A]{\vstar}$ and $\p[\A']{\vstar}=\p[\A]{\vstar}$: Since $\util[\A']{\vstar}{\vstar}=\util[\A]{\vstar}{\vstar}$, and there's no other bid with an increase in the designer's revenue from $\A$ to $\A'$, the tie-breaking does not change and agent $\vstar$ bids $\vstar$ in $\A'$.
        \item Else, let $u$ be the new bid of $\vstar$ in $\A'$, and suppose $u\neq \vstar$. By the definition of $\A'$,
        \begin{align*}
            \util[\A']{\vstar}{\vstar}&=\xhat[\A']{\vstar}\vstar-\p[\A']{\vstar}=\mylim\brackets{\x[\A]{v_i}v_i-\p[\A]{v_i}}=\mylim\util[\A]{v_i}{v_i}\\
            &\geq\mylim\util[\A]{v_i}{\vstar}=\mylim\brackets{\x[\A]{\vstar}v_i-\p[\A]{\vstar}}=\x[\A]{\vstar}\vstar-\p[\A]{\vstar}\\
            &=\util[\A]{\vstar}{\vstar}\geq\util[\A]{\vstar}{u}=\util[\A']{\vstar}{u},
        \end{align*}
        where the inequalities are due to the truthfulness of $\A$. In particular we get that $\util[\A']{\vstar}{\vstar}\geq \util[\A]{\vstar}{\vstar}$, which means that $\vstar$ can get nonnegative utility in $\A'$. Additionally this equation shows that $\util[\A']{\vstar}{\vstar}\geq \util[\A']{\vstar}{u}$;
        thus the only reason for $\vstar$ to deviate from $\vstar$ to $u$ is if $\rew[\A']{\cdot}{u}>\rew[\A']{\cdot}{\vstar}=B$. But $\rew[\A']{\cdot}{u}=\trew[\A]{u}$ and $\rew[\A']{\cdot}{\vstar}\geq \trew[\A]{\vstar}$, so in this case agent $\vstar$ would also bid $u$ in $\A$, which is a contradiction to the truthfulness of $\A$. Therefore, $\vstar$ also bids truthfully in $\A'$.
    \end{itemize}
    Thus $\trew[\A']{\vstar}=B\geq\trew[\A]{w}$, by Eq.~\ref{eq:reward-at-sup} and definition of $B$.
    \item If $w<\vstar$:
    \begin{itemize}
        \item If $w\in\remlft$, since $\x[\A]{w}>\xstar$ then by Prop.~\ref{prop:un_monotonicity_means_filter}, we have $\util[\A]{w}{w}<\utilformula[\A]{w}{v_i}$ for all $i$. Taking the limit, and since there is no filter on bid $\vstar$ in $\A'$, it holds that $\util[\A]{w}{w}<\util[\A']{w}{\vstar}$, thus $w$ will bid $\vstar$ in $\A'$; there is no other new option for $w$. Hence $\trew[\A']{w}=\rew[\A']{\cdot}{\vstar}=B$.
        \item  Else, $w\notin\remlft$ so it is possible for $w$ to bid truthfully in $\A'$, in which case $\trew[\A']{w}=\trew[\A]{w}$. If agent $w$ does not bid truthfully, it would only be to change their bid to $\vstar$, in which case $\trew[\A']{w}=B$.
    \end{itemize}
    \item If $w>\vstar$,
        \begin{itemize}
        \item If $w\notin\remrt$, they only possible non-truthful bid is to bid $\vstar$. This will not happen, since $\p[\A']{\vstar}\geq\p[\A]{v_i}$ for all $i$, and in particular for some $v_j<w$. Thus, by the truthfulness of $w$ in $\A$, we know $\util[\A']{w}{w}=\util[\A]{w}{w}\geq\util[\A]{w}{v_j}\geq\util[\A']{w}{\vstar}$.
        
        All in all, if $w\notin\remrt$, indeed $\trew[\A']{w}=\trew[\A]{w}$.
        \item If $w\in\remrt$, then $\trew[\A]{w}<B$. Since $\vstar$ would be truthful in $\A'$, then by Prop.~\ref{prop:benefit_swap} and by removal of $\remlft$ we know that if agent $w$ bids below $\vstar$ this can only be to some bid with allocation $\xstar$. But due to Prop.~\ref{prop:decreasing-payments} there is no reason for agent $w$ to prefer the lower bid over $\vstar$. Since we showed now that $w$ bids $\geq\vstar$, by design of $\A'$ we have $\trew[\A']{w}\geq B>\trew[\A]{w}$.
        
        Also note that since $w>\vstar$, agent $w$ can get nonnegative utility in $\A'$, since $\util[\A']{w}{\vstar}\geq\util[\A']{\vstar}{\vstar}\geq 0$.
    \end{itemize}

    In all the different cases, we showed that the agent has a bid in $\A'$ that yields nonnegative utility, which proves that $\A'$ is indeed IR.

\end{itemize}
\end{proof}

We are now ready to use the transformation from Claim~\ref{claim:rev_create_a_tag} to prove the monotonicity properties of optimal revenue-like \swacs.

\begin{proof}[Proof of Lemma~\ref{lemma:rev-like-rmono}: Reward Monotonicity]
Let $\A$ be a truthful and optimal \swac with respect to revenue-like reward.
 We will find a set that when removed, imposes monotonicity of the reward on the remaining valuations. We then show that if this set is of positive measure, we could create a \swac with higher expected reward, which would contradict the optimality of $\A$.
 
    For all $x\in\sqrrange{n}$ define
    \begin{align*}
        &\S_x=\sett{v\in\V}{\x[\A]{v}=x\textnormal{, }\exists w\in\V\textnormal{ s.t. } w<v\textnormal{ and }\trew[\A]{w}>\trew[\A]{v}},\textnormal{ and}\\
        &\S=\bigcup_{x=0}^n \S_x=\sett{v\in\V}{\exists w\in\V\textnormal{ s.t. }w<v\textnormal{ and }\trew[\A]{w}>\trew[\A]{v}}.
    \end{align*}

    Observe that for all $w,v\in\V\setminus\S$ such that $w<v$, necessarily $\trew[\A]{w}\leq \trew[\A]{v}$, because otherwise $v$ would be in $\S$.
    Suppose that $\A$ is not \rmono, so $\Pr[\S]>0$. Since $\S$ is a finite union of sets, thus there is some $\S_x$ with $\Pr\sqbr{\S_x}>0$. Define a set $\U\subseteq\S_x$ in the following manner:
    \begin{enumerate}
        \item If there is some $v\in\S_x$ with $\Pr[v]>0$ (in the case of a distribution with a point mass), define $\U\coloneq \set{v}$ and get $\Pr[\U]>0$. By definition of $\S_x$, there is some valuation $w<v=\min\U$ such that $\trew[\A]{w}>\trew[\A]{v}$.
        \item Else, we can choose some $u\in\S_x$ such that $\U\coloneq\sett{v\in\S_x}{v\geq u}$ has $\Pr[\U]>0$. Since $u\in\S_x$ there is a valuation $w<u$ with $\trew[\A]{w}>\trew[\A]{u}$. Thus, due to Prop.~\ref{prop:decreasing-payments} and since the reward is valuation-independent, we have that also $\trew[\A]{w}>\trew[\A]{v}$ for all $v\in\U$.
    \end{enumerate}

     Define $\C\coloneq \sett{w\in\V}{w<\min\U}$, and note that the minimum does indeed exist either way $\U$ was defined. Let $\A'$ be the \swac from Claim~\ref{claim:rev_create_a_tag} over $\C$. See Figure~\ref{fig:violation} for an illustration.
     
\begin{figure}[h]
    \centering
    \includegraphics[width=0.8\linewidth]{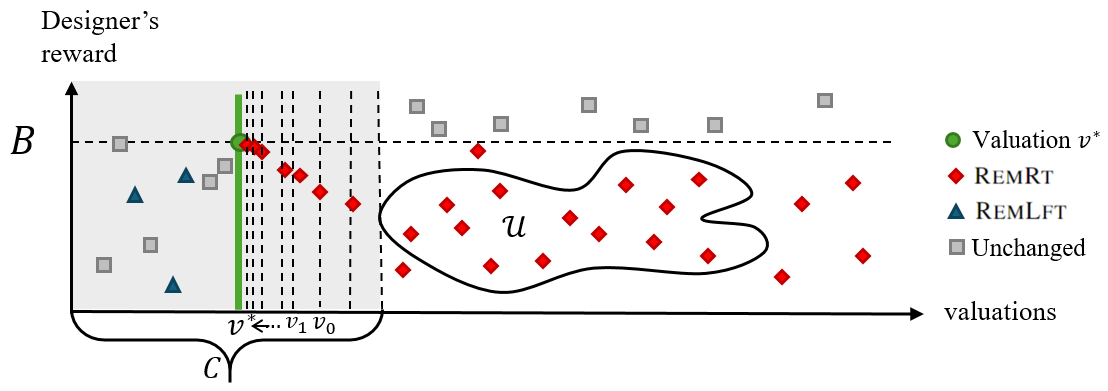}
    \caption[Illustration of the reward-monotonicity violation found in the proof of reward monotonicity for revenue-like reward]{The reward-monotonicity violation, in the form of the set $\U$ and the sequence $v_0, v_1, \ldots \rightarrow v^*$ whose
    rewards approach $B$.
    Also depicted are the sets $\remrt$
    (valuations above $v^*$ with reward below $B$)
    and $\remlft$ (valuations below $v^*$ with allocations above $x^*$; note that allocations are not shown in this figure).}
    \label{fig:violation}
\end{figure}

     We will show that:
    \begin{itemize}
        \item $\A'$ is individually rational,
        \item for all $w\in\V$, $\trew[\A']{w}\geq \trew[\A]{w}$, and
        \item for all $w\in\U$ it holds that $\trew[\A']{w}>\trew[\A]{w}$.
    \end{itemize}
    These points result in $\E[\A']>\E[\A]$, which is a contradiction to the optimality of $\A$.

    Let $\vstar$ be the limit of the sequence from Claim~\ref{claim:sequence_to_supremum}, and consider $w\in\V$. If $w\leq\vstar$, then due to Claim~\ref{claim:rev_create_a_tag}, in some cases, we have $\trew[\A']{w}=B\geq \trew[\A]{w}$, where the inequality is due to $w\in\C$. Otherwise and for all other $w\in\V$, $\trew[\A']{w}\geq \trew[\A]{w}$.  Specifically, since for all $w\in\U$ we know that there is valuation $v\in\C$ with $\trew[\A]{v}>\trew[\A]{w}$, then also $B>\trew[\A]{w}$ and thus $\trew[\A']{w}>\trew[\A]{w}$ (a strict inequality).

    Finally we show that $\A'$ has a higher expected designer's reward than $\A$:
    \begin{align*}
        \E[\A']&=\E[\underbrace{\trew[\A']{w}}_{>\trew[\A]{w}}\mid w\in\U]\cdot\underbrace{\Pr[\U]}_{>0}+\E[\underbrace{\trew[\A']{w}}_{\geq \trew[\A]{w}}\mid w\notin\U]\cdot\Pr\sqbr{\Bar{\U}}\\
        &>\E[\trew[\A]{w}\mid w\in\U]\cdot\Pr[\U]+\E[\trew[\A]{w}\mid w\notin\U]\cdot\Pr\sqbr{\Bar{\U}}
        =\E[\A']
    \end{align*}

    This is a contradiction to the optimality of $\A$, and completes the proof.
\end{proof}

Using the same transformation applied differently, we prove allocation monotonicity properties of optimal revenue-like \swacs.

\begin{proof}[Proof of Lemma~\ref{lemma:rev-like-wamono}: Weak Allocation Monotonicity]
    Let $\A$ be an optimal truthful \swac with respect to revenue-like reward. Due to  Lemma~\ref{lemma:rev-like-rmono}, $\A$ is \rmono. Let $\V_1$ be the remaining set from the definition, such that $\Pr\sqbr{\V_1}=1$ and for all valuations $w,v\in\V_1$, if $w<v$ then $\trew[\A]{w}\leq \trew[\A]{v}$. 
    
    Define the following sets for all $x\in\sqrrange{n}$:    
    \begin{align*}
        &\T_x=\sett{w\in\V_1}{\x[\A]{w}=x\textnormal{, }\exists v\in\V_1\textnormal{ s.t. } w<v\textnormal{, }x>\x[\A]{v}\textnormal{ and }\trew[\A]{w}<\trew[\A]{v}},\\
        &\T=\bigcup_{x=0}^n \T_x=\sett{w\in\V_1}{\exists v\in\V_1\textnormal{ s.t. }w<v\textnormal{, }\x[\A]{w}>\x[\A]{v}\textnormal{ and }\trew[\A]{w}<\trew[\A]{v}}.
    \end{align*}

    Observe that for all $w,v\in\V_1\setminus\T$, it holds that $\trew[\A]{w}\leq \trew[\A]{v}$, and if this inequality is strict, then $\x[\A]{w}\leq\x[\A]{v}$ (because otherwise $w$ would be in $\T$). Therefore, if $\Pr[\T]=0$ we have that $\A$ is \wamono and we are done.
    Suppose otherwise; since $\T$ is a finite union of sets, there is some $\T_x$ with $\Pr[\T_x]>0$. We define a set $\U\subseteq\T_x$ and find a valuation $v'$ in the following manner:
    \begin{enumerate}
        \item If there is some $w\in\T_x$ with $\Pr[w]>0$ (in the case of a distribution with a point mass), define $\U\coloneq \set{w}$ and get $\Pr[\U]>0$. By definition of $\T_x$, there is some valuation $v'>w=\max\U$ such that $\x[\A]{w}>\x[\A]{v'}$ and $\trew[\A]{w}<\trew[\A]{v'}$.
        \item Else, we can choose some $u\in\T_x$ such that $\U\coloneq\sett{w\in\T_x}{w\leq u}$ has $\Pr[\U]>0$. Since $u\in\T_x$ there is a valuation $v'>u=\max\U$ with $x=\x[\A]{u}>\x[\A]{v'}$ and $\trew[\A]{u}<\trew[\A]{v'}$, so in particular $\x[\A]{w}>\x[\A]{v'}$ for all $w\in\U$. Additionally, since the designer's reward is non-decreasing over valuations in $\U$, the reward function is valuation-independent, and 
        due to Prop.~\ref{prop:decreasing-payments}, the payments for all valuations in $\T_x$ are equal. Therefore also $\trew[\A]{w}<\trew[\A]{v'}$ for all $w\in\U$.
    \end{enumerate}
    
    Define $\C\coloneq \set{v'}$. Let $\A'$ be the auction from Claim~\ref{claim:rev_create_a_tag} over $\C$, with the sequence $\brackets{v_i}$, limit $\vstar$, allocation $\xstar$ and $B=\trew[\A]{v'}$. Clearly $\vstar=v'$, $v_i=v'$ for all $i$ and $\xstar=\x[\A]{v'}$. Due to Claim~\ref{claim:rev_create_a_tag}, for all $w\in\V_1$:
    \begin{itemize}
        \item If $w<v'$, then $\trew[\A']{w}=\trew[\A]{w}$ or $\trew[\A']{w}=B=\trew[\A]{v'}\geq \trew[\A]{w}$, where the inequality holds due to $w,v'\in\V_1$.
        \item For all $w\in\U$, since $w<v'$ and $\x[\A]{w}>\x[\A]{v'}=\xstar$, then $\trew[\A']{w}=B=\trew[\A]{v'}>\trew[\A]{w}$.
        \item For $w=v'$, it holds that $\trew[\A']{v'}=B=\trew[\A]{v'}$.
        \item For $w>v'$, it holds that $\trew[\A']{w}\geq \trew[\A]{w}$.
    \end{itemize}

    Putting it together, since $\U\subseteq\V_1$, we have
    \begin{align*}
         \E[\A']&=\E[\trew[\A']{w}\mid w\in\U]\cdot\Pr[\U]+\E[\trew[\A']{w}\mid w\in\V_1\setminus\U]\cdot\Pr\sqbr{\V_1\setminus\U}\\
         &\quad+\E[\trew[\A']{w}\mid w\in\V\setminus\V_1]\cdot\Pr[\V\setminus\V_1]\\
        &>\E[\trew[\A]{w}\mid w\in\U]\cdot\Pr[\U]+\E[\trew[\A]{w}\mid \V_1\setminus\U]\cdot\Pr\sqbr{\V_1\setminus\U}\\
         &\quad+\E[\trew[\A]{w}\mid w\in\V\setminus\V_1]\cdot\Pr[\V\setminus\V_1]\\
        &=\E[\A'].
    \end{align*}
    The inequality is strict because $\Pr[\U]>0$ and for $w\in\U$ the reward in $\A'$ is strictly higher than in $\A$. We also use the fact that for $w\in\V_1\setminus\U$ the reward in $\A$ does not decrease, and that $\Pr\sqbr{\V\setminus\V_1}=0$.
    We showed $\E[\A']>\E[\A]$ which is a contradiction to the optimality of $\A$.
\end{proof}

Using reward-monotonicity, we show that \fr\ \swacs\ are w.l.o.g. for revenue-like reward.

\begin{corollary}\label{corollary:rev-like-fr-wlog}
For any optimal stagewise auction $\A$ with respect to a revenue-like reward, there exists a \fr\ stagewise auction $\A'$ such that $\E\sqbr{\A'}=\E\sqbr{\A}$. 
\end{corollary}
To show this,
we prove that
we can transform any truthful optimal \swac $\A$ into a new auction $\A'$ which ``flattens'' all payment schedules of $\A$,
so it becomes a \fr\ auction:
instead of charging variable-rate payments $(p_1,\ldots,p_x)$,
we simply charge the average $\sum_{i = 1}^x p_i / x$ every day,
preserving the total payment.
The effect is that some agents may \emph{overbid} in $\A'$, because it has no filters.
However, since the reward function is valuation-independent (this is key), and due to reward-monotonicity, overbidding does not decrease the designer's reward,
and $\A'$ has the same expected reward as $\A$.
\begin{proof}[Proof of Corollary~\ref{corollary:rev-like-fr-wlog}]
    Let $\A$ be an optimal and truthful \swac with respect to revenue-like reward, where truthfulness is without loss due to Lemma~\ref{lemma:simplifications-wlog}. Recall that due to Lemma~\ref{lemma:rev-like-rmono}, there is a set $\S$ of probability measure zero, such that for all $w,v\in\S$, if $w<v$ then $\trew[\A]{w}\leq \trew[\A]{v}$. Define a new auction $\A'$ which will be like $\A$ except for the following modifications:
    \begin{enumerate}
        \item Let $\B=\brackets{\V\setminus\S}\cup\set{\psi}$ be the bidding set of $\A'$.
        \item For any bid $v\in\B$, the designer in $\A'$ will sell $\x[\A]{v}$ units at a total price of $\p[\A]{v}$ with uniform payments. Specifically, all filters from $\A$ were removed.
        \item For bid $\psi$, zero units are allocated with zero payment.
    \end{enumerate}

    Suppose there is an agent $w$ that changed his bid between $\A$ and $\A'$. Observe that the only reasons for an agent to change his bid, are either if $w\in\S$ and he was forced to change his bid, or if he could improve his utility (since for any bid in $\V$, the designer receives the same utility for that bid in $\A$ and in $\A'$, this tie-breaking rule will not cause a change in bid). Increasing utility in $\A'$ is not possible by underbidding, since it was just as possible in $\A$ and would yield the same utility in $\A'$ as in $\A$. Also bidding $\psi$ will not increase the agent's utility. On the other hand, there are now no filters on higher bids, so $w$ now overbids some $v>w$. Due to $\A$ being \rmono, we know that $\trew[\A']{w}=\rew[\A']{\cdot}{v}=\trew[\A]{v}\geq \trew[\A]{w}$.

    All in all we have $\trew[\A']{v}\geq \trew[\A]{v}$ for all $v\in\V$, and this gives us
    \begin{align*}
        \E[\A']&=\E[\trew[\A']{v}\mid v\in\S]\cdot\underbrace{\Pr[\S]}_{=0}+\E[\trew[\A']{v}\mid v\notin\S]\cdot\Pr[\Bar{\S}]\\
        &\geq\E[\trew[\A]{v}\mid v\in\S]\cdot\Pr[\S]+\E[\trew[\A]{v}\mid v\notin\S]\cdot\Pr[\Bar{\S}]\\
        &=\E[\A].
    \end{align*}

    Since any agent can bid $\psi$ to get zero utility in every timestep, $\A'$ is also IR. The fact that $\A'$ is \fr is immediate since any bid results in uniform payments. We showed that $\A'$'s expected reward is at least that of an optimal auction, completing the proof.
\end{proof}

\subsection{Positive and Negative Tradeoff Reward: Allocation Monotonicity}

In this section we deal with positive tradeoff and negative tradeoff reward functions: positive tradeoff 
is the reward function $g(v, p) = \alpha v + \beta p$ where $\alpha, \beta > 0$ are constants, 
and negative tradeoff is the reward function $g(v,p) = \alpha v- \beta p$,
where $\alpha\geq\beta > 0$.\footnote{Observe that this section does not cover welfare-like reward functions because here we restrict only to functions linear in $v$. Moreover, the results here are mostly restricted to \finaucs, while for welfare-like reward we gave more general results (see Section~\ref{sec:welfare-like-mono}).  Also note that the restriction $\alpha>0$ plays an important part in the proofs, which is why they do not cover revenue-like reward.}

\begin{proposition}\label{prop:finauc_payments_at_least_thresh}
        The payments of each interval in a \finauc\ satisfy $p_i\geq t_i$ for all intervals $\I_i$.
\end{proposition}
\begin{proof}
    Suppose the claim is false, so let $\A$ be such an auction and $\I_i$ be an interval for which this claim does not hold: there is some $v\in\I_i$ and $\p{v}<t_i$. Select some $w\in\I_{i-1}$ such that $w>\p{v}$. If $x_{i-1}>x_{i}=\x{v}$ we would reach a contradiction from Prop.~\ref{prop:un_monotonicity_means_filter}, thus $x_{i-1}<x_i$. Observe that
    \begin{align*}
        \util{w}{v}=x_i w-\p{v}\geq(x_{i-1}+1)w-\p{v}>x_{i-1} w\geq \util{w}{w},
    \end{align*}
    which is a contradiction to the truthfulness of $\A$.
\end{proof}

\begin{lemma}\label{lemma:pos-neg-tradeoff-samono}
    Any \finauc\ that is FM-optimal with respect to a positive or negative tradeoff reward
    is \samono.
\end{lemma}
\begin{proof}
Suppose for contradiction that $\A$
is a \finauc\ (and as such, truthful)
that is FM-optimal w.r.t.\ the reward function
$g(v,p) = \alpha v + \beta p$,
where  either $\alpha,\beta > 0$ or $\beta<0$ and $\alpha\geq|\beta|$.
This covers positive and negative tradeoff reward that is not welfare-\ or revenue-like (recall that for welfare-\ or revenue-like reward our results are not restricted to FM-optimality).

Let $\I_1,\ldots,\I_k$
be the menu intervals of $\A$,
and let $x_1,\ldots,x_k$
and $p_1,\dots,p_k$
be the corresponding allocations and total payments, respectively.
Since $\A$ is not \amono\ (which is equivalent to \samono\ in a \finauc), we can define 
\begin{equation*}
    x^\ast\coloneq\max\set{x\in\sqrrange{n}:x\textnormal{ violates monotonicity from the left}},
\end{equation*}
where we say that ``$x$ violates monotonicity from the left'' if there is some agent $w$ with allocation $x$ in $\A$, and another agent $v>w$ that gets a smaller allocation.

Let $\I_i$ be the \emph{first} interval with $x_i=x^\ast$. By the choice of $\I_i$ it follows that:
\begin{itemize}
    \item All lower intervals have lower allocations:
    for all $j<i$ we have $x_j<x_i$.
    \item The next interval, $I_{i+1}$, does not have a higher allocation: $x_i\geq x_{i+1}$.
\end{itemize}

We construct a new \finauc\ that achieves strictly higher expected reward than $\A$, contradicting its FM-optimality.
The transformation depends on the relationship between 
$\E\sqbr{\trew[\A]{v}\mid v\in\I_i}$,
the actual expected reward from agents in $\I_i$,
and
$\E\sqbr{\rew[\A]{v}{\I_{i+1}}\mid v\in\I_i}$, 
the expected reward from agents in $\I_i$ if they were to bid $\I_{i+1}$ in $\A$.

\paragraph{The case of $\E\sqbr{\trew[\A]{v}\mid v\in\I_i}\geq\E\sqbr{\rew[\A]{v}{\I_{i+1}}\mid v\in\I_i}$.}

We design a \swac $\A'$ that is identical to $\A$ except for the following changes:
\begin{enumerate}
    \item We remove interval $\I_{i+1}$,
    and instead extend interval $\I_i$ into $\I_i' = \I_i \cup \I_{i+1}$.
    \item Using filters, we adjust the payment schedule for bids above $\I_{i+1}$ while keeping total payments unchanged, so  that agents in $\I_{i+1}$ \emph{cannot overbid},
    but agents in intervals above $\I_{i+1}$
    do not have to change their bid. Formally, for intervals $\I_j$ where $j>i+1$ rearrange the payment schedule to be as follows:
            \begin{itemize}
                \item The payment for the first unit is $\max\set{t_j,\sfrac{p_j}{x_j}}$ (possible due to Prop.~\ref{prop:finauc_payments_at_least_thresh}).
                \item The payment for subsequent units is divided evenly and sums up, together with the first unit payment, to $p_i$.
            \end{itemize}
\end{enumerate}
The intervals below and above intervals $\I_i$ and $\I_{i+1}$ remain unchanged,
$\I_j' = \I_j$ for each $j \notin \set{i, i+1}$.
Intervals below $\I_i$ keep the same allocation and payment schedule, and intervals above $\I_{i+1}$ keep the same allocation but have a modified payment schedule as described above.

Observe that no agent
\emph{outside} $\I_{i+1}$
has an incentive to change their bid in $\A'$ compared to $\A$:
no filter was removed and no payment was reduced.
Therefore, agents with valuations outside $\I_{i+1}$ are still truthful in $\A'$.
On the other hand, agents in $\I_{i+1}$ are \emph{forced} to change their bid.
By design of $\A'$ they cannot overbid,
so they must switch to a lower interval $\I_j'$
for some $j \leq i$.
In fact, all agents in $\I_{i+1}$
switch to $\I_i'$:
by Prop.~\ref{prop:benefit_swap},
if some agent $v \in \I_{i+1}$
prefers to switch to some lower interval $\I_j'$
where $j < i$,
then 
since that interval has a lower allocation ($x_j < x_i$),
agents from $\I_i$ \emph{also}
prefer $\I_j$ to $\I_i$.
This would also be true in the original auction $\A$,
contradicting its truthfulness.
Since $\I_{i+1} \subseteq \I_i'$ and agents in $\I_{i+1}$
now bid $\I_i'$,
their bids in $\A'$ are truthful.

 Thanks to agents in $\I_{i+1}$
 switching their bid to $\I_i'$, 
 the designer's reward strictly increases in $\A'$ compared to $\A$, as we prove in either case:
\begin{itemize}
    \item  If $x_i=x_{i+1}$, then due to Prop.~\ref{prop:decreasing-payments} we have $p_i>p_{i+1}$. Together with the assumption $\E\sqbr{\trew[\A]{v}\mid v\in\I_i}\geq\E\sqbr{\rew[\A]{v}{\I_{i+1}}\mid v\in\I_i}$, which gives $\beta(p_i-p_{i+1})\geq0$, we get $\beta>0$. Therefore:
            \begin{align*}
            \E\sqbr{\trew[\A']{v}\given v\in\I_{i+1}}&=\alpha x_i\E\sqbr{v\given v\in\I_{i+1}}+\beta p_i-c(x_i)\\
            &>\alpha x_{i+1}\E\sqbr{v\given v\in\I_{i+1}}+\beta p_{i+1}-c(x_{i+1})\\
            &=\E\sqbr{\trew[\A]{v}\mid v\in\I_{i+1}}
            \end{align*}
    \item Else, $x_i>x_{i+1}$. In this case, 
 \begin{align*}
    \E\sqbr{\A'}-\E\sqbr{\A}&=\E\sqbr{\rew[\A']{v}{\I_i'}\given v\in\I_{i+1}}-\E\sqbr{\rew[\A]{v}{\I_{i+1}}\given v\in\I_{i+1}}\\
    &=\alpha(x_i-x_{i+1})\E\sqbr{v\given v\in\I_{i+1}}+\beta(p_i-p_{i+1})-c(x_i)+c(x_{i+1})\\
    &>\alpha(x_i-x_{i+1})\E\sqbr{v\given v\in\I_{i}}+\beta(p_i-p_{i+1})-c(x_i)+c(x_{i+1})\\
    &=\E\sqbr{\trew[\A]{v}\mid v\in\I_i}-\E\sqbr{\rew[\A]{v}{\I_{i+1}}\mid v\in\I_i}\geq0.
\end{align*}
The first inequality is due to
the fact that $x_i>x_{i+1}$ (allocation-monotonicity is violated)
and $\I_{i+1}$ contains only valuations greater than $\I_i$;
the second is our assumption in this case.
\end{itemize}

This shows that $\A'$ is a \finauc that achieves higher expected reward than $\A$,
contradicting the FM-optimality of $\A$.

\paragraph{The case of $\E\sqbr{\trew[\A]{v}\mid v\in\I_i}<\E\sqbr{\rew[\A]{v}{\I_{i+1}}\mid v\in\I_i}$.}
Design a new \finauc\ $\A'$ that differs from $\A$ in the following manner:
        \begin{itemize}
            \item Bid $\I_i$ is removed, and instead we extend interval $\I_{i+1}$ into $\I_{i+1}'=\I_i\cup\I_{i+1}$.
            \item The filter on bid $\I_{i+1}$ is lowered to $t_i$ (possible due to Prop.~\ref{prop:un_monotonicity_means_filter}).
        \end{itemize}
        Clearly agents not in $\I_i$ do not change their bid, and are truthful in $\A'$. Agents in $\I_i$ will now bid $\I_{i+1}'$ since the filter was lowered and due to Prop.~\ref{prop:un_monotonicity_means_filter}. Therefore 
        \begin{align*}
            \E\sqbr{\trew[\A']{v}\given v\in\I_i}=\E\sqbr{\rew[\A]{v}{\I_{i+1}}\mid v\in\I_i}>\E\sqbr{\trew[\A]{v}\mid v\in\I_i}.
        \end{align*}

    We showed that also in this case, $\E\sqbr{\A'}>\E\sqbr{\A}$, which is a contradiction to the FM-optimality of $\A$.
\end{proof}

We remark that while Lemma~\ref{lemma:pos-neg-tradeoff-samono} establishes allocation-monotonicity of \emph{FM-optimal} \swacs,  for consumer surplus
we can show allocation monotonicity for \emph{any} optimal \swac:
\begin{lemma}\label{lemma:consumer-surplus-amono}
For consumer surplus, %
    all truthful and optimal \swacs\ are \amono.
\end{lemma}

The proof uses a crucial property specific to consumer surplus: if an agent 
switches to a bid that results in a lower allocation
because it increases their utility, 
then the designer's \emph{reward} from the agent also increases.
\begin{proof}[Proof of Lemma~\ref{lemma:consumer-surplus-amono}]
    Let $\g(v,p)=v- p$ be the reward function consumer surplus, and
    let $\A$ be an optimal \swac with respect to $\g$ and some cost function $c$.
    Suppose, for contradiction, that $\A$ is not \amono.
    We use the following terminology an notation in this proof: If there are valuations $w<v$ such that $\x{w}>\x{v}$ we say that $w$ violates monotonicity \emph{from the left} (with respect to $v$) and $v$ violates monotonicity \emph{from the right} (with respect to $w$). For each $v\in\V$ we define the set $\S_v\coloneq\sett{u\in\V}{u<v\textnormal{ and }\x{u}>\x{v}}$  to be all the valuations that violate monotonicity from the left with respect to $v$.
    
    The idea behind the proof is to construct a new auction $\A'$ based on $\A$. We find a valuation $\vstar$ that violates monotonicity from the right with respect to a set $\mathcal{U}$ that has positive probability. We use Prop.~\ref{prop:un_monotonicity_means_filter} to show that if we reduce the filter on bid $\vstar$, it results in all the bidders of $\mathcal{U}$ changing their bid to $\vstar$ due to an increase in utility. Our reward function increases with the agents' utility, and the new allocation for agents in $\mathcal{U}$ is lower than their original allocation (meaning it has a lower cost) the designer's reward from all agents in $\mathcal{U}$ strictly increases. If we show that no other agent changes their bid, we reach a contradiction to the optimality of $\A$.
    
    We now give the details and prove the above formally. 
    For each allocation $x\in\sqrrange{n}$ we define the set of valuations $\S_x$ as the set of valuations that violate allocation-monotonicity from the left with respect to some valuation with allocation $x$. Formally, 
    $$\S_x\coloneq\sett{v\in\V}{\exists u\in\V\textnormal{ s.t. }v<u\textnormal{ and }\x{v}>\x{u}=x}.$$

    If for each $x$ the set $\S_x$ would be of probability zero, then the set 
    $${\S\coloneq\bigcup_{x\in\sqrrange{n}}\S_x}=\sett{v\in\V}{\exists u\in\V\textnormal{ s.t. }v<u\textnormal{ and }\x{v}>\x{u}}$$
    would also be of probability zero, which would be a contradiction to the fact that $\A$ is \emph{not} \amono. Thus we can define 
    \begin{align*}
        &\xstar\coloneq\min\sett{x\in\sqrrange{n}}{\Pr[\S_x]>0}.
    \end{align*}
    Since we are dealing with a continuous distribution, there is some $\vstar$ with $\x{\vstar}=\xstar$ such that ${\Pr[\S_{\vstar}]>0}$.
    Define
    \begin{align*}
        &w\coloneq\sup\sett{v\in\V}{v<\vstar\textnormal{ and }\x{v}<\xstar}.
    \end{align*}
    
    Define a new auction $\A'$ that is based on $\A$, but for bid $\vstar$ the filter is reduced to $\max\set{w,\frac{\p{\vstar}}{\x{\vstar}}}$. Note that this change does not affect the agents with valuations below $w$, since the filter is not lowered enough. Additionally it does not affect valuations from $\vstar$ or higher, since the total payment does not change. Thus the only agents who may bid non-truthfully in $\A'$ are those in $[w,\vstar)$. Let $v\in[w,\vstar)$:
    \begin{enumerate}
        \item If $\x[\A]{v}>\xstar=\x[\A]{\vstar}$ due to Prop.~\ref{prop:un_monotonicity_means_filter} we know that agent $v$ strictly increases their utility by bidding $\vstar$, and this is possible in $\A'$ since the filter is lowered enough, thus $v$ will bid $\vstar$ in $\A'$. Additionally since $\x[\A]{v}>\xstar=\xhat[\A']{v}$ we have $c\brackets{\xhat[\A']{v}}\leq c\brackets{\xhat[\A]{v}}$, thus:
        \begin{align*}
            \trew[\A']{v}-\trew[\A]{v}&=\util[\A']{v}{\vstar}-\util[\A]{v}{v}+c\brackets{\xhat[\A]{v}}-c\brackets{\xhat[\A']{v}}>0.
        \end{align*}
        This shows that the designer's reward from $v$ is strictly larger in $\A'$ than in $\A$.
        \item If $\x[\A]{v}=\xstar$ and $v$ bids $\vstar$ in $\A'$, this also strictly increases the designer's reward as in the above case, since the agent's utility strictly increases and the cost does not change.
        \item Else, $\x[\A]{v}<\xstar$. But this is only possible for $v=w$ by definition of $w$. Since the distribution is continuous, $\Pr[w]=0$, so even if bidder $w$ yields less reward in $\A'$ than in $\A$, this does not affect the expected utility of $\A'$.
    \end{enumerate}
    
    We showed that all agents in $\mathcal{U}\coloneq\S_{\vstar}\cap[w,\vstar)$ yield strictly larger utility for the designer in $\A'$ relative to $\A$, and there is no other change in terms of expected reward for the designer. We want to show that $\Pr[\mathcal{U}]>0$ and we do this by contradiction. Suppose $\Pr[\mathcal{U}]=0$. In this case $\Pr[\S_{\vstar}\cap[0,w)]>0$. Once again due to the continuity of the distribution, we can find a valuation $u\leq w$ with $\x[\A]{u}<\xstar$ and  $\Pr[\S_{\vstar}\cap[0,u)]>0$.
    By definition of $\S_u$, $S_{\vstar}$ and since $\x[\A]{u}<\x[\A]{\vstar}$  it holds that $\S_{\vstar}\cap[0,u)\subseteq\S_u$, so also $\Pr[\S_u]>0$. This is a contradiction to the choice of $\xstar$, so we deduce that $\Pr[\mathcal{U}]>0$.

    Putting it all together, we show that the expected reward of $\A'$ is larger than that of $\A$:
    \begin{align*}
        \E\sqbr{\A'}&=\E\sqbr{\trew[\A']{v}\mid v\notin\mathcal{U}}\cdot\Pr[v\notin\mathcal{U}]+\E\sqbr{\rew[\A']{v}{\vstar}\mid v\in\mathcal{U}}\cdot\Pr[v\in\mathcal{U}]\\
        &>\E\sqbr{\trew[\A]{v}\mid v\notin\mathcal{U}}\cdot\Pr[v\notin\mathcal{U}]+\E\sqbr{\trew[\A]{v}\mid v\in\mathcal{U}}\cdot\Pr[v\in\mathcal{U}]
        =\E\sqbr{\A}
    \end{align*}

    This result is in contradiction to the optimality of $\A$, completing the proof.
\end{proof}

\begin{corollary}\label{corollary:consumer-surplus-samono-wlog}
    There exists an optimal stagewise auction \(\A'\) with respect to consumer surplus that is \emph{\samono}.
\end{corollary}
\begin{proof}
    Let $\A$ be a truthful optimal \swac with respect to consumer surplus matching the conditions of the Lemma. Due to Lemma~\ref{lemma:consumer-surplus-amono} we know that $\A$ is \amono, thus there exists a set $\S$ with $\Pr[\S]=0$ such that for all $u,v\in\V\setminus\S$ it holds that $\x[\A]{u}\leq\x[\A]{v}$. Define an auction $\A'$ with a bidding set $\B\coloneq\brackets{\V\setminus\S}\cup\set{\emptyset}$. For bids in $\V\setminus\S$ agents get in $\A'$ the exact same conditions as in $\A$ (allocation and payment schedule), while an agent bidding $\emptyset$ gets nothing and pays nothing.

    Clearly all agents in $\V\setminus\S$ will still bid truthfully in $\A'$, and the change in reward from agents in $\S$ has no affect on the total expected reward of $\A'$ when compared to $\A$, due to the fact that $\Pr[\S]=0$. Additionally $\A'$ is \swir\ since agents have the option of bidding $\emptyset$ for a utility of zero.

    Finally, it remains to show that $\A'$ is \samono. Let $u,v\in\V$ and suppose that $\xhat[\A']{u}>\xhat[\A']{v}$. Due to Prop.~\ref{prop:un_monotonicity_means_filter}, neither agent bids $\emptyset$ in $\A'$. Thus by the design of $\A'$, agent $w$ bids a value higher than that of $v$'s bid in $\A'$, but this means that agent $w$ \emph{is} able to make the same bid as $w$ in $\A'$, in contradiction to Prop.~\ref{prop:un_monotonicity_means_filter}. Therefore $\A'$ is \samono\ and this completes the proof.
\end{proof}

\subsection{Positive Tradeoff: Reward Monotonicity}
\begin{lemma}\label{lemma:pos-tradeoff-rmono}
    Let $\g(v,p)=\alpha v+\beta p$ be a positive tradeoff reward function ($\alpha,\beta>0$). An FM-optimal \swac, w.r.t. to $\g$ and any cost function $c$, is \rmono.
\end{lemma}
\begin{proof}
    Let $\A$ be an FM-optimal \swac w.r.t. a positive tradeoff reward $\g(v,p)=\alpha v+\beta p$ and cost function $c$. Denote its intervals in increasing order by $\I_1,\dots,\I_k$ and their respective allocations with $x_1\leq\dots\leq x_k$. Due to Lemma~\ref{lemma:pos-neg-tradeoff-samono}, $\A$ is \samono. Suppose, for contradiction, that $\A$ is not \rmono. In this case there is some interval $\I_i$ such that
    \begin{equation}
        \E\sqbr{\trew[\A]{v}\given v\in\I_i}>\E\sqbr{\trew[\A]{v}\given v\in\I_i}.
    \end{equation}
    Denoting $v_j=\E\sqbr{v\in\I_j}$ for all $j\in[k]$ and expanding on the definition of reward gives us
    \begin{equation}\label{eq:pos-tradeoff-rmono}
        \alpha x_iv_i+\beta p_i-c(x_i)>\alpha x_{i+1} v_{i+1}+\beta p_{i+1}-c(x_{i+1}).
    \end{equation}
    Design a new \swac $\A'$ that is based on $\A$ but for these differences:
    \begin{enumerate}
        \item We remove interval $\I_{i+1}$,
        and instead extend interval $\I_i$ into $\I_i' = \I_i \cup \I_{i+1}$.
        \item We adjust the payment schedule for bids above $\I_{i+1}$ while keeping total payments unchanged, so that agents in $\I_{i+1}$ \emph{cannot} overbid (this is feasible due to Prop.~\ref{prop:finauc_payments_at_least_thresh}).
    \end{enumerate}
    The intervals below and above intervals $\I_i$ and $\I_{i+1}$ remain unchanged,
    $\I_j' = \I_j$ for each $j \notin \set{i, i+1}$.
    Intervals below $\I_i$ keep the same allocation and payment schedule, and intervals above $\I_{i+1}$ keep the same allocation but have a possibly modified payment schedule.
    Observe that no agent
\emph{outside} $\I_{i+1}$
has an incentive to change their bid in $\A'$ compared to $\A$:
no filter was removed and no payment was reduced.
Therefore, agents with valuations outside $\I_{i+1}$ are truthful in $\A'$.
On the other hand, agents in $\I_{i+1}$ are \emph{forced} to change their bid.
By design of $\A'$ they cannot overbid,
so they must switch to a lower interval $\I_j'$
for some $j \leq i$.
In fact, we show that agents in $\I_{i+1}$
must switch to $\I_i'$:
suppose they switch to some lower interval $\I_j' = \I_j$
where $j < i$.
By choice of $\I_i$ we have $x_j \leq x_i$. 

If $x_j=x_i$ then by Prop.~\ref{prop:decreasing-payments} agents in $\I_{i+1}$ would prefer to bid $\I_i'$ in $\A'$. 

Else, $x_j<x_i$.
However,
if 
agents from $\I_{i+1}$ prefer to bid $\I_j' = \I_j$
over $\I_i'$,
this means agents from $\I_i$ \emph{also}
prefer $\I_j'$ to $\I_i'$,
as their utility would be improved due to Prop.~\ref{prop:benefit_swap}.
Consequently, since $\I_i$ and $\I_j$ are equivalent to $\I_i'$ to $\I_j'$,
we obtain a contradiction to the truthfulness of agents in $\I_i$ in the original auction $\A$.
 Thanks to agents in $\I_{i+1}$
 switching their bid, 
 the designer's reward strictly increases in $\A'$ compared to $\A$:
\begin{align*}
    \E[\A']-\E[\A]&=\E\sqbr{\trew[\A']{v}\given v\in\I_{i+1}}-\E\sqbr{\trew[\A]{v}\given v\in\I_{i+1}}\\
    &=\alpha v_{i+1}x_i+\beta p_i-c(x_i)-\alpha v_{i+1}x_{i+1}-\beta p_{i+1}+c(x_{i+1})\\
    &>\alpha v_{i+1}x_i+\beta p_{i+1}-c(x_{i+1})-\alpha x_iv_i-\beta p_{i+1}+c(x_{i+1})\\
    &=\alpha x_i (v_{i+1}-v_i)\geq 0.
\end{align*}
The first inequality is due to our assumption, as seen in Eq.~\ref{eq:pos-tradeoff-rmono}.
Since $\A'$ is a \finauc we reach a contradiction to the FM-optimality of $\A$.
\end{proof}

\subsection{Positive Tradeoff: Fixed-Rate W.L.O.G} 
In order to show that \fr\ \swacs\ are w.l.o.g. for positive tradeoff reward, we actually show that although Myerson's payment rule is not the only payment that guarantees truthfulness,
in truthful positive tradeoff \swacs, it is in fact FM-optimal among \finaucs.

\begin{lemma}\label{lemma:positive-tradeoff-fr-wlog}
    Any truthful and FM-optimal \swac w.r.t. a positive tradeoff reward must follow the unique payment rule of Myerson: for every interval $\I_i$ with allocation $x_i$ and threshold $t_i$, the payment must be
    \begin{equation}
        p_i=p_{i-1}+(x_i-x_{i-1})t_i,
    \end{equation}
    where $x_0=p_0=0$.
    Moreover, \fr\ pricing is w.l.o.g.
\end{lemma}

For intuition, consider first why Myerson's classic proof for the unique payment rule of truthful auctions does \emph{not} apply in our setting, even when 
we restrict to monotone allocation rules.
In order to ensure truthfulness with a monotone allocation rule, the designer must prevent overbidding and underbidding. Without \swir (i.e., under \oir), the only way to
achieve this is through the total payment: it must be set precisely so as to ensure truthfulness, and there is only one way to do so, which is exactly the unique payment rule.
However, under \swir, the designer has an additional tool, \emph{filters}, which can prevent overbidding by restricting lower-valuation agents from bidding on high allocations.
This means that high total payments are not always necessary to prevent overbidding,
and Myerson's payment rule is no longer unique.

Despite this additional flexibility, when maximizing a positive tradeoff reward,
filters are not \emph{necessary} to prevent overbidding:
the designer can always prevent low-valuation agents from bidding like high-valuation agents by increasing the total payments,
and this only increases the designer's reward.
The key limitation on the total payment is that it must remain small enough to prevent \emph{underbidding}.
Crucially, we prove that the maximum total payment that still ensures truthfulness in this setting aligns exactly with the unique payment rule of Myerson. Thus, no matter the allocation, an \emph{FM-optimal truthful} (and not just truthful) \finauc\ must charge a total payment equal to Myerson’s payment rule.
Since filters are not required to prevent overbidding under positive tradeoff reward, and are \emph{never} useful to prevent underbidding,
the FM-optimal \swac for positive tradeoff can use \fr payments w.l.o.g.;
it uses Myerson's payment rule, and the payment is equally split over the allocation.

We now give the formal proof which follows the above reasoning.

\begin{proof}[Proof of Lemma~\ref{lemma:positive-tradeoff-fr-wlog}]
    Let $\g$ and $c$ be function per the claim, and let $\A$ be an FM-optimal \swac. Due to Lemma~\ref{lemma:pos-neg-tradeoff-samono}, $\A$ is \samono. We denote the intervals of $\A$ by $\I_1,\dots,\I_k$ with allocations $x_1,\dots,x_k$, thresholds $t_1,\dots,t_k$ and total payments $p_1,\dots,p_k$.

    Consider an interval $\I_i$ in $\A$ and the following interval $\I_{i+1}$. If $p_{i+1}>p_i+(x_{i+1}-x_i)t_{i+1}$, for some $v\in\I_{i+1}$ we have
    \begin{align*}
        \util{v}{\I_i}=x_iv-p_i>x_iv-p_{i+1}+x_{i+1}t_{i+1}-x_it_{i+1}=\util{v}{\I_{i+1}}+x_i(v-t_{i+1})\geq \util{v}{\I_{i+1}},
    \end{align*}
    which is a contradiction to the truthfulness of $\A$, giving us
    \begin{align}\label{eq:positive-tradeoff-max-pay}
        p_{i+1}\leq p_i+(x_{i+1}-x_i)t_{i+1}.
    \end{align}

    We design a stagewise auction $\A'$ over the same intervals as $\A$, and set the payments recursively as ${p_{i+1}'=p_i'+(x_{i+1}-x_i)t_{i+1}}$, where we define $p_0'=0$ and $x_0=0$ ($p_i'$ is the payment in $\A'$ for interval $\I_i$). The payment schedule is \fr. Observe that this is exactly the unique payment rule of Myerson~\cite{myerson1981optimal} that guarantees optimality for a monotone allocation rule. Since in a \fr \swac there are no filters, Myerson's proof applies as is, and $\A'$ is truthful and a \finauc.

    Given Eq.~\ref{eq:positive-tradeoff-max-pay}, and due to the optimality of $\A$, we have that $p_{i+1}= p_i+(x_{i+1}-x_i)t_{i+1}$ exactly.  
\end{proof}

\subsection{Negative Tradeoff: Requires Threshold Payments} Although optimal negative tradeoff \swacs are \amono, in contrast to the previous reward functions they require variable payment schedules. In fact, variable payments are significantly better than \fr payments:
consider an $n$-unit \swac,
where valuations are drawn i.i.d\ uniformly from \([0,1]\),
with consumer surplus reward and the over-time cost function of Def.~\ref{def:cost-function}.
A \fr \swac can achieve an expected reward of \(1/2\):
this is attained by selling a single unit to the agent,
for free, yielding a reward of $\E_v[g(v, 0)-c(1)] =1/2$.
We show in Lemma~\ref{lemma:best-fr-consumer-surplus}
that \emph{no fixed-rate mechanism can do better.}  
In contrast we show in Lemma~\ref{lem:best-consumer-surplus-uni} that there is a \emph{non-fixed-rate} \swac \(\A\) that has
$\lim_{n\rightarrow\infty}{\E[\A]}=1$.
Due to the over-time cost function
and the fact that $g(v,p) = v - p \leq 1$,
this is the best net reward possible.

We prove that although \fr payments are not optimal, an optimal \finauc for negative threshold can still have a simple structure, which we call \emph{threshold payments}:
for each interval $\I_i = [t_i, t_{i+1})$
with corresponding allocation $x_i$,
the payment is $t_i$ on the first day, and $0$ on the $x_i - 1$ remaining days. We say that a \swac is a \emph{threshold auction} if the allocation rule is monotone and it requires threshold payments. If the allocation of a threshold auction is either $1$ or $n$, we call it a \emph{one-or-all} threshold auction.

\begin{lemma}\label{lemma:neg-tradeoff-threshold-auction-truthful-and-optimal}
    Let $\alloc:\V\rightarrow\sqrrange{n}$ be a monotone non-decreasing step-function, and let $\A$ be the matching threshold auction. It holds that:
    \begin{enumerate}
        \item $\A$ is truthful.
        \item If the reward is a negative tradeoff reward, then 
        given another truthful auction $\A'$ with allocation rule $\alloc$, the expected rewards satisfy $\E\sqbr{\trew[\A]{v}}\geq\E\sqbr{\trew[\A']{v}}$.
    \end{enumerate}
\end{lemma}
\begin{proof}
    Consider some $\alloc,\A$ and $\g(v,p)=\alpha v-\beta p$ matching the conditions of the claim ($\alpha\geq\beta>0$).
    
   Let $v\in\V$ and suppose, to contradict truthfulness, that it is not optimal for $v$ to bid truthfully, i.e. there is a bid $b$ with  $\util{v}{b}>\util{v}{v}$. Due to the payment schedule of $\A$ we know that $v\geq\filter(b)=\p{b}$. Additionally, since bids with the same allocation also have the same payment, it follows that $\x{b}+1\leq\x{v}$ and $\p{b}<\p{v}$. Hence
   \begin{align*}
       \x{b}v-\p{v}=\util[\A]{v}{b}>\util[\A]{v}{v}=\x{v}v-\p{v}\geq\x{b}v+v-\p{v}\geq\x{b}v,
   \end{align*}
   which is a contradiction to $\p{v}\geq0$. This completes the proof of truthfulness of $\A$.

   Next, suppose there a truthful auction $\A'$ with allocation rule $\alloc$ and $\E\sqbr{\trew[\A]{v}}<\E\sqbr{\trew[\A']{v}}$. This means that there is at least one agent $v$ with $\trew[\A]{v}<\trew[\A']{v}$, because otherwise we'd have
   \begin{align*}
       0&\leq\E\sqbr{\left|\trew[\A]{v}-\trew[\A']{v}\right|}=\E\sqbr{\trew[\A]{v}-\trew[\A']{v}}\\
       &=\E\sqbr{\trew[\A]{v}}-\E\sqbr{\trew[\A']{v}}<0,
   \end{align*}
   which is a contradiction. Therefore for agent $v$ we have:
   \begin{align*}
       &\alpha\x[\A]{v}v-\beta\p[\A]{v}-c\brackets{\x[\A]{v}}=\trew[\A]{v}<\trew[\A']{v}\\
       &\quad=\alpha\x[\A']{v}v-\beta\p[\A']{v}-c\brackets{\x[\A']{v}}
       =\alpha\x[\A]{v}v-\beta\p[\A']{v}-c\brackets{\x[\A]{v}}
   \end{align*}
   It follows that $\p[\A']{v}<\p[\A]{v}$. We will show that this contradicts the truthfulness of $\A'$. Consider agent ${w=\frac{\p[\A']{v}+\p[\A]{v}}{2}>\p[\A']{v}}$. Since $w<\p[\A]{v}=\filter[\A](v)$ we deduce $\x[\A']{w}=\x[\A]{w}\leq\x[\A]{v}-1=\x[\A']{v}-1$. But observe that
   \begin{align*}
       \util[\A']{w}{w}&=\x[\A']{w}w-\p[\A']{w}\leq\x[\A']{v}w-\p[\A']{w}-w<\x[\A']{v}w-\p[\A']{w}-\p[\A']{v}\\
        &=\util[\A']{w}{v}-\p[\A']{w} \leq\util[\A']{w}{v}
   \end{align*}

   which contradicts the truthfulness of $\A'$. We conclude that $\E\sqbr{\trew[\A]{v}}\geq\E\sqbr{\trew[\A']{v}}$.
\end{proof}

Together with Lemma~\ref{lemma:pos-neg-tradeoff-samono} and Corollary~\ref{corollary:consumer-surplus-samono-wlog}, we get the following result:

\begin{corollary}
\label{cor:neg-tradeoff-threshold-payment}
Under negative tradeoff reward, threshold payments are FM-optimal.
Moreover, if the reward is consumer surplus, they are globally optimal.
\end{corollary}

Having shown that the FM-optimal (or, in some cases, the globally optimal) \swac for negative threshold uses threshold payments,
it remains to find the intervals
and the corresponding allocations that make up the auction (under threshold payments, the payment schedule is determined given the intervals).
We explain how this is done in the next section.

\section{Finding Optimal Rental Mechanisms}\label{sec:tech-overview-rental-mechs}

Recall that an optimal rental mechanism consists, w.l.o.g., of a sequence of optimal \swacs $\A_1,\ldots,\A_n$,
where $\A_i$ is responsible for allocating the item on day $i$:
it has a horizon of $n - i+1$ days,
the same reward function as the rental game,
and the over-time cost function from Definition~\ref{def:cost-function}.

In this section we derive two optimal \swacs: one for reward functions where \fr\ payments are w.l.o.g. (welfare-like, revenue-like and positive tradeoff),
and another for reward functions that require threshold payments (negative tradeoff). 

A unifying theme for both types of \swacs is that they optimize a generalized notion of \emph{ironed virtual values}.
Myerson~\cite{myerson1981optimal} defines a virtual value $\varphi(v)=v-(1-F(v))/f(v)$ for revenue, where $F$ and $f$ are the cdf and pdf (resp.) of the agent's distribution.
This captures the optimal per-unit revenue that a seller can obtain
in a standard auction,
in the sense that
$\E[\varphi(v)\x{v}]=\E[\p{v}]=E[\mathit{revenue}]$.
Virtual values, together with Myerson's Lemma, essentially reduce the mechanism design problem --- i.e., the problem of finding an optimal global allocation and payment rule that result in truthfulness --- to a pointwise single-parameter optimization problem,
that of finding an optimal allocation for each possible valuation in isolation. The formal ironing definition is presented in Section~\ref{sec:ironing}.

Using the following Lemma, and our generalized notion of virtual values, we are able to reduce  the \swac design problem to a pointwise single-parameter optimization problem: 
finding a monotone allocation rule $\alloc$ that maximizes the expression ${\Bar{\theta}(v)\,\x{v}-c\bigl(\x{v}\bigr)}$ 
for each $v\in\V$.
We then show,
on a per-case basis,
that given such an allocation rule, we can design a payment schedule for each bid that implements this allocation rule in a truthful and optimal way.

\begin{lemma}\label{lemma:irn-maximizer-maximizes-expected}
    Let $\theta:\V\rightarrow\reals$ be a function, and let $\Bar{\theta}$ be its ironing. A monotone non-decreasing allocation rule $\alloc:\V\rightarrow\sqbr{n}$ that maximizes $\Bar{\theta}(v)\x{v}-c\brackets{\x{v}}$ pointwise for all $v\in\V$, it holds that the value $\E\sqbr{\theta(v)\x{v}-c\brackets{\x{v}}}$ is maximized among all monotone non-decreasing allocation rules.
\end{lemma}

The proof is deferred to Section~\ref{sec:proof-optimal-rental-mechs}, as it relies strongly on the ironing procedure definitions and is rather straightforward.

\subsection{Optimal \FR Mechanism}\label{sec:fr-mech}

In this section we present a \fr rental mechanism that is optimal among all \fr mechanisms for any reward function. To be precise,
it is optimal for welfare- and revenue-like reward and FM-optimal for positive tradeoff.

The \swac  we design optimizes a \emph{generalized} notion of virtual values~\cite{HartlineR08} 
which achieves the same goal, specific to a given family of \swacs,
in the following sense:
\begin{definition}[Fixed-rate-optimal virtual value]
\label{def:virtual}
A function \(\theta:\V\to\mathbb{R}\) is called a \emph{\fr-optimal virtual value} 
for a \swac setting \(\mathcal{S} = (n, \D, \g, c)\) 
if there exists a \swac that is optimal within the family of \fr \swacs 
for \(\mathcal{S}\) with an allocation rule \(X\) such that
\begin{equation}\label{eq:virv-reward}
    \E\bigl[\trew{v}\bigr] = \E\Bigl[\theta(v) X(v) - c\bigl(X(v)\bigr)\Bigr].
\end{equation}
\end{definition}

As can be seen from the definition, $\theta$
is a \fr-optimal virtual value if it captures the optimal per-unit reward achievable, just like Myerson's virtual value.
However, unlike standard auctions,
in our case there is no unique payment rule,
so the definition explicitly refers to some optimal \swac 
and is not necessarily unique \footnote{
    In previous sections we showed that for the reward classes we are interested in, \fr mechanisms are w.l.o.g., justifying the restriction to monotone allocation rules.
    }. We give an explicit example for a \fr-optimal virtual value below.

For \fr \swacs, defining generalized virtual values is relatively straightforward. Since a \fr \swac has no filters, an agent’s decision depends solely on the total payment associated with each bid—mirroring the characteristics of classical Myersonian auctions. Consequently, the proof of Myerson’s lemma~\cite{myerson1981optimal} applies directly, leading to the following observation:

\begin{observation}\label{ob:fr-means-myerson}
    A truthful \fr \swac has a monotone non-decreasing allocation rule, and its total payments follow the unique payment rule of Myerson:
    \begin{align}
        \p{v}=\sum_{i=1}^l z_i\cdot\Delta_i,
    \end{align}
    where $z_1,\dots,z_l$ are the locations of the jumps in allocation up to (and including) $v$ and $\Delta_1,\dots,\Delta_l$ are the sizes of the jumps.
\end{observation}

Using this observation, it follows that in a \fr \swac the equality $\E\sqbr{\p{v}}=\E\sqbr{\varphi(v)\x{v}}$ holds (where $\varphi(\cdot)$ is the virtual value for revenue). Due to the properties of the reward function, we can define a function that is a \fr-optimal virtual value function for reward functions in which \fr \swacs are w.l.o.g. For example, if $\g(v,p)=3v+2p$, we define $\varphi_{\g}(v)=3v+2\brackets{v-{(1-F(v))}/{f(v)}}=5v-2{(1-F(v))}/{f(v)}$.

\begin{proposition}\label{prop:generelized-virv}
    The function $\varphi_{\g}^{\D}(v)=\g\brackets{v,\varphi(v)}$, where $\varphi(\cdot)$ is the virtual value for revenue, is a \fr-optimal virtual value.
\end{proposition}
\begin{proof}
    Let $\g$ be a reward function, and let $\A$ be a truthful \fr\ \swac\ that is optimal within the family of \fr\ \swacs. 
    By Observation~\ref{ob:fr-means-myerson} the payments of $\A$ are \fr Myerson payments, and thus by~\cite{myerson1981optimal} and as restated in~\cite[Chapter 3]{hartline2013mechanism}, we have $\E_{v\sim\D}\sqbr{\p{v}}=\E_{v\sim\D}\sqbr{\x{v}\varphi(v)}$.

    For $v\in\V$ such that $\x{v}=0$ we have $\trew[\A]{v}=0$. It follows that:
    \begin{align*}
        \E\sqbr{\trew[\A]{v}}&= \E\sqbr{\trew[\A]{v}\given\x{v}>0}\cdot\Pr[\x{v}>0]\\
        &=\E\sqbr{\x{v}\g\brackets{v,\frac{\p{v}}{\x{v}}}-c\brackets{\x{v}}\given\x{v}>0}\cdot\Pr[\x{v}>0]\numberthis\label{eq:fr-rew}
    \end{align*}

    Since our reward takes the form of $\g(v,p)=\alpha_v+\beta_p p$ and using the linearity of expectation:
    \begin{align*}
        &\E\sqbr{\x{v}\g\brackets{v,\frac{\p{v}}{\x{v}}}\given\x{v}>0}=\E\sqbr{\alpha_v\x{v}+\beta_v\p{v}\given\x{v}>0}\\
        &\quad\quad=\E\sqbr{\alpha_v\x{v}+\beta_v\x{v}\varphi(v)\given\x{v}>0}
        =\E\sqbr{\x{v}\g\brackets{v,\varphi(v)}\given\x{v}>0}\\
        &\quad\quad=\E\sqbr{\x{v}\varphi_{\g}^{\D}(v)\given\x{v}>0}\numberthis\label{eq:rew-to-virv-fr}
    \end{align*}

    Combining Equations~\ref{eq:fr-rew} and~\ref{eq:rew-to-virv-fr} completes the proof:
    \begin{align*}
        \E\sqbr{\trew[\A]{v}}&= \E\sqbr{\x{v}\varphi_{\g}^{\D}(v)-c\brackets{\x{v}}\given\x{v}>0}\cdot\Pr[\x{v}>0]\\
        &=\E\sqbr{\x{v}\varphi_{\g}^{\D}(v)-c\brackets{\x{v}}}
    \end{align*}

    The last transition is due to $c(0)=0$ and using the law of total expectation.
\end{proof}

Using the above characterization, we design a \swac that is optimal within \fr \swacs for settings with an over-time cost function. 
It consists of several increasing thresholds, each with a different allocation of days,
using Myerson's payment rule.
We derive a dynamic programming algorithm that systematically builds on solutions for smaller horizons to find solutions for longer ones. The algorithm
runs in polynomial time in the horizon, assuming constant-time access to the pdf and the cdf
of each agent distribution,
and to the inverse virtual value.

\begin{theorem}\label{thm:optima-fr-rental-mech}
Algorithm~\ref{alg:fr-rental} is optimal among all \fr \swacs.
\end{theorem}
\begin{proof}
    To begin with, we explain the precomputations of Algorithm~\ref{alg:cal-intervals}. For each horizon $h$, the algorithm returns a set of disjoint consecutive intervals over the ironed virtual value space, such that there is a single interval for each allocation in $\set{1,\dots,h}$.
Each interval is of type $\textsc{Interval}$, that contains fields $left,\ right,\ alloc$ and $pay$.

    We use $\A_h$ to denote Algorithm~\ref{alg:fr-rental} at horizon $h$, for each $h\in\set{0,1,\dots,n}$. We will prove by induction on $h$ that (1) $\A_h$ is optimal, among all \fr \swacs, with respect to the over-time cost function, and (2) that $R[h]$ is exactly $R_h$, the expected reward from the final $h$ horizons.

    The horizon of $h=0$ is trivial, so we assume that $\A_i$ is  optimal (among \fr\ \swacs) for $i<h$ and prove for $h>0$. Consider an agent $v\in\V$. Let $J$ be the interval in $intervals[h]$ that includes $q\coloneq\irn_h(v)$, and denote $x=J.alloc$. We will show that $x$ maximizes the expression $qy-c_{n,\g}^{\boldsymbol{\D}}(y)$ over all $y\in\sqrrange{n}$, and chooses the \emph{smallest} maximizing $y$. By definition of the over-time cost function, this is equivalent to maximizing $qy+R_{h-\max\set{y,1}}$. 

    \begin{itemize}
        \item If $x=0$, since $R_{h}$ is always nonnegative and monotone non-decreasing, we have that $q<0$, because $q$ wasn't even in an interval with allocation $1$. Therefore $0\cdot q+R_{h-1}\geq yq+R_{h-y}$ for any $y\geq1$.
        \item If $x=1$, by the algorithm's definition $0\leq q\leq \frac{R[h-1]-R[h-y]}{y-x}$ for $y>x$. It follows that:
            \begin{itemize}
                \item Since $c(0)=c(1)$ and $0\leq q$, clearly $q+R_{h-1}\geq 0+R_{h-\max\set{0,1}}$.
                \item By the induction hypothesis, $q\leq \frac{R_{h-1}-R_{h-y}}{y-x}$ for all $y>x$, meaning $yq+R_{h-y}\leq xq+R_{h-x}$ for all $y>1$.
            \end{itemize}
        \item Else, $x\geq2$. By the algorithm's definition, $q$ was allocated $x-1\geq 1$ units at horizon $h-1$, and thus $(x-1)q+R_{h-1-(x-1)}\geq y'q+R_{h-1-y'}$ for all $y'\in\sqrrange{h-1}$.
        Setting $y=y'+1$ it follows that $xq+R_{h-x}\geq yq+R_{h-y}$ for all $y\in\psqrrange{h}$. Additionally since $q$ is greater than the $left$ of the interval, and all of the $left$ fields are at least $0$ by the algorithm's definition, we have $q\geq0$ and therefore $xq+R_{h-x}\geq q+R_{h-1}\geq R_{h-1}$.
    \end{itemize}

    We finished showing that $x$ maximizes the wanted expression. The fact that the allocation chosen is the smallest maximizing one follows from the design of the algorithm, that 
    starts an interval only after the end of the intervals for smaller allocations.
    We will show that the allocation rule $\alloc$ we get from Algorithm~\ref{alg:cal-intervals} is monotone non-decreasing, using the fact that it maximizes $\irn_h(v)\x{v}-c\brackets{\x{v}}$. Let $w<v$ be valuations and suppose, for contradiction, that $\x{w}>\x{v}$. Due to the maximization of $\alloc$ and by the fact that the smallest allocation is chosen,
    \begin{align*}
    &\irn_h(w)\x{w}-c\brackets{\x{w}}>\irn_h(w)\x{v}-c\brackets{\x{v}},\textnormal{ and}\\
        &\irn_h(v)\x{v}-c\brackets{\x{v}}\geq\irn_h(v)\x{w}-c\brackets{\x{w}}.
    \end{align*}
    Combining the above inequalities we get $\irn_h(w)\brackets{\x{w}-\x{v}}>\irn_h(v)\brackets{\x{w}-\x{v}}$. Since $\irn_h(\cdot)$ is monotone non-decreasing, this yields $\irn_h(w)\brackets{\x{w}-\x{v}}>\irn_h(w)\brackets{\x{w}-\x{v}}$ which is a contradiction. Thus, $\alloc$ is monotone non-decreasing.   

    Due to Observation~\ref{ob:fr-means-myerson}, in a \fr \swac the allocation determines the payment schedules, and its allocation rule is monotone. Additionally, as we showed in Proposition~\ref{prop:generelized-virv}, $\theta_h$ is a \fr-optimal virtual value.
    Thus, together with Lemma~\ref{lemma:irn-maximizer-maximizes-expected}, $\A_h$ is optimal among all \swacs with monotone allocation rules, and specifically, among all \fr \swacs.
    
    Additionally, based on the induction hypothesis and since we showed that $x\irn(v)+R_{h-x}$ is maximized, also $R[h]$ is exactly the expected reward from horizon $h$. Finally, we conclude that $\A_h$ is an optimal \swac\ for horizon $h$, among \fr\ \swacs. 
    
     This completes the proof by induction.
\end{proof}

Using Lemma~\ref{lemma:welfare-like}, Corollary \ref{corollary:rev-like-fr-wlog} and Lemma~\ref{lemma:positive-tradeoff-fr-wlog} we get the following result:
\begin{corollary}
    Algorithm~\ref{alg:fr-rental} is optimal for welfare- and revenue-like reward, and FM-optimal for positive tradeoff reward.
\end{corollary}

\begin{algorithm}[h!]
    \KwIn{Distributions $\D_{n},\dots,\D_1$ for horizons $n,\dots,1$ respectively, and reward function $\g$.}
    \KwOut{For each horizon $1,\dots,n$, a set of \textsc{Interval}s.}

    $intervals[0]\gets\emptyset,R[0]\gets0$\;
    \For{h=1:n}{
        $\varphi_h(v)\gets\g(v,\varphi_{\D_h}(v))$\;
        $\irn_h\gets$ ironing of $\varphi_h$\;
        $last\_right\gets\min\brackets{\frac{R[h-1]-R[h-2]}{1},\frac{R[h-1]-R[h-3]}{2},\dots,\frac{R[h-1]}{h-1},\infty}$\;
        $intervals[h]\gets\set{\textsc{NewInterval}(left=0,right=last\_right,alloc=1)}$\;
        
        \For{$i=1:(h-1)$}{
            $J\gets intervals[h-1].forAlloc(i)$\;
            $lft\gets\max\set{J.left, \frac{R[h-1]-R[h-i]}{i-1},last\_right}$\;
            \If{$lft<J.right$}{
                $I\gets \textsc{NewInterval}(left=lft,right=J.right,alloc=J.alloc+1)$\;
                Add $I$ to $intervals[h]$\;
                $last\_right\gets J.right$\;
            }
        }
        $\textsc{SetPaymentsAndReward}(intervals,h,R,\irn_h)$\;
    }
    \Return $intervals$\;
	\caption{Precomputation for Algorithm~\ref{alg:fr-rental}}
    \label{alg:cal-intervals}
\end{algorithm}

\begin{savenotes}
\begin{algorithm}[h!]
    \tcc{Sets the required payment for each interval, and adds the expected reward from horizon $h$ to the list $R$ (i.e. $R_h$).\\
    Assumes $intervals[h]$ is sorted in increasing order by $alloc$.}

        $prev\_alloc\footnote{If the reward is decreasing in payment, $prev\_alloc$ can be set to the first allocation size actually given, even if it's positive. This just removes a constant from all the payments, and is still considered Myerson's unique payment rule.}\gets 0,pay\gets 0,R[h]\gets 0$\;
        \For{$J$ \textnormal{in} $intervals[h]$}{
            $ prob\gets \Pr_{v\sim\D_i}\sqbr{\irn_h(v)\in J}$\;
            \If{$prob>0$}{
                $pay\gets pay+\irn_h^{-1}(J.thresh)\cdot(J.alloc-prev\_alloc)$\footnote{As $\irn_h$ may not have a well-defined inverse function,  we define 
$\irn_h^{-1}(y)\coloneq\sup\sett{v\in\textnormal{support}(\D)}{\irn_h(v)\leq y}$, and if that does not exist, we set it to be $0$.}\;
                $J.setPay(pay)$\;
                $R[h]\gets R[h]+\E_{v\sim\D_i}\sqbr{\irn_h(v)\given \irn_h(v)\in J}\cdot J.alloc+R[h-J.alloc]$\;
                $prev\_alloc\gets J.alloc$\;
            }
        }
	\caption{\textsc{SetPaymentsAndReward}($intervals,h,R,\irn_h$)}
    \label{alg:set-payments-reward}
\end{algorithm}
\end{savenotes}

\begin{algorithm}[h!]
\tcc{Requires precomputation by Algorithm~\ref{alg:cal-intervals}.\\
We set $\varphi_n(v)=\g(v,\varphi(v))$, and $\irn_n$ is the ironing of $\varphi_n$.}
	\KwIn{Distributions $\boldsymbol{\D}=(\D_n,\dots,\D_1)$ for horizons $n,\dots,1$ respectively, and reward $\g$.}
          Accept bid $v$ from agent\;
          \If{$\irn_n(v)$\textnormal{ is in some interval of }intervals[n]} {
              $J\gets$ interval in $intervals[n]$ that includes $\irn_n(v)$\;
             Sell $J.alloc$ units, with a total payment of $J.pay$, divided equally over all days\; 
          }
	\caption{\FR Stagewise Auction with Over-Time Cost}
	\label{alg:fr-rental}
\end{algorithm}

\subsection{Optimal Threshold Mechanism}
\label{sec:threshold-mech}

In this section, we derive a threshold mechanism  for i.i.d.\ agents; we explain the reason for this in Section~\ref{sec:future} below. We show that this mechanism is FM-optimal for any negative tradeoff reward, while for consumer surplus we have global optimality.

\paragraph{Horizon-specific virtual values for threshold payments.}
Recall that negative tradeoff requires threshold payments, a non-\fr payment schedule where
the full payment is charged up front on the first day.
Moreover,
since the payment depends only on the \emph{threshold} for the allocation and not the \emph{allocation size}, the influence that a given payment has on the designer's reward decreases as the horizon increases:
the payment never exceeds $\sup \V$ (the highest possible valuation),
even for very long allocations.
It is therefore impossible to design a generalized virtual value
that
encapsulates the per-unit contribution of the payment to the reward
without taking into consideration the horizon.
To handle this, we introduce a different kind of reward-\ and \emph{horizon-}specific virtual value function $\varphi_{\g,n}(\cdot)$ for a negative tradeoff reward function $\g$ at horizon $n$:
\begin{align}
    \varphi_{\g,n}=\g\brackets{v,\frac{\varphi(v)}{n-1}},
    \quad \text{ where $\varphi(\cdot)$ is the virtual value for revenue.}
\end{align}

There are four steps in the process of finding an optimal threshold mechanism. 
Our first step is to prove that in an optimal threshold auction under a specific cost assumption, all agents receive at least one unit.
Next, we use this characterization to prove that given a monotone allocation rule $\alloc$ implemented by threshold payments, the reward of a threshold auction for each valuation $v$ is \emph{bounded from above} by the reward- and horizon-specific virtual welfare: 
\begin{equation}\label{eq:horizon-virv-bound}
    \trew{v}\leq\varphi_{\g,n}(v)\x{v}-c\brackets{\x{v}}.
\end{equation} This follows from the threshold-payment structure, which ensures that payments are no smaller than a fraction of the Myerson payment. 
The third step is to
find an allocation rule $\x{\cdot}$ that maximizes $\irn_{\g,n}(v)\x{v}-c\brackets{\x{v}}$ pointwise (for every $v$). 
The fourth and final step is to prove that 
the reward obtained from 
our \swac that implements $\x{\cdot}$
 satisfies~\eqref{eq:horizon-virv-bound} \emph{with equality},
proving its optimality.

In order to prove the upper bound on the reward, we give a characterization of optimal allocation rules. 

\begin{lemma}\label{lemma:consumer-surplus-jumps}
    Let $\A$ be a threshold auction that is optimal with respect to negative tradeoff reward, among all threshold auctions.
    Let $z_1,\dots,z_l$ be the valuations in which there is a jump in the allocation $\x{\cdot}$, and $\Delta_i$ be the jump in allocation at $z_i$.
    If the cost function $c$ is such that $c(0)=c(1)=0$ then
     all agents are allocated at least 1 unit, and $\sum_{i=1}^l\Delta_i\leq n-1$.
\end{lemma}
\begin{proof}
    Let $\A$ be a threshold auction that is optimal for negative tradeoff, among all threshold auctions, with cost function $c$ with $c(0)=c(1)=0$. $\A$ is truthful due to Lemma~\ref{lemma:neg-tradeoff-threshold-auction-truthful-and-optimal}. Let $[0,t)\cap\V$ be the interval of agents allocation $0$ units, and $[t,s)$ the interval of agents allocated 1 unit. If no agents are allocated 0 units we are done, as $\alloc$ is monotone and its maximum value is at most $n$.
    
    Else, $\Pr[v<t]>0$. Design a new auction $\A'$ which is like $\A$ but allocates 1 unit to agents with valuations below $t$.
    Let $[t,s)$ be the interval of agents who are allocated 1 item in $\A$ (it can be empty). Since in $\A'$ they can receive 1 unit for free instead of 1 unit at the cost of $t$, they will bid below $t$ in $\A'$. In contrast, an agent $v\geq s$ will not change their bid in $\A'$, since
    \begin{align*}
        \util[\A']{v}{v}=\x{v}v-\p{v}\geq 2v-v=v
    \end{align*}
    and $v$ is their utility from bidding below $t$, which is the only difference between $\A$ and $\A'$.

    Therefore, if we denote the reward as $\g(v,p)=\alpha v-\beta p$ (with $\alpha\geq\beta$):
    \begin{align*}
        \E\sqbr{\trew[\A']{v}}
        &=\E\sqbr{\alpha v-0-c(1)\mid v<t}\cdot\Pr[v<t]+\E\sqbr{\alpha v-0-c(1)\mid t\leq v<s}\cdot\Pr[t\leq v<s]\\
        &\ \ \ \ \ +\E\sqbr{\trew[\A]{v}\mid s\leq v}\cdot\Pr[s\leq v]\\
        &\geq\E\sqbr{\alpha v\mid v<t}\cdot\Pr[v<t]+\E\sqbr{\alpha v-\beta t-c(1)\mid t\leq v<s}\cdot\Pr[t\leq v<s]\\
        &\ \ \ \ \ +\E\sqbr{\trew[\A]{v}\mid s\leq v}\cdot\Pr[s\leq v]\\
        &>0\cdot\Pr[v<t]+\E\sqbr{\trew[\A]{v}\mid t\leq v<s}\cdot\Pr[t\leq v<s]\\
        &\ \ \ \ \ +\E\sqbr{\trew[\A]{v}\mid s\leq v}\cdot\Pr[s\leq v]\\
        &=\E\sqbr{\trew[\A]{v}},
    \end{align*}

    which contradicts the optimality of $\A$.
\end{proof}

We are now ready to prove the upper bound on the expected reward for a negative tradeoff reward \swac.
\begin{lemma}\label{lemma:consumer_surplus_virtual_welfare_bounds_reward}
    Let $\A$ be a threshold auction with $n$ units that is optimal with respect to negative tradeoff reward $\g(v,p)=\alpha v-\beta p$, among all threshold auctions.
   If the cost function $c$ is such that $c(0)=c(1)=0$ then
    \begin{equation*}
        \E\sqbr{\alpha\x{v}v-\beta\p{v}}\leq\E\sqbr{\varphi_{\g,n}(v)\x{v}}.
    \end{equation*}
\end{lemma}
\begin{proof}
    Consider an auction matching the conditions of the claim, and some agent $v$. Denote the locations of the jumps in $\alloc$ by $z_1,\dots,z_l$ and the jumps themselves by $\Delta_1,\dots,\Delta_l$ respectively. By design of the threshold payment, and denoting $j\coloneq\arg\max_{i\in\psqrrange{l}}\set{z_i:z_i\leq v}$,
    \begin{align*}
        \p{v}=z_j=\frac{\sum_{i=1}^{j}\Delta_i}{\sum_{i=1}^{j}\Delta_i}z_j\geq\frac{\sum_{i=1}^{j}\Delta_iz_i}{\sum_{i=1}^{j}\Delta_i}\geq\frac{1}{n-1}\sum_{i=1}^{j}\Delta_iz_i.
    \end{align*}
    The last transition is due to Lemma~\ref{lemma:consumer-surplus-jumps}. Recall that $\sum_{i=1}^{j}\Delta_iz_i$ would be the payment in a \fr auction~\cite{myerson1981optimal}
    \footnote{This is assuming the lowest payment is $0$, and this is fine for us since we are minimizing payments.}, and denote it by $\p[My]{v}$. Since $\p{v}\geq\frac{\p[My]{v}}{n-1}$ and $\E\sqbr{\p[My]{v}}=\E\sqbr{\varphi(v)\x{v}}$ (as shown in~\cite{myerson1981optimal}), it follows that
    \begin{align*}
        \E\sqbr{\alpha\x{v}v-\beta\p{v}}\leq\E\sqbr{\alpha\x{v}v-\frac{\beta\p[My]{v}}{n-1}}=\E\sqbr{\alpha\x{v}v-\frac{\beta\varphi(v)}{n-1}\x{v}}=\E\sqbr{\varphi_n(v)\x{v}}.
    \end{align*}
\end{proof}

Finally, we give the (FM-)optimal negative tradeoff maximizing auction with an over-time cost function. 
\begin{savenotes}
    
\begin{algorithm}[h!]
\caption{A Stagewise Auction for Negative Tradeoff Maximization}\label{alg:consumer-surplus-auc}
 Accept a bid $v\in\V$\;
    \eIf {$n=1$ or  $\irn_n(v)<\frac{R_{\g,n-1}^{\D}}{n-1}$}{
        Sell $1$ unit for free\;
    }{
        Sell all $n$ units and charge $\irn_n^{-1}\brackets{\frac{R_{\g,n-1}^{\D}}{n-1}}$\footnote{As $\irn_n$ may not have a well-defined inverse function,  we define 
$\irn_n^{-1}(y)\coloneq\sup\sett{v\in\V}{\irn_n(v)=y}$.} on the first day.
    }
\end{algorithm}

\end{savenotes}
\begin{savenotes}
    
\begin{algorithm}[h!]
\caption{Precomputation for Algorithm~\ref{alg:consumer-surplus-auc}}\label{alg:threshold-precomputation}
 \KwIn{Reward function $\g(v,p)=\alpha v-\beta p$, and distributions $\D_1,\dots,\D_n$ for horizons $1,\dots,n$ respectively.}
 \KwOut{\begin{enumerate}
 \item A list $R$, such that for all $i\in\sqbr{n}$, $R[i]$ is the expected reward of a rental mechanism consisting of \swacs that run Algorithm~\ref{alg:consumer-surplus-auc} in the last $i$ days.
     \item 
 Ironed virtual values for threshold payments $\irn_1,\dots,\irn_n$, for horizons $1,\dots,n$ .
 \end{enumerate}}
 $R[1]=\alpha\E\limits_{v\sim\D_1}[v]$\;
 \For{$i=2,\dots,n$}{
$\tau_i\gets \irn_i^{-1}\brackets{\frac{R[i-1]}{i-1}}$\;
 $R[i]\gets \brackets{\alpha\E\limits_{v\sim\D_i}\sqbr{v\given v<\tau_i}+R[i-1]}\Pr\limits_{v\sim\D_i}\sqbr{v<\tau_i}
    +\brackets{\alpha i\E\limits_{v\sim\D_i}\sqbr{v\given v\geq \tau_i}-\beta \tau_i}\Pr\limits_{v\sim\D_i}\sqbr{v\geq \tau_i}$\;
 }
 \Return $R$
\end{algorithm}

\end{savenotes}

\begin{theorem}\label{theorem:threshold-mech}
    Algorithm~\ref{alg:consumer-surplus-auc} is an FM-optimal and truthful negative tradeoff \swac with over-time cost, for any horizon $n$ and distribution $\D$. For consumer surplus, it is globally optimal.
\end{theorem}
\begin{proof} For simplicity, when clear, we omit the superscripts and/or subscripts from $R_{\g,n}^{\D}$ and $c_{\g,n}^{\D}$ in this proof. Additionally, when we write \emph{optimal} in this proof, we mean FM-optimal, unless the reward is consumer surplus, in which case we are really referring to global optimality. Assume that the reward takes the form of $\g(v,p)=\alpha v-\beta p$.

First, observe that truthfulness follows from Lemma~\ref{lemma:neg-tradeoff-threshold-auction-truthful-and-optimal}.

    We move on to show, by induction on $n$, that:
    \begin{enumerate}
        \item The allocation from Algorithm~\ref{alg:consumer-surplus-auc}, which we denote $\alloc$, maximizes $\irn_n(v)\x{v}-c\brackets{\x{v}}$ pointwise, and\label{item:1_consumer_surplus}
        \item  Algorithm~\ref{alg:consumer-surplus-auc} for horizon $n$ is optimal.\label{item:2_consumer_surplus}
        \item The value $R[n]$ returned from Algorithm~\ref{alg:threshold-precomputation} is equal to the expected reward from horizon $n$, which we denote $R_n$.
    \end{enumerate}

     For $n=1$ all claims are easy to see (when there is a single unit, for negative tradeoff it is clearly optimal to give the unit for free to any agent, resulting in an expected reward of $\E[\alpha v]$). Now assume correctness up to $n-1$ and prove for $n\geq2$.

    The first step is to show that an allocation of $0$ is never required for optimality:
    by definition of the virtual value for revenue, 
    \begin{equation*}
        \varphi_{\g,n}(v)
        =
        \alpha v-\frac{\beta}{n-1}\brackets{v-\frac{1-F(v)}{f(v)}}
        \geq
        \beta v\brackets{1-\frac{1}{n-1}}+\beta\frac{1-F(v)}{f(v)(n-1)}
        \geq 0
        .
        \end{equation*}
        By definition of the ironing process (Def.~\ref{def:ironing}), also $\irn_n(v)\geq0$.
         Since $c_n(1)=c_n(0)$,
          $
             \irn_n\cdot0 - c_n\brackets{0}\leq \irn_n\cdot 1 - c_n\brackets{1},
             $
         {and}
         hence 
         {an optimal allocation}
         can always allocate at least one day to any agent.
         {This allows us to focus on strictly positive allocations,
         and prove pointwise optimality
         by showing that
         for any valuation $v \in \V$,
        our allocation rule is no worse than any other positive allocation $x \geq 1$:
         \begin{equation*}
            \irn_n(v)\cdot\x{v}-c_n(\x{v})\geq\irn_n(v)\cdot x-c_n(x).
        \end{equation*}
         }

Suppose $\irn_n(v)<R_{n-1}/(n-1)$ ,
        so that $\x{v} = 1$.
        Then for any $x \geq 1$,
        \begin{align*}
            \irn_n(v)x-c_n(x)&=\irn_n(v)+(x-1)\irn_n(v)-R_{n-1}+R_{n-x}
            <
            \irn_n(v)+\frac{x-1}{n-1}R_{n-1}-R_{n-1}+R_{n-x}
            \\
            &=
            \irn_n(v)+(n-x)\brackets{\frac{R_{n-x}}{n-x}-\frac{R_{n-1}}{n-1}}.
        \end{align*}
        We now rely on an interesting (and intuitive) fact, which we prove in Lemma~\ref{lemma:avg_reward_increase_with_horizon}: if a monotone allocation rule of \emph{one or all}, as in Algorithm~\ref{alg:consumer-surplus-auc} is w.l.o.g., then
the expected average daily reward from a rental game with negative tradeoff reward is non-decreasing in the horizon;
that is,
if $n > m$
then ${R_n}/{n}\geq {R_m}/{m}$.
We can apply Lemma~\ref{lemma:avg_reward_increase_with_horizon} thanks to our induction hypothesis, which tells us that Algorithm~\ref{alg:consumer-surplus-auc} is indeed optimal.
Therefore,        
        \begin{align*}
                    \irn_n(v)+(n-x)\brackets{\frac{R_{n-x}}{n-x}-\frac{R_{n-1}}{n-1}}
            \leq \irn_n(v)
            =
            \irn_n(v) \cdot 1 -c_n(1).
        \end{align*}

        Now suppose $\irn_n(v)\geq R_{n-1}/(n-1)$, so that $\x{v}=n$. Then for any $x\geq 1$,
        \begin{align*}
            n\irn_n(v)-c(n)&=x\irn_n(v)-R_{n-1}+R_{n-x}+(n-x)\irn_n(v)-R_{n-x}\\
            &\geq x\irn_n(v)-R_{n-1}+R_{n-x}+(n-x)\brackets{\frac{R_{n-1}}{n-1}-\frac{R_{n-x}}{n-x}}\\
            &\geq x\irn_n(v)-R_{n-1}
        \end{align*}
        Here too the last inequality is due to Lemma~\ref{lemma:avg_reward_increase_with_horizon} and Item~\ref{item:2_consumer_surplus} from the induction hypothesis.

This completes the induction proof for Item~\ref{item:1_consumer_surplus} for $n$. 

By Lemma~\ref{lemma:pos-neg-tradeoff-samono} (and Lemma~\ref{lemma:consumer-surplus-amono} for consumer surplus), allocation monotonicity is w.l.o.g. for $\g$, meaning that there is an optimal \swac $\A^\ast$ with a monotone allocation rule $\x[\A^\ast]{\cdot}$.
Since $\alloc$ is pointwise maximizing, we can apply Lemma~\ref{lemma:irn-maximizer-maximizes-expected}, telling us that $\E\sqbr{\varphi_{n}(v)\x{v}}-c\brackets{\x{v}}$ is maximized among all monotone non-decreasing allocation rules. Specifically,
\begin{align}\label{eq:varphi-smaller}
    \E\sqbr{\varphi_{n}(v)\x[\A^\ast]{v}-c\brackets{\x[\A^\ast]{v}}}\leq\E\sqbr{\varphi_{n}(v)\x{v}-c\brackets{\x{v}}}.
\end{align}

Combining Equation~\ref{eq:varphi-smaller} with Lemma~\ref{lemma:consumer_surplus_virtual_welfare_bounds_reward} gives us:
\begin{align}\label{eq:opt-bound}
    \E\sqbr{\trew[\A^\ast]{v}}\leq\E\sqbr{\varphi_{n}(v)\x{v}-c\brackets{\x{v}}}.
\end{align}

 Denote Algorithm~\ref{alg:consumer-surplus-auc} at horizon $n$ by $\A_n$.
Observe that since $\A_n$ is a one-or-all threshold auction, the payment charged is exactly equal to ${1}/{(n-1)}$ of the Myerson payment: if $v$ is allocated one unit, its threshold payment and the Myerson payment are both zero; if $v$ is greater than the threshold $t$ and the agent is allocated $n$ units, the threshold payment is $t$ while the Myerson payment is $(n-1)\cdot t$.
Therefore, if we denote the total payment rule of $\A_n$ by $\p{\cdot}$ and the Myerson payment rule by $\p[My]{\cdot}$,
\begin{align*}
    \E\sqbr{\trew[\A_n]{v}}&=\E\sqbr{\alpha\x{v} v -\beta\p{v}-c\brackets{\x{v}} }
    \\
                        &=\E\sqbr{\alpha\x{v} v -\frac{\beta}{n-1}\p[My]{v}-c\brackets{\x{v}} }
                        =\E\sqbr{\alpha\x{v} v -\beta\frac{\varphi(v)\x{v}}{n-1}-c\brackets{\x{v}}}
                        \\
                        &
                        =\E\sqbr{\x{v}\g\brackets{v,\frac{\varphi(v)}{n-1}}-c\brackets{\x{v}}}
                        =\E\sqbr{\varphi_{n}(v)\x{v}-c\brackets{\x{v}}}.
\end{align*}

Combining this equality with Equation~\ref{eq:opt-bound},
\begin{align}
    \E\sqbr{\trew[\A_n]{v}}\geq\E\sqbr{\trew[\A^\ast]{v}},
\end{align}
completing the proof of the optimality of $\A_n$, which is Item~\ref{item:2_consumer_surplus} in the induction claim.

Once we showed $\A_n$ is optimal, it is clear that $R[n]=R_n$, as it calculates the expected reward from an optimal \swac.
    
\end{proof}

The following lemma shows that in some families of reward functions, the average reward per day grows as the horizon grows. This is true, among others, for welfare and consumer surplus.
\begin{lemma}\label{lemma:avg_reward_increase_with_horizon}
    Let  $\S$ be a \swac setting with over-time cost, such that there is an optimal \swac for horizons $\leq n$ in which the allocation for horizon $h$ is in $\set{1,h}$. 

    It holds that for all $m<n$,
    \begin{equation*}
        \frac{R_n^{\D,\g}}{n}\geq\frac{R_m^{\D,\g}}{m}.
    \end{equation*}
\end{lemma}
It is trivial that a shorter rental game yields no-greater reward than a longer one, but it is not trivial to show that the \emph{average} daily reward is no greater, as both the numerator (the reward) and the denominator (the number of days) decrease.
In fact, currently our proof is only for negative tradeoff,
although we expect that this fact is true for any reward function.
\begin{proof}[Proof of Lemma~\ref{lemma:avg_reward_increase_with_horizon}]
    We prove this claim by induction on $n$. For $n=1$ it is immediate, so we assume correctness up to $n-1$ and prove for $n\geq 2$.
    Let $\g,\D$ as defined in the lemma, and let $n>m\geq 1$. By choice of $\g$ and $\D$, we can select an optimal and truthful $\brackets{m,\D,\g,c_m^{\D,\g}}$-auction $\A_m$ with $\x[\A_m]{v}\in\set{1,m}$ for all $v\in\V$. Denote the threshold of $\alloc[\A_m]$ with $t_m$, and the matching total payments with $p_1,p_m$ such that:
    \begin{align*}
     \x[\A_m]{v}=\begin{cases}
        1,&v<t_m\\
        m,&t_m\leq v
    \end{cases}\ \ \ \ \textnormal{and}\ \ \ \ \ \    \p[\A_m]{v}=\begin{cases}
        p_1,&v<t_m\\
        p_m,&t_m\leq v
    \end{cases}.
    \end{align*}
    Additionally, denote the corresponding filters with $f_m$.

    We first make an observation that will aid as in the design of a new auction. Let $v<t_m$. We know that either $v<f_m$ or $mv-p_m\leq v-p_1$. If $mv-p_m\leq v-p_1$ it follows that $v\leq(p_m-p_1)/(m-1)\leq p_m$. So either way, $p_m\geq t_m$.

    We design an $n$-auction as follows:
    \begin{align*}
        \x[\A_n]{v}=\begin{cases}
            1,& v<t_m\\
            n,& v\geq t_m
        \end{cases}.
    \end{align*}
    For bids $v$ such that $\x[\A_m]{v}= 1$ we keep $\p[\A_n]{v}=\p[\A_m]{v}=p_1$. For larger bids, the payment rule is a bit more complicated:
    \begin{itemize}
        \item The total payment will be
        \begin{equation}
            \p[\A_n]{v}=p_n\coloneq\frac{n}{m}p_m
        \end{equation}
        \item The payment for the first day will be the \emph{filter}:
        \begin{equation}
            \f[\A_n]{v}=f_n\coloneq t_m
        \end{equation}
        That is, the filter is the lowest valuation of agent who is allocated $m$ units in $\A_m$. Indeed it's possible, as we observed that $t_m\leq p_m\leq \frac{n}{m}p_m$. 
        \item The payment for the subsequent units will be \fr and amount to the remaining payment due. That is, for units $2,\dots,n$ the payment will be $\frac{p_n-f_n}{n-1}$ each day.
    \end{itemize}

    We will show truthfulness of $\A_n$. Consider some $v\in\V$.
    \begin{itemize}
        \item If $v<t_m$, the agent can only bid below $t_m=f_n$, which will always result in  the same outcome: 1 unit for payment $p_1$. Hence in this case, agent $v$ is truthful.
        \item If $v\geq t_m$, we have:
        \begin{align*}
            \util[\A_n]{v}{v}=\x[\A_n]{v}v-\p[\A_n]{v}=nv-\frac{n}{m}p_m=\frac{n}{m}\util[\A_m]{v}{v}\geq \util[\A_m]{v}{v}.
        \end{align*}
        Since in $\A_m$ agent $v$ is truthful, and besides for bidding $v$ there is no other bid in $\A_n$ that would give a different outcome compared to $\A_m$, this proves that also in $\A_n$ agent $v$ is truthful.
    \end{itemize}

    We move on to prove the claim. Let $\A_1,\dots,\A_{m-1},\A_{m+1},\dots,\A_n$ be  optimal auctions for horizons $1,\dots,m-1,m+1,\dots,n-1$ respectively, distribution $\D$, reward function $\g$ and the matching over-time cost function. Due to Theorem~\ref{theorem:rental_as_auctions}, the rental mechanism $\M_m\coloneq\brackets{\A_1,\dots,\A_m}$ is an optimal $\brackets{m,\D,\g}$-rental mechanism, i.e. $R_m=\E\sqbr{\M_m}$. Define similarly $\M_n\coloneq\brackets{\A_1,\dots,\A_n}$.
    It holds that:
    \begin{align*}
        \frac{R_n}{n}-\frac{R_m}{m}&\geq\frac{\E\sqbr{\M_n}}{n}-\frac{\E\sqbr{\M_m}}{m}\\
        &=\Pr\sqbr{v<t_m}\brackets{\brackets{\frac{1}{n}-\frac{1}{m}}\E\sqbr{\g(v,p_1)\mid v<t_m}+\frac{R_{n-1}}{n}-\frac{R_{m-1}}{m}}\numberthis\label{eq:4.6}\\
        &\ \ \ \ +\Pr[t_m\leq v]\brackets{\frac{n\cdot\E\sqbr{\g\brackets{v,\frac{p_n}{n}}\mid t_m\leq v}}{n}-\frac{m\cdot\E\sqbr{\g\brackets{v,\frac{p_m}{m}}\mid t_m\leq v}}{m}}
    \end{align*}

    Since $p_n=\frac{n}{m}p_m$, the second summand is zero. We continue the analysis depending on the value of $m$. If $m=1$, using the induction hypothesis $\frac{R_{n-1}}{n}\geq\frac{R_{1}(n-1)}{n}=R_1\brackets{\frac{1}{m}-\frac{1}{n}}$. Otherwise, $\frac{R_{n-1}}{n}-\frac{R_{m-1}}{m}\geq\frac{R_{m-1}(n-1)}{n(m-1)}-\frac{R_{m-1}}{m}=\frac{R_{m-1}}{m-1}\brackets{\frac{1}{m}-\frac{1}{n}}\geq R_1\brackets{\frac{1}{m}-\frac{1}{n}}$. 
    Therefore,
    \begin{align}\label{eq:4.7}
        \frac{R_n}{n}-\frac{R_m}{m}\geq\Pr\sqbr{v<t_m}\brackets{\brackets{\frac{1}{m}-\frac{1}{n}}\brackets{R_1-\E\sqbr{\g(v,p_1)\mid v<t_m}}}.
    \end{align}

    Observe that we can define a $(1,\D,\g)$-rental mechanism $\M$ as follows:
    \begin{align*}
        \x[\M]{v}=1\ \ \ \ \ \ \textnormal{ and }\ \ \ \ \ \ \p[\M]{v}=p_1
    \end{align*}
    Since $\g$ is non-decreasing with the valuations, we get $R_1\geq\E[\M]=\E[\g(v,p_1)]\geq\E[\g(v,p_1)\mid v<t_m]$. Plugging this into Equation~\ref{eq:4.7} yields $\frac{R_n}{n}-\frac{R_m}{m}$ and completes the proof.
\end{proof}

\section{Gap Between Variable Pricing to Fixed-Rate Pricing}\label{sec:fr-gap}

In this section we demonstrate, by example, how big the gap between \fr \swacs to \swacs with variable pricing can be.
The \swac setting we focus on is one in which the agent's valuations are distributed uniformly over $[0,1]$, the reward function is consumer surplus ($\g(v,p)=v-p$), and the cost function is the over-time cost function from an $n$-rental game in which the agent's valuations are distributed i.i.d. (uniformly over $[0,1]$). We denote this \swac setting by $\S_n$. 

\begin{lemma}\label{lemma:best-fr-consumer-surplus}
    In the \swac setting $\S_n$:
    \begin{enumerate}
        \item Among \fr \swacs it is w.l.o.g.\ to sell one unit to any agent that arrives, for free, and this achieves an expected reward of $0.5$.
        \item $R_n\geq0.5n$
        \item A rental mechanism that is \fr-optimal has an expected reward of $0.5n$.
    \end{enumerate}
\end{lemma}
\begin{proof}
    Clearly a rental mechanism in our setting over $n$ days can get an expected reward of at least $0.5n$ by renting the asset to every agent for a single day, each day getting a reward of $\E[v]=0.5$.

    We now prove the first point using Algorithm~\ref{alg:fr-rental}, which we showed is optimal among all \fr \swacs in Theorem~\ref{thm:optima-fr-rental-mech}. Denoting $\D=\Uni[0,1]$, the virtual value for revenue is $\varphi(v)=v-\frac{1-v}{1}=2v-1$. The virtual value for consumer surplus is $\varphi_{\g}^{\D}(v)=v-\brackets{2v-1}=1-v$. This is a \emph{decreasing} function of $v$. Applying the ironing process steps by step gives $\irn_{\g}^{\D}(v)=0.5$. 

    The lower bound for allocation $1$ is $\brackets{\irn_{\g}^{\D}}^{-1}(0)=0$, as we defined the inverse function. We now want to calculate the upper bound, which we denote $U$. By the first definition of $last\_right$ in Algorithm~\ref{alg:cal-intervals}, for $n=1$ clearly all valuations will get an allocation of $1$, so we can move on to $n>1$. Observe that an $(n_1+n_2)$-rental mechanism with i.i.d. distributions can consist of two shorter rental mechanisms with horizons $n_1$ and $n_2$, such that the agents in the first rental mechanism are not allocated the asset past the first $n_1$ days, in which the first rental mechanism is applied, and after $n_1$ days the second rental mechanism is applied. Therefore, using the induction hypothesis, for any $x\in\set{1,\dots,n}$,
    \begin{align*}
        \frac{R_{n-1}-R_{n-x}}{x-1}\geq\frac{R_{x-1}}{x-1}\geq\frac{1}{2}.
    \end{align*}
    This means that $\brackets{\irn_{\g}^{\D}}^{-1}(U)\geq\brackets{\irn_{\g}^{\D}}^{-1}(0.5)=1$. Therefore, all agents are allocated 1 unit in the \swac with $n$ units, which gives a reward of $\E[v]-0=0.5$. Due to the algorithm's optimality among \fr \swacs, we are done proving the optimal among \fr \swac.

    In terms of a rental mechanism that is \fr-optimal, observe that since at every horizon it is optimal for the designer to rent to the agent for a single day for free (as we just proved), the rental mechanism would yield an expected reward of $n\cdot\E[v]=0.5n$.
\end{proof}

\begin{lemma}\label{lemma:consumer_surplus_filter}
    There is a non \rmono \swac $\A$ that is optimal for $\S_4$ with $\E[A]=0.625$.

    Additionally, $R_1=0.5,\ R_2=1$ and $R_3=1.5$. The thresholds from Algorithm~\ref{alg:threshold-precomputation} are $\tau_2=\tau_3=1$.
\end{lemma}
\begin{proof}
We use the optimal threshold allocation from Algorithm~\ref{alg:consumer-surplus-auc}, which we prove that is optimal in Theorem~\ref{theorem:threshold-mech}.

We begin by the precomputations needed, as done by Algorithm~\ref{alg:threshold-precomputation}. First $R_1=0.5$.
Next for $i=2$, the horizon-specific virtual value is 
$\varphi_{2}(v)=v-\frac{v-\frac{1-v}{1}}{1}=1-v$. Therefore, following the ironing procedure, $\irn_{2}(v)=0.5$. We get $\tau_2=\irn_2^{-1}(0.5)=1$, meaning $R_2=0.5+0.5=1$. Similarly for $i=3$ we have $\varphi_{3}(v)=v-\frac{v-\frac{1-v}{1}}{2}=0.5$. Since this is already non-decreasing, also $\irn_3(v)=0.5$. 
Thus $\tau_3=\irn_3^{-1}(1/2)=1$ and $R_3=0.5+1=1.5$.

Finally we can calculate $\E[\A]$. For $i=4$ we have $\varphi_4(v)=v-\frac{v-\frac{1-v}{1}}{3}=\frac{1}{3}(v+1)$. This is non-decreasing, so also $\irn_4(v)=\frac{1}{3}(v+1)$. It follows that $\tau_4=\irn_4^{-1}(1.5/3)=3\cdot0.5-1=0.5$. Therefore, since at horizon $4$ we have $c(1)=R_3-R_3=0$ and $c(4)=R_3=1.5$,
\begin{align*}
    \E[\A]&=\E\sqbr{\trew[\A]{v}\given v<0.5}\cdot\Pr[v<0.5]+\E\sqbr{\trew[\A]{v}\given v\geq0.5}\cdot\Pr[v\geq0.5]\\
    &=(0.25)\cdot0.5+(4\cdot0.75-0.5-1.5)\cdot0.5=0.625.
\end{align*}

\end{proof}

Comparing the results of Lemma~\ref{lemma:best-fr-consumer-surplus} with Lemma~\ref{lemma:consumer_surplus_filter} gives us the following result:
\begin{corollary}\label{cor:neg-tadeoff-fr-not-wlog}
    For negative tradeoff reward, \fr payment schedules are not w.l.o.g.
\end{corollary}

\begin{lemma}\label{lem:best-consumer-surplus-uni}
    For the settings $\S_n$ it holds that $\lim\limits_{n\rightarrow\infty}\frac{R_n}{n}=1$.
\end{lemma}

\begin{proof}
    Before we begin, observe that for $n>3$ it holds that the horizon-specific virtual value for consumer surplus is 
    \begin{align*}
        \varphi_n(v)=v-\frac{2v-1}{n-1}=\brackets{1-\frac{2}{n-1}}v+\frac{1}{n-1}.
    \end{align*}

    Since for $n>3$ this is non-decreasing in $v$, also $ \irn_n(v)=\brackets{1-\frac{2}{n-1}}v+\frac{1}{n-1}$, and thus
    \begin{align*}
        \irn_n^{-1}(x)=\frac{(n-1)x-1}{n-3}.
    \end{align*}

    Recall that Algorithm~\ref{alg:consumer-surplus-auc} is optimal for $\S_n$. Based on its precomputations in Algorithm~\ref{alg:threshold-precomputation} and due to its correctness,
    we consider two recurring relations that are intertwined with each other:
    \begin{align}
        &\tau_i=\frac{R_{i-1}-1}{i-3}=\brackets{1-\frac{1}{i-3}}\ell_{i-1}-\frac{1}{i-3},\label{eq:tau}\\
        &\ell_i\coloneq\frac{R_i}{i}=\brackets{\frac{\tau_i}{2i}+\frac{i-1}{i}\frac{R_{i-1}}{i-1}}\tau_i+\brackets{\frac{\tau_i+1}{2}-\frac{\tau_i}{i}}\brackets{1-\tau_i}\label{eq:ell}.
    \end{align}

    The above sequences start at $\ell_1=\ell_2=\ell_3=0.5$ and $\tau_2=\tau_3=1$ based on Lemma~\ref{lemma:consumer_surplus_filter}.

    Observe that since the per-day reward of a consumer surplus \swac is bounded by $1$ (the highest possible valuation), $\ell_i=R_i/i\leq1$. Additionally by Lemma~\ref{lemma:avg_reward_increase_with_horizon} the sequence $\ell_i$ is non-decreasing in $i$, and thus also $\tau_i$ is non-decreasing with $i$. It follows that $\tau_i$ and $\ell_i$ have a constant limit: denote their limits by $L_1$ and $L_2$ respectively.

    Substituting these limits into Eq.~\ref{eq:tau} and taking $i$ to infinity, we get $L_1=L_2$, so denote $L=L_1=L_2$. Now substituting this limit with $\ell_i$ and $\tau_i$ in Eq.~\ref{eq:ell} and taking $i$ to infinity we get:
    \begin{align*}
        L=L^2+\brackets{\frac{L}{2}+\frac{1}{2}}\brackets{1-L}=\frac{L^2}{2}+\frac{1}{2}.
    \end{align*}
    Solving the above equation we get $L=1$, completing the proof.
\end{proof}

Together with Lemma~\ref{lemma:best-fr-consumer-surplus} we get the following result:
\begin{corollary}\label{cor:fr-gap}
    For a consumer-surplus-maximizing rental mechanism where the agent's valuations are independently and identically distributed (i.i.d.) uniformly over $[0,1]$, the ratio of variable pricing to \fr pricing tends to $2$.

    Formally, the ratio between the optimal rental mechanism and the optimal \fr rental mechanism tends to $2$.
\end{corollary}

\section{Summary of Our Techniques}
\label{sec:our-techniques}

Our main technical contributions in this work rely on several ideas.

\paragraph{Proving monotonicity for optimal \swacs by exploiting violations.}
Since not every truthful \swac is reward- or allocation-monotone,
our results require a different style of proof compared to the standard~\cite{myerson1981optimal}. 
To establish reward- and allocation-monotonicity
for each class of reward function (except negative threshold, where reward-monotonicity does not hold),
 we assume that some optimal \swac is not monotone, and then carefully identify an ``actionable'' violation of monotonicity.
In its most general form, the violation takes the form of two disjoint valuation sets $S_1, S_2 \subseteq \V$ such that $\sup S_1 \leq \inf S_2$, yet the reward or allocation is \emph{greater} for all valuations in $S_1$ than for those in $S_2$.
Our goal now is to force all agents in one of the two sets, which has non-zero measure, %
to switch to bids that yield improved reward,
 using the existence of the other set and the monotonicity violation between the two sets to show that such a bid exists and becomes desirable for these agents.

It is easy to get agents in a truthful \swac to switch their bid by simply removing their valuations from the \swac's bidding space.
However, if we are not careful,
the results can be unpredictable:
these agents can switch to a bid
that will yield lower reward.
To prevent this,
we further shape the auction by
\begin{enumerate*}
    \item removing undesirable bids from the bidding space, so that the agents we are interested in cannot switch to them; and
    \item adding or removing \emph{filters}:
    adding a filter can
prevent agents below some threshold from bidding in a certain range,
and removing a filter ensures that a desired bid becomes available.
\end{enumerate*}
Of course, these changes may also cause \emph{other} agents to now switch their bids as well,
and care is required to ensure that this does not decrease our reward.
 Ultimately we obtain a \swac that yields better reward than the original one, contradicting its optimality.

\paragraph{Generalized virtual values and their role in optimal mechanisms.}
Myerson's virtual values~\cite{myerson1981optimal} rely on  the fact that in a standard auction,
a given allocation can only be truthfully implemented by a unique payment rule.
This is no longer true for \swacs.
Thus, we define \emph{generalized virtual values} with respect to \emph{some} optimal and truthful \swac whose per-unit reward they capture;
our generalized virtual values do not capture the per-unit reward of \emph{every} truthful \swac.
Nevertheless, we show that they are useful in developing optimal \swacs,
as they still reduce the mechanism design problem to 
that of finding a pointwise-optimal allocation rule,
analogous to Myerson's approach.
This still leaves open the problem of finding a truthful payment rule matching the allocation that we  found,
which can be delicate
and even intertwined with the properties of the generalized virtual value itself:
for example, in the case of negative reward,
the generalized virtual value that we define
is always an \emph{upper bound} on the reward,
but is only guaranteed to be a \emph{lower bound} (i.e., tight)
under a specific payment structure (threshold payments with a single threshold). 

\section{Future Directions}\label{sec:future}

We highlight several avenues for future research, including open questions and extensions of the rental game model.

\paragraph{Optimal rental mechanisms for negative tradeoff with non-identical distributions.}  
Our optimal rental mechanism for negative tradeoff reward is currently restricted to i.i.d.\ agent valuations:
the function that we define in~\eqref{eq:horizon-virv-bound}
is a generalized virtual value only for threshold mechanisms that always allocate either 1 or $n$ units.
When agent distributions are not identical,
an optimal allocation may need to allocate intermediate values 
(e.g., rent for $2$ days if a stronger agent is guaranteed to arrive in two days, but not in one day).
This causes the horizon-specific virtual value to no longer reflect the reward accurately.
It is readily seen that the function defined in Eq.~\ref{eq:horizon-virv-bound} treats all days symmetrically,
and is therefore not suitable for non-identical distributions --- although interestingly, it is still an \emph{upper bound} on the reward even in this case.
Recently, we showed in~\cite{berzack2026optimal} that if a stagewise auction has any \emph{concave} cost function, then it is still optimal to always allocate either $0$ or $n$ units; this allows extending the rental game beyond i.i.d.\ buyers, as long as the over-time cost function remains concave.
However, finding a generalization
for arbitrary non-identical agent distributions
remains an interesting open problem.

\paragraph{Competitive ratios and benchmarks.}  
Understanding the \emph{competitive ratio} of online rental mechanisms presents two key challenges: defining an appropriate benchmark and computing the ratio itself. Ideally, the benchmark for an online problem should be an \emph{offline} version of the same problem, where the entire input is known in advance. In our case this is difficult to pin down. It is natural to consider the \emph{prophet benchmark}, which is the standard in prophet inequalities (and was used, for example, in~\cite{AbelsPS23}). 
This benchmark assumes that the mechanism has full {advance knowledge} of the realized valuation sequence, but this mixes up the \emph{online} aspect --- uncertainty about future agents --- with the \emph{mechanism design} aspect, uncertainty about the \emph{current} agent's valuation.
(In revenue maximization,
for example, this is very powerful:
if the agent's valuation is known, it can always be extracted in full).
Weaker prophet models may offer a more fair comparison, but it is challenging to define a benchmark that preserves the game-theoretic aspects of the problem while eliminating the online aspects.

On the computational side, prior work~\cite{AbelsPS23} 
with known agent valuations
was only able to analyze competitive ratios for \emph{suboptimal} algorithms, and only for welfare maximization. Extending their techniques to our broader class of reward functions remains an open problem even when the agent valuations are {known} to the designer (as in~\cite{AbelsPS23}).

\paragraph{Rental extensions and market structures.}  
In our model, a rental agreement cannot be extended, but in some scenarios extensions make economic sense. Interestingly, if we allow the current renter of the asset to negotiate for an extension after their (irrevocable) initial agreement ends and the next agent arrives,
we can potentially exploit the competition between the two agents. Similar competition may arise if agents can \emph{reserve} the item ahead of time, as in many rental markets (hotels, equipment rentals, etc.).
We point out that some variations on the rental game may still be reducible to fixed sequences of 
\swacs
with modified cost functions (reflecting a different opportunity cost), in which case our structural results from Section~\ref{sec:tech-overview-mono} would still apply. Recently, we made a step in this direction~\cite{berzack2026optimal},  effectively extending the rental game to allow for multiple customers arriving every day, but still assuming one customer can rent the asset at a time.

\paragraph{Infinite horizons and discounting.}  
Our results assume a \emph{finite and known horizon}, a standard assumption in {secretary problems} and {prophet inequalities}.
It is interesting to analyze the \emph{infinite-horizon} case with \emph{discounted reward} (see, e.g.,~\cite{BalseiroML17}), where optimal strategies remain unclear.

\paragraph{Non-Bayesian and learning-based mechanisms.}  
Our model assumes \emph{prior knowledge of valuation distributions}. While this is common in auction theory, a fascinating question is: what can be achieved with limited information? Investigating mechanisms that operate with \emph{no prior knowledge} or a limited number of \emph{sampled valuations},
or mechanisms that
have no prior knowledge and
\emph{minimize regret},
is an interesting challenge.

\paragraph{Combinatorial Preferences and Learning-Based Mechanisms.}  
We assume that the agents' utilities are additive over time, but real-world preferences are often more complex. Some agents may have a preferred rental duration, while others may experience diminishing returns. Extending rental mechanisms to accommodate such combinatorial preferences is an open question.
Additionally, our analysis assumes \emph{prior knowledge of valuation distributions}. While this is common in auction theory, a key question is: what can be achieved with limited information? Investigating mechanisms that operate with \emph{no prior knowledge} or a limited number of \emph{sampled valuations},
or mechanisms that \emph{minimize regret},
is an interesting challenge.

\section*{Acknowledgments}
We thank Zohar Barak for helpful comments on an earlier draft and additional feedback. We also thank Michal Feldman, Amos Fiat, Adi Fine and Omri Porat for valuable discussions and suggestions.
This work
is funded by the European Union (ERC grant No. 101077862, project: ALGOCONTRACT, PI:
Inbal Talgam-Cohen), by the Israeli Science Foundation (Grants No. 2801/20, 3725/24 and 3331/24), by a Google Research Scholar Award, and by the United States-Israel Binational Science Foundation (BSF grant no.~2021680).

\bibliography{main}

@techreport{beccuti2020optimality,
  title={On the optimality of price-posting in rental markets},
  author={Beccuti, Juan},
  year={2020},
  institution={Discussion Papers}
}

@inproceedings{DevanurPS15,
  author       = {Nikhil R. Devanur and
                  Yuval Peres and
                  Balasubramanian Sivan},
  editor       = {Piotr Indyk},
  title        = {Perfect Bayesian Equilibria in Repeated Sales},
  booktitle    = {Proceedings of the Twenty-Sixth Annual {ACM-SIAM} Symposium on Discrete
                  Algorithms, {SODA} 2015, San Diego, CA, USA, January 4-6, 2015},
  pages        = {983--1002},
  publisher    = {{SIAM}},
  year         = {2015},
  url          = {https://doi.org/10.1137/1.9781611973730.67},
  doi          = {10.1137/1.9781611973730.67},
  bibsource    = {dblp computer science bibliography, https://dblp.org}
}

@inproceedings{BravermanSW21,
  author       = {Mark Braverman and
                  Jon Schneider and
                  S. Matthew Weinberg},
  editor       = {P{\'{e}}ter Bir{\'{o}} and
                  Shuchi Chawla and
                  Federico Echenique},
  title        = {Prior-free Dynamic Mechanism Design With Limited Liability},
  booktitle    = {{EC} '21: The 22nd {ACM} Conference on Economics and Computation,
                  Budapest, Hungary, July 18-23, 2021},
  pages        = {204--223},
  publisher    = {{ACM}},
  year         = {2021},
  url          = {https://doi.org/10.1145/3465456.3467654},
  doi          = {10.1145/3465456.3467654},
  bibsource    = {dblp computer science bibliography, https://dblp.org}
}

@unpublished{hartline2013mechanism,
  title={Mechanism design and approximation},
  author={Hartline, Jason D},
  year={2013}, 
  note={Book manuscript available at \url{https://jasonhartline.com/MDnA/}}
}

@inproceedings{AbelsPS23,
  author       = {Andreas Abels and
                  Elias Pitschmann and
                  Daniel Schmand},
  editor       = {Kevin Leyton{-}Brown and
                  Jason D. Hartline and
                  Larry Samuelson},
  title        = {Prophet Inequalities over Time},
  booktitle    = {Proceedings of the 24th {ACM} Conference on Economics and Computation,
                  {EC} 2023, London, United Kingdom, July 9-12, 2023},
  pages        = {1--20},
  publisher    = {{ACM}},
  year         = {2023},
  url          = {https://doi.org/10.1145/3580507.3597741},
  doi          = {10.1145/3580507.3597741},
  bibsource    = {dblp computer science bibliography, https://dblp.org}
}

@inproceedings{FeldmanGL15,
  author       = {Michal Feldman and
                  Nick Gravin and
                  Brendan Lucier},
  title        = {Combinatorial Auctions via Posted Prices},
  booktitle    = {Proceedings of the Twenty-Sixth Annual {ACM-SIAM} Symposium on Discrete
                  Algorithms, {SODA} 2015, San Diego, CA, USA, January 4-6, 2015},
  pages        = {123--135},
  publisher    = {{SIAM}},
  year         = {2015},
}

@article{bergemann2019dynamic,
  title={Dynamic mechanism design: An introduction},
  author={Bergemann, Dirk and V{\"a}lim{\"a}ki, Juuso},
  journal={Journal of Economic Literature},
  volume={57},
  number={2},
  pages={235--274},
  year={2019},
  publisher={American Economic Association}
}

@unpublished{bergemann2010dynamic,
  title={Dynamic auctions: A survey},
  author={Bergemann, Dirk and Said, Maher},
  year={2010},
  note={Cowles Foundation Discussion Paper}
}

@inproceedings{MirrokniLTZ19,
  author       = {Vahab S. Mirrokni and
                  Renato Paes Leme and
                  Pingzhong Tang and
                  Song Zuo},
  title        = {Optimal Dynamic Auctions Are Virtual Welfare Maximizers},
  booktitle    = {The Thirty-Third {AAAI} Conference on Artificial Intelligence, {AAAI}
                  2019, The Thirty-First Innovative Applications of Artificial Intelligence
                  Conference, {IAAI} 2019, The Ninth {AAAI} Symposium on Educational
                  Advances in Artificial Intelligence, {EAAI} 2019, Honolulu, Hawaii,
                  USA, January 27 - February 1, 2019},
  pages        = {2125--2132},
  publisher    = {{AAAI} Press},
  year         = {2019},
  url          = {https://doi.org/10.1609/aaai.v33i01.33012125},
  doi          = {10.1609/AAAI.V33I01.33012125},
  bibsource    = {dblp computer science bibliography, https://dblp.org}
}

@inproceedings{AshlagiDH16,
  author       = {Itai Ashlagi and
                  Constantinos Daskalakis and
                  Nima Haghpanah},
  editor       = {Vincent Conitzer and
                  Dirk Bergemann and
                  Yiling Chen},
  title        = {Sequential Mechanisms with Ex-post Participation Guarantees},
  booktitle    = {Proceedings of the 2016 {ACM} Conference on Economics and Computation,
                  {EC} '16, Maastricht, The Netherlands, July 24-28, 2016},
  pages        = {213--214},
  publisher    = {{ACM}},
  year         = {2016},
  url          = {https://doi.org/10.1145/2940716.2940775},
  doi          = {10.1145/2940716.2940775},
  bibsource    = {dblp computer science bibliography, https://dblp.org}
}

@inproceedings{MirrokniLTZ18,
  author       = {Vahab S. Mirrokni and
                  Renato Paes Leme and
                  Pingzhong Tang and
                  Song Zuo},
  editor       = {{\'{E}}va Tardos and
                  Edith Elkind and
                  Rakesh Vohra},
  title        = {Non-clairvoyant Dynamic Mechanism Design},
  booktitle    = {Proceedings of the 2018 {ACM} Conference on Economics and Computation,
                  Ithaca, NY, USA, June 18-22, 2018},
  pages        = {169},
  publisher    = {{ACM}},
  year         = {2018},
  url          = {https://doi.org/10.1145/3219166.3219224},
  doi          = {10.1145/3219166.3219224},
  bibsource    = {dblp computer science bibliography, https://dblp.org}
}

@inproceedings{PapadimitriouPP16,
  author       = {Christos H. Papadimitriou and
                  George Pierrakos and
                  Christos{-}Alexandros Psomas and
                  Aviad Rubinstein},
  editor       = {Robert Krauthgamer},
  title        = {On the Complexity of Dynamic Mechanism Design},
  booktitle    = {Proceedings of the Twenty-Seventh Annual {ACM-SIAM} Symposium on Discrete
                  Algorithms, {SODA} 2016, Arlington, VA, USA, January 10-12, 2016},
  pages        = {1458--1475},
  publisher    = {{SIAM}},
  year         = {2016},
  url          = {https://doi.org/10.1137/1.9781611974331.ch100},
  doi          = {10.1137/1.9781611974331.CH100},
  bibsource    = {dblp computer science bibliography, https://dblp.org}
}

@inproceedings{MirrokniLTZ16,
  author       = {Vahab S. Mirrokni and
                  Renato {Paes Leme} and
                  Pingzhong Tang and
                  Song Zuo},
  title        = {Dynamic Auctions with Bank Accounts},
  booktitle    = {Proceedings of the Twenty-Fifth International Joint Conference on
                  Artificial Intelligence, {IJCAI}},
  pages        = {387--393},
  publisher    = {{IJCAI/AAAI} Press},
  year         = {2016},
}

@inproceedings{BalseiroML17,
  author       = {Santiago R. Balseiro and
                  Vahab S. Mirrokni and
                  Renato {Paes Leme}},
  title        = {Dynamic Mechanisms with Martingale Utilities},
  booktitle    = {Proceedings of the 2017 {ACM} Conference on Economics and Computation,
                  {EC}},
  pages        = {165},
  publisher    = {{ACM}},
  year         = {2017},
}

@inproceedings{0001DMS18,
  author       = {Shipra Agrawal and
                  Constantinos Daskalakis and
                  Vahab S. Mirrokni and
                  Balasubramanian Sivan},
  title        = {Robust Repeated Auctions under Heterogeneous Buyer Behavior},
  booktitle    = {Proceedings of the 2018 {ACM} Conference on Economics and Computation, {EC}},
  pages        = {171},
  publisher    = {{ACM}},
  year         = {2018}
}

@inproceedings{BergemannCW17,
  author       = {Dirk Bergemann and
                  Francisco Castro and
                  Gabriel Y. Weintraub},
  title        = {The Scope of Sequential Screening with Ex Post Participation Constraints},
  booktitle    = {Proceedings of the 2017 {ACM} Conference on Economics and Computation,
                  {EC}},
  pages        = {163--164},
  publisher    = {{ACM}},
  year         = {2017},
}

@article{krahmer2016optimality,
  title={Optimality of sequential screening with multiple units and ex post participation constraints},
  author={Kr{\"a}hmer, Daniel and Strausz, Roland},
  journal={Economics letters},
  volume={142},
  pages={64--68},
  year={2016},
  publisher={Elsevier}
}

@article{hausch1986multi,
  title={Multi-object auctions: Sequential vs. simultaneous sales},
  author={Hausch, Donald B},
  journal={Management Science},
  volume={32},
  number={12},
  pages={1599--1610},
  year={1986},
  publisher={INFORMS}
}

@article{myerson1981optimal,
  title={Optimal auction design},
  author={Myerson, Roger B},
  journal={Mathematics of Operations Research},
  volume={6},
  number={1},
  pages={58--73},
  year={1981},
  publisher={INFORMS}
}

@inproceedings{HajiaghayiKS07,
  author       = {Mohammad Taghi Hajiaghayi and
                  Robert D. Kleinberg and
                  Tuomas Sandholm},
  title        = {Automated Online Mechanism Design and Prophet Inequalities},
  booktitle    = {Proceedings of the Twenty-Second {AAAI} Conference on Artificial Intelligence,
                  July 22-26, 2007, Vancouver, British Columbia, Canada},
  pages        = {58--65},
  publisher    = {{AAAI} Press},
  year         = {2007},
  url          = {http://www.aaai.org/Library/AAAI/2007/aaai07-009.php},
  bibsource    = {dblp computer science bibliography, https://dblp.org}
}

@inproceedings{KleinbergW12,
  author       = {Robert Kleinberg and
                  S. Matthew Weinberg},
  editor       = {Howard J. Karloff and
                  Toniann Pitassi},
  title        = {Matroid prophet inequalities},
  booktitle    = {Proceedings of the 44th Symposium on Theory of Computing Conference,
                  {STOC} 2012, New York, NY, USA, May 19 - 22, 2012},
  pages        = {123--136},
  publisher    = {{ACM}},
  year         = {2012},
  url          = {https://doi.org/10.1145/2213977.2213991},
  doi          = {10.1145/2213977.2213991},
  bibsource    = {dblp computer science bibliography, https://dblp.org}
}

@inproceedings{CristiO24,
  author       = {Andr{\'{e}}s Cristi and
                  Sigal Oren},
  editor       = {Dirk Bergemann and
                  Robert Kleinberg and
                  Daniela Sab{\'{a}}n},
  title        = {Planning against a prophet: a graph-theoretic framework for making
                  sequential decisions},
  booktitle    = {Proceedings of the 25th {ACM} Conference on Economics and Computation,
                  {EC} 2024, New Haven, CT, USA, July 8-11, 2024},
  pages        = {806},
  publisher    = {{ACM}},
  year         = {2024},
  url          = {https://doi.org/10.1145/3670865.3673625},
  doi          = {10.1145/3670865.3673625},
  bibsource    = {dblp computer science bibliography, https://dblp.org}
}

@article{perezsalazar2024,
  author       = {Sebastian Perez{-}Salazar and
                  Victor Verdugo},
  title        = {Optimal Guarantees for Online Selection Over Time},
  journal      = {CoRR},
  volume       = {abs/2408.11224},
  year         = {2024},
  url          = {https://doi.org/10.48550/arXiv.2408.11224},
  doi          = {10.48550/ARXIV.2408.11224},
  eprinttype    = {arXiv},
  eprint       = {2408.11224},
  bibsource    = {dblp computer science bibliography, https://dblp.org}
}

@inproceedings{HartlineR08,
  author       = {Jason D. Hartline and
                  Tim Roughgarden},
  editor       = {Cynthia Dwork},
  title        = {Optimal mechanism design and money burning},
  booktitle    = {Proceedings of the 40th Annual {ACM} Symposium on Theory of Computing,
                  Victoria, British Columbia, Canada, May 17-20, 2008},
  pages        = {75--84},
  publisher    = {{ACM}},
  year         = {2008},
  url          = {https://doi.org/10.1145/1374376.1374390},
  doi          = {10.1145/1374376.1374390},
  bibsource    = {dblp computer science bibliography, https://dblp.org}
}

@article{ganesh2023combinatorial,
  title={Combinatorial Pen Testing (or Consumer Surplus of Deferred-Acceptance Auctions)},
  author={Ganesh, Aadityan and Hartline, Jason},
  journal={arXiv preprint arXiv:2301.12462},
  year={2023}
}

@inproceedings{hartline2009simple,
  title={Simple versus optimal mechanisms},
  author={Hartline, Jason D and Roughgarden, Tim},
  booktitle={Proceedings of the 10th ACM conference on Electronic commerce},
  pages={225--234},
  year={2009}
}

@unpublished{ezra2024optimalmechanismsconsumersurplus,
  title={Optimal Mechanisms for Consumer Surplus Maximization},
  author={Tomer Ezra and Daniel Schoepflin and Ariel Shaulker},
  year={2024}, 
  note={To appear in STOC 2025: 57th Annual ACM Symposium on Theory of Computing}
}

@inproceedings{OrenR21,
  author       = {Sigal Oren and
                  Oren Roth},
  editor       = {Michal Feldman and
                  Hu Fu and
                  Inbal Talgam{-}Cohen},
  title        = {Mechanisms for Trading Durable Goods},
  booktitle    = {Web and Internet Economics - 17th International Conference, {WINE}
                  2021, Potsdam, Germany, December 14-17, 2021, Proceedings},
  series       = {Lecture Notes in Computer Science},
  volume       = {13112},
  pages        = {262--279},
  publisher    = {Springer},
  year         = {2021},
  url          = {https://doi.org/10.1007/978-3-030-94676-0\_15},
  doi          = {10.1007/978-3-030-94676-0\_15},
  bibsource    = {dblp computer science bibliography, https://dblp.org}
}

@article{doval2024optimal,
  title={Optimal mechanism for the sale of a durable good},
  author={Doval, Laura and Skreta, Vasiliki},
  journal={Theoretical Economics},
  volume={19},
  number={2},
  pages={865--915},
  year={2024},
  publisher={Wiley Online Library}
}

@article{cremer2003rental,
  title={Rental of a durable good},
  author={Cr{\'e}mer, Jacques and Hariton, Cyril},
  year={2003},
  publisher={IDEI Working Paper}
}

@inproceedings{collina2024repeated,
  title={Repeated contracting with multiple non-myopic agents: Policy regret and limited liability},
  author={Collina, Natalie and Gupta, Varun and Roth, Aaron},
  booktitle={Proceedings of the 25th ACM Conference on Economics and Computation},
  pages={640--668},
  year={2024}
}

@misc{berzack2026optimal,
      title={Optimal Auction Design for Constrained Buyers}, 
      author={Batya Berzack and Rotem Oshman and Inbal Talgam-Cohen},
      year={2026},
      eprint={2606.30776},
      archivePrefix={arXiv},
      primaryClass={cs.GT},
      url={https://arxiv.org/abs/2606.30776}, 
}

@inproceedings{goldner2026multidimensional,
 author = {Goldner, Kira and Lundy, Taylor},
 booktitle = {Advances in Neural Information Processing Systems},
 editor = {D. Belgrave and C. Zhang and H. Lin and R. Pascanu and P. Koniusz and M. Ghassemi and N. Chen},
 pages = {20034--20059},
 publisher = {Curran Associates, Inc.},
 title = {Multidimensional Bayesian Utility Maximization: Tight Approximations to Welfare},
 url = {https://proceedings.neurips.cc/paper_files/paper/2025/file/1cead3207db5232c172bf76e4708a2b7-Paper-Conference.pdf},
 volume = {38},
 year = {2025}
}

\appendix
\section{Appendix Organization}
\begin{itemize}
    \item In Appendix~\ref{app:generelizations-and-definitions} we give the general definitions of the rental game, rental mechanism and the \swac, including some additional notation. We also prove the simplified version considered in the main body is w.l.o.g., and give an overview of \oir.
    \item In Appendix~\ref{app:rental-to-swacs} we provide the full results on the reduction from the rental mechanism design problem to the problem of predetermining an optimal sequence of \swacs.
    \item In Appendix~\ref{app:rental-mechs} we provide the full ironing procedure definition, together with some missing proofs from Section~\ref{sec:tech-overview-rental-mechs}, finding the optimal rental mechanisms.
\end{itemize}

\section{Our Model: Generalizations and Formal Definitions}\label{app:generelizations-and-definitions}
In this section we give general definitions of the rental game and the stagewise auction with seller cost setting, for which we present simplified versions in Section~\ref{subsec:rental}. We then go on to show that the simplifications are without loss of generality, which allows us to focus on the simplified versions in the main body of the paper. The main generalizations are that the stagewise auction can include randomization, as well as back-and-forth interaction between the designer and each agent.

\subsection{General Definitions}\label{sec:general}
The rental game setting and the stagewise auction setting are the same settings as described above, but the \swac itself, and thus the rental mechanism, is defined differently.

\paragraph{Stagewise auctions.} An auction $\A$ is a designer algorithm that governs the process of deciding on the number of units sold and the payment schedule, based on the agent behavior. The process of the auction is as follows:
\begin{itemize}
    \item For $i=1,\dots,n$:
    \begin{itemize}
        \item Designer declares the bidding set $\B_i$, and a mapping $\spay_i:\B_i\rightarrow\preals\cup\set{\psi}$, which represents the payment an agent would need to make for the $i$th unit, or a termination of the auction (given a bid $\psi$).
        \item Agent submits a bid $b\in\B_i$
        \item If $b=\psi$, terminate the auction.
        \item Else, agent receives the $i$th unit after paying $\spay_i(b)$
    \end{itemize}
\end{itemize}
The bidding sets $\B_i$ and mappings $\spay_i$ can be randomized, and the agent can have a mixed (randomized) strategy. Additionally the designer's choices of bidding sets and mappings at each day can be dependent on the history, that is, the previous bids of the agent. The agent is aware of the algorithm the designer uses to determine the future behavior of the auction.

Observe that at each day, the choice of the mapping $\spay_i$ can be random, but once $\spay_i$ is decided the payments per bid are fixed. It is an open question whether allowing the agent to buy lotteries would benefit the designer, for example, given bid $b$ the $i$th unit will be sold at price $p$ with probability $q$ and given away for free otherwise.

The \emph{history} of a \swac at some day is a pair $(R,\preals^\ast)$ of the randomness $R$ used by the designer, and the bids submitted by the agent up to the current day. A \emph{strategy} of an agent is a sequence of bids, and we denote by $\sigma^{\A}(v)$ the strategy chosen by the utility-maximizing agent $v$.

For auction $\A$ and agent $v$, we denote by $\xhat[\A]{v}$ the actual allocation of the agent, that is, $\xhat[\A]{v}=\x[\A]{\sigma^{\A}(v)}$. Similarly, we denote $\phat[\A]{v}=\p[\A]{\sigma^{\A}(v)}$ as the \emph{total} payment made by agent $v$, when they bid according to their strategy. 
For payments $p_1,\dots,p_{\x{b}}$ induced by bid $b$, the cumulative average payment at day $\ell$ is defined as $\insumi{1}{\ell}p_i/\ell$.

\paragraph{Rental mechanism and history.} The history of the rental game is defined as  a pair $(R,\preals^\ast)$ of the randomness $R$ used by the designer, and the bids submitted by all the agents up to the current timestep. From this the designer can deduce, deterministically, what the horizon at the current timestep is, and whether the asset is available or not. In the general definition, a rental mechanism $\M$ is defined just as in the simplified version, but with a general history and general \swacs.

\subsection{Simplifications of Stagewise Auctions are W.L.O.G.}\label{app:simplification-are-wlog}

\begin{lemma}\label{lemma:auction_wlog_deterministic}
    A \swac is w.l.o.g.  
 deterministic.
\end{lemma}
\begin{proof}
    Let $\A$ be a \swac with randomization, and denote its randomness by $S$. Recall that a random auction is a distribution over deterministic auctions. We will show that one of the deterministic auctions in the support of $\A$ yields as much reward as $\A$. Denote $R=\trew[\A]{v}$, which is the random variable whose expectation we aim to maximize.
    
     Assume for the sake of contradiction that for all $s\in S$ it holds that $\E[R\given S=s]<\E[R]$. Define the random variable $Y=\E[R]-R$, and observe that due to  our assumption,
     \begin{align}
         &\E[Y\given S]=\E[R]-\E[R|S]>0\label{eq:random-to-det}\\
         &\then \E[Y\given S]=|\E[Y\given S]|.
     \end{align}

    Using the law of total expectation,     
     \begin{align}\label{eq:total-prob}
         0=\E[Y]=\E[\E[Y\given S]]=\E\sqbr{|\E[Y\given S]|}.
     \end{align}

    From expectation properties, equation~(\ref{eq:total-prob}) yields $\E[Y|S]=0$, but this is a contradiction to equation~(\ref{eq:random-to-det}).

    Thus there is some $s$ for which $\E[R\given S=s]\geq\E[R]$, proving that the deterministic auction $\A'\coloneq \A\given S=s$ has at least as much reward as $\A$.
\end{proof}

\begin{lemma}\label{lemma:simplifications-wlog}
    Let $\A$ be a \swac in setting $\S$. There exists a \swac $\A'$ for the same setting, such that:
    \begin{enumerate}
        \item $\A'$ takes the simplified \swac form, that is, the agent submits a single bid at the beginning of the auction, resulting in a deterministic allocation and per-day payments.
        \item $\A'$ is truthful.
        \item $\E\sqbr{\A'}\geq\E\sqbr{\A}$.
    \end{enumerate}
\end{lemma}
\begin{proof}
    Let $\A$ be a deterministic \swac (w.l.o.g. due to Lemma~\ref{lemma:auction_wlog_deterministic}).
    Let $v\in\V$. If agent $v$ has a pure strategy, denote 
    $x_v=\xhat[\A]{v}$ (this is a constant) and 
    denote by $p_{v_1},\dots,p_{v_{x_v}}$ the payment schedule of $v$ in $\A$. Otherwise, agent $v$ has a mixed strategy that is a convex combination of different strategies. Since $\A$ is deterministic, there is a single strategy $s_v$ with allocation $x_v$ and payment schedule $p_{v_1},\dots,p_{v_{x_v}}$ that results in at least as much utility for $v$ as the mixed strategy.

    Define a new \swac $\A'$ that accepts a single bid from $\V$ at the beginning of the auction. For a bid $v\in\V$ it allocates $x_v$ units, and the payment schedule is exactly $p_{v_1},\dots,p_{v_{x_v}}$. It is easy to see that $\A'$ takes the simplified \swac form and that it is truthful, since all agents can get at least as much utility in $\A'$ by bidding truthfully as they did in $\A$, and there are no new payment schedules (so there is no new option that was not available in $\A'$).

    For an agent $v$ with pure strategy we have $\trew[\A']{v}=\trew[\A]{v}$ exactly, as the outcomes are exactly the same. For an agent $v$ with mixed strategy, a distribution over a set of strategies $\Sigma$, suppose for the sake of contradiction that $\trew[\A']{v}<\trew[\A]{v}$. Since $\trew[\A']{v}=\rew[\A]{v}{s_v}$, there is a subset $\Sigma^\ast\subseteq\Sigma$ of strategies with $\Pr\sqbr{\Sigma^\ast}$ that result in higher designer reward than $s_v$. Let $s^\ast\in\Sigma^\ast$.
    Since tie-breaking is in favor of the designer, we get that $\util[\A]{v}{s^\ast}<\util[\A]{v}{s_v}$. But this contradicts the utility-maximization of $v$ in $\A$, because removing $\Sigma^\ast$ from the pool of strategies $v$ might play in $\A$ would strictly increase $v$'s utility in $\A$. Therefore we deduce that $\trew[\A']{v}\geq\trew[\A]{v}$, resulting in $\E\sqbr{\A'}\geq\E\sqbr{\A}$.
\end{proof}

\subsection{\Oir: Structural Results}\label{appendix:oir}

In the body of the paper we focus on \swir, in which the agent requires nonnegative utility at every stage of the game. In contrast, we can consider the \oir setting, in which the agent can go into temporary deficit, as long as their overall utility will be nonnegative. Formally, individual rationality requires that for every agent $v$, there is some bid $b$ such that:
\begin{align*}
     \util{v}{b}\coloneq\utilformula{v}{b}\geq0
\end{align*}

In this setting, the reductions shown in Appendix~\ref{app:simplification-are-wlog} are still relevant and important. Once the rental problem is reduced to history-independent stagewise auctions with an over-time cost function, the mechanism design problem in the \oir setting is much simpler when compared to the \swir setting, especially due to Lemma~\ref{lemma:overall-ir-oneshot}.

In the \oir setting Myerson's lemma applies, since the over-time flavor of the stagewise auction has no practical effect, and thus a truthful stagewise auction is indeed \samono (see Definition~\ref{def:allocation-monotonicity}) and requires the unique payment rule of Myerson. We provide the proof here for conciseness. 

\begin{lemma}\label{lemma:oir-alloc-mono}
    In the \oir setting, all truthful stagewise auctions with seller cost are \samono.
\end{lemma}
\begin{proof}
    Let $\auc{A}$ be a truthful stagewise auction with seller cost. Let $w,v\in \V$ such that $w<v$. In the \oir setting, due to the utility maximization of the agents, we have that $\util{w}{w}\geq \util{w}{v}$ and $\util{v}{v}\geq \util{v}{w}$. Using the definition of $ \util{\cdot}{\cdot}$ in the \oir setting, we have $(\x{w}-\x{v})w\geq\p{w}-\p{v}$ and $(\x{w}-\x{v})v\leq\p{w}-\p{v}$, thus $(\x{w}-\x{v})w\geq (\x{w}-\x{v})v$. Since  $w<v$ it immediately follows that $\x{w}\leq \x{v}$.
\end{proof}

In addition, we have the following immediate result from this setting.
\begin{lemma}\label{lemma:overall-ir-oneshot}
    In the \oir setting, all truthful stagewise auctions with seller cost are \fr without loss of generality.
\end{lemma}
\begin{proof}
    Given an auction $\auc{A}$ with variable prices for some bids, we can design a new auction $\auc{A'}$ which has the same outcomes $\x{\cdot}$ and $\p{\cdot}$ for all bids, but with uniform prices, i.e. it is \fr. Since the utility of the agents in this setting is not affected by the payment schedule and only by the total payment, just as the designer's reward, all agents will submit the same bid in $\auc{A}$ and $\auc{A'}$, thus $\E\sqbr{\trew[A]{v}}=\E\sqbr{\trew[A']{v}}$ and we are done.
\end{proof}

Once we proved allocation-monotonicity and \fr payments schedules w.l.o.g., we can skip ahead to Section~\ref{sec:fr-mech}, where the optimal rental mechanism is detailed: all the conditions apply, for all reward functions.

\section{Full Results and Proofs for Section~\ref{subsec:rental-to-swacs}: Rental to Stagewise Auctions}\label{app:rental-to-swacs}
In this section we provide all the details and proofs needed for the two-part reduction of the rental game from an \emph{online} problem to an \emph{offline} problem, presented in Section~\ref{subsec:rental-to-swacs}. 

\subsection{Required Definitions and Notations}
For the reductions and definition in this section, all subscripts will refer to the \emph{horizon} of the rental, as opposed to the \emph{timestep}, since the proofs here depend on the horizon of the game rather than the current time. Specifically: $\history_h$ is the collection of bids given \emph{before} horizon $h$, and $\ind{h}$ is the indicator for the event that the asset was available at horizon $h$.

Given a rental mechanism $\mech{M}$ and a stagewise $h$-auction $\A_h$, we denote by $\brackets{\mech{M}_{-h},\A_h}$  a rental mechanism defined as,
\begin{align*}
    \brackets{\mech{M}_{-h},\A_h}(\history_l)=\begin{cases}
        \mech{M}(\history_l),& l\neq h\\
        \A_h,& l=h
    \end{cases},
\end{align*}
where $\history_l$ is the history at horizon $l$. That is, $\brackets{\mech{M}_{-h},\A_h}$ follows $\M$ for all timesteps, except for horizon $h$ in which $\A_h$ is used.

Given $n$ stagewise auctions $\A_n,\dots, \A_1$ for $n,\dots,1$ units (resp.), we define the rental mechanism $ \brackets{\A_n,\dots,\A_1}$ as:
\begin{align*}
    \brackets{\A_n,\dots,\A_1}(\history_h)=\A_h.
\end{align*}

Additionally we  define the $n$-rental mechanism $\brackets{\M,\A_h,\dots,\A_1}$ as:
\begin{align}
    \brackets{\M,\A_h,\dots,\A_1}(\history_l)=\begin{cases}
        \M(\history_l),& l>h\\
        \A_l,& l\leq h
    \end{cases}.
\end{align}
This rental mechanism can be history-dependent until horizon $h+1$, after which it becomes history-independent. Similarly we define $\brackets{\A_h,\M}$ when $\M$ is an $(h-1)$-rental mechanism.

Using the definition of \emph{reward} and our assumptions on the reward function, we have that given histories $\history_n,\dots,\history_1$ and agent valuations $v_n,\dots,v_1$,  for rental mechanism $\M$:
\begin{align*}
    \trew[\M]{v_n,\dots,v_1}&=\sumof{l}{1}{n}\ind[\M]{l}\xhat[\M(\history_l)]{v_l}\g\brackets{v_l,\frac{\phat[\M(\history_l)]{v_l}}{\xhat[\M(\history_l)]{v_l}}}.
\end{align*}

We define \emph{subgame optimal} in a recursive manner. 
\begin{definition}
    An $h$-stagewise auction $\A_h^*$ is said to be \textnormal{subgame optimal} if, given $h-1$ subgame optimal auctions $\A_{h-1},\dots,\A_1$ for horizons $h,\dots,1$,
    \begin{align*}
    \brackets{\A_h^*,\A_{h-1},\dots,\A_1}\in\arg\max_{\A_h}\E\sqbr{\brackets{\A_h,\A_{h-1},\dots,\A_1}}
    \end{align*}
\end{definition}

\subsection{Part 1: History Independence}\label{appendix:history-independent}

\begin{lemma}\label{lemma:auction-history-independent}
Let $\A_h$ be a truthful subgame optimal $(h,\D_h,\g,c)$-stagewise auction, for all $h\in\nat$ and distributions $\D_1,\D_2,\dots$. It holds that $\brackets{\A_n,\dots,\A_1}$ is an optimal $(n,\boldsymbol{\D},\g)$-rental mechanism, where $\boldsymbol{\D}=\brackets{\D_n,\dots,\D_1}$.
\end{lemma}
    \begin{proof}
    Let $\mopt$ be an optimal mechanism in the $(n,\boldsymbol{\D},\g)$-rental game. Recall that $\mopt$ is a mapping from a history, which also reflects the current horizon, to a stagewise auction with the corresponding number of units (i.e. at horizon $h$ we get an $h$-unit auction). We will prove by induction on $h$ (from $0$ to $n$) that the rental mechanism $\brackets{\mopt,\dots,\mopt,\A_h,\dots,\A_1}$ is optimal.

    The base of horizon $0$ is immediate. Now suppose the claim is true for horizons below $h$ and prove for $h$. By the induction hypothesis, the rental mechanism $\M^1=\brackets{\mopt,\dots,\mopt,\A_{h-1},\dots,\A_1}$  is optimal. We now analyze the mechanism $\M^2=\brackets{\mopt,\dots,\mopt,\A_{h},\dots,\A_1}=\brackets{\M^1_{-h},\A_h}$ and prove that it is also optimal. For an illustration of $\M^1$ and $\M^2$ see figure~\ref{figure:hi-mechanism}.

\begin{figure}[H]
\centering
\begin{subfigure}[t]{0.47\textwidth}
    \centering
    \includegraphics[width=\textwidth]{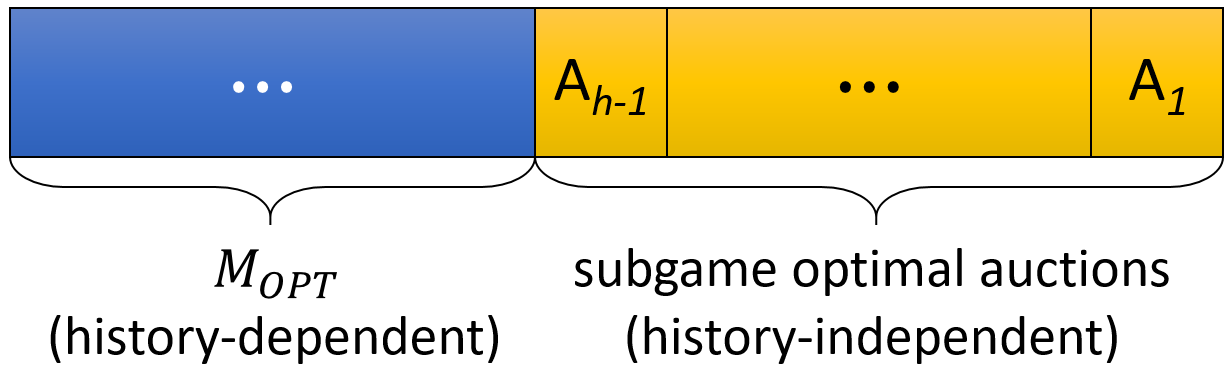}
    \caption{Rental mechanism $\M^1$, which is optimal}
\end{subfigure}
\begin{subfigure}[t]{0.5\textwidth}
    \centering
    \includegraphics[width=\textwidth]{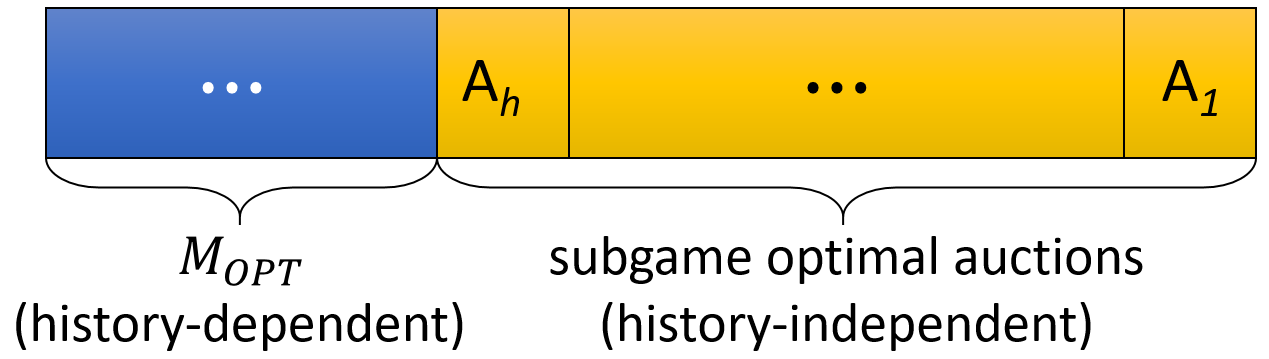}
    \caption{Rental mechanism $\M^2$}
\end{subfigure}
\caption[A comparison of mechanisms designed in the proof of history-independent rental mechanisms being w.l.o.g.]{A comparison of $\M^1$ and $\M^2$}
\label{figure:hi-mechanism}
\end{figure}

    For any auction $\A$, denote by $G^{\A}(v)$  the \emph{gross} reward of the auction, given an agent $v$ and history $\history_h$. Specifically: 
    \begin{align*}
        G^{\A}(v)=\xhat[\A]{v}\g\brackets{v,\frac{\phat[\A]{v}}{\xhat[\A]{v}}}=\trew[\A]{v}+c\brackets{\xhat[\A]{v}}
    \end{align*}

     Fix the valuations of the agents arriving in the rental game, and denote them according to their horizon: $v_n,\dots,v_1$.      
     We prove by backwards induction on the horizon that, for all $l\in \set{h+1,\dots,n+1}$:
     \begin{enumerate}
         \item $\history^{\M^1}_{l-1}=\history^{\M^2}_{l-1}$,\label{item:1_history_independent}
         \item $G^{\M^1(\history_l)}(v_l)=G^{\M^2(\history_l)}(v_l)$ for $l\leq n$ (superscripts in the histories omitted due Item~\ref{item:1_history_independent}), and  \label{item:2_history_independent}
         \item $\ind[\M^1]{l-1}=\ind[\M^2]{l-1}$.\label{item:3_history_independent}
     \end{enumerate}

    This will essentially show that $\M^1$ and $\M^2$ will behave exactly the same until horizon $h$ and yield the same reward from the first $n-h+1$ agents. This is very intuitive since for all $l\in \set{h+1,\dots,n+1}$, we have $\M^1(\history_l)=\mopt(\history_l)=\M^2(\history_l)$, but we will prove it formally. The induction is backwards on the horizon, which means that it goes forwards with the timesteps, from before the rental begins, and until horizon $h+1$.
     If $h=n$ this is immediate, since in this case $\M^1$ and $\M^2$ are identical, so suppose $h<n$. For the base of  $l=n+1$ this is trivial.  Suppose the claim is true for horizons $l>t$, and prove for  $t\in \set{h+1,\dots,n}$. By the induction hypothesis, $\ind[\M^1]{t}=\ind[\M^2]{t}$, so omit the superscript for simplicity. If $\ind{t}=0$, all items in the claim follow immediately. If $\ind{t}=1$, since both mechanisms are deterministic (as they're truthful), and $l>h$, we have that the auctions $\M^1(\history_t)$ and $\M^2(\history_t)$ are identical, thus the agents bids will be identical in either case, and so will the auction outcomes, proving the claim for $l=t$.
     
    In particular, we deduce from this that:
    \begin{align}
        &\sumof{l}{h+1}{n}G^{\M^1(\history_l)}(v_l)=\sumof{l}{h+1}{n}G^{\M^2(\history_l)}(v_l),\label{eq:123}\\
        &\ind{h}\coloneq\ind[\M^1]{h}=\ind[\M^2]{h},\textnormal{ and}\\
        &\history_h\coloneq\history^{\M^1}_{h}=\history^{\M^2}_{h}
    \end{align}

     Analyzing the total reward of mechanism $\M^2$, we split to cases depending on $\ind{h}$.
     
     If $\ind{h}=0$, then since the histories prior to horizon $h$ were the same in either mechanism, and the realized auctions were the same, then the next timestep in which the asset becomes available will be identical in either case. Since the auctions in $\M^1$ and $\M^2$ after horizon $h$ are identical, and due to Equation~\ref{eq:123}, we have that $\trew[\M^2]{v_n,\dots,v_1}=\trew[\M^1]{v_n,\dots,v_1}$.

     If $\ind{h}=1$ we have, due to the above claim:
     \begin{align*}
         \trew[\M^2]{v_n,\dots,v_1}&=\sumof{l}{h+1}{n}G^{\M^2(\history_l)}(v_l)+G^{\M^2(\history_h)}(v_h)+\sumof{l}{1}{h-1}G^{\M^2\brackets{\history_l^{\M^2}}}(v_l)\\
         &=\sumof{l}{h+1}{n}G^{\M^1(\history_l)}(v_l)+G^{\A_h}(v_h)+\sumof{l}{1}{h-1}\A_l(v_l)\\
         &\geq\sumof{l}{h+1}{n}G^{\M^1(\history_l)}(v_l)+G^{\M^1(\history_h)}(v_h)+\sumof{l}{1}{h-1}\A_l(v_l)\\
         &=\sumof{l}{h+1}{n}G^{\M^1(\history_l)}(v_l)+G^{\M^1(\history_h)}(v_h)+\sumof{l}{1}{h-1}G^{\M^1\brackets{\history_l^{\M^1}}}(v_l)\\
         &=\trew[\M^1_{-h}]{v_n,\dots,v_1}
     \end{align*}
    The inequality holds due to our induction hypothesis, that $\A_{h-1},\dots,\A_1$ are subgame optimal $(h-1)$-rental mechanism, and since $\A_h$ is subgame optimal.

     Since $\trew[\M^2]{v_n,\dots,v_1}\geq\trew[\M^1]{v_n,\dots,v_1}$ for any realization of agents' valuations, we have $\E[\M^2]\geq\E[\M^1]=\E[\mopt]$,
    which proves that $\M^2$ is indeed optimal and completes the proof.
    
    \end{proof} 

\begin{lemma}\label{lemma:rental_history_independent}
Any rental mechanism
$\mech{M}$
can be converted
into a rental mechanism $\mech{M}'$,
such that
\begin{enumerate}
\item The designer's expected reward does not change, i.e. $\E\sqbr{\M}=\E\sqbr{\M'}$ ,
and
\item $\mech{M'}$ is history-independent.
\end{enumerate}

\end{lemma}

\begin{proof}
     This follows immediately from Lemma~\ref{lemma:auction-history-independent}, since a history-independent mechanism is constructed and proven to be optimal.
\end{proof}

\subsection{Part 2: Defining the Cost Functions}%

\begin{lemma}
\label{lemma:over-time-cost-subgame-optimal}
For every horizon $h\in\nat$, an optimal stagewise $\brackets{h,\D_h,\g,c_{h}^{\D,\g}}$-auction is subgame optimal, where $c_{h}^{\boldsymbol{\D}_{\text{last }h},\g}$ is the $(h,\boldsymbol{\D}_{\text{last }h},\g)$-over-time cost function.
\end{lemma}
\begin{proof}
    For simplicity, in this proof we write $c_h$ instead of $c_{h}^{\boldsymbol{\D}_{\text{last }h},\g}$ and $R_h$ instead of $R_h^{\boldsymbol{\D}_{\text{last }h},\g}$. Recall that $R_h^{\boldsymbol{\D}_{\text{last }h},\g}$ represents the expected reward from the final $h$ days, using an optimal $(h,\D_{\text{last }h},\g)$ rental mechanism.
    
    Suppose the claim is false. Thus there is a horizon $h$ and an optimal $\brackets{h,\D_h,\g,c}$-auction which we call $\A$, such that $\A$ is not subgame optimal. Thus, by definition of subgame optimal, there are subgame optimal auctions $\A_h,\dots,\A_1$ such that
    \begin{align*}
        \E\sqbr{\brackets{\A,\A_{h-1},\dots,\A_1}}<\E\sqbr{\brackets{\A_h,\A_{h-1},\dots,\A_1}}
    \end{align*}

    The last transition is due to Lemma~\ref{lemma:auction-history-independent}.

    It holds that, since $A_{h-1},\dots,\A_1$ are subgame optimal:
    \begin{align*}
        \E\sqbr{\A_h,\A_{h-1},\dots,\A_1}&=\E\sqbr{\A_h(v_h)+\sumof{l}{h-1}{1}\ind[\A_h,\A_{h-1},\dots,\A_1]{l}\A_l(v_l)}\\
        &=\E\sqbr{\A_h(v_h)+\sumof{l}{1}{h-\max\set{\xhat[\A_h]{v_h},1}}\ind[\A_h,\A_{h},\dots,\A_1]{l}\A_l(v_l)\given\ind[\A_h,\A_{h},\dots,\A_1]{h-\max\set{\xhat[\A_h]{v_h},1}}=1}\\
        &=\E\sqbr{\A_h(v_h)+\trew[\A_{h-\max\set{\xhat[\A_h]{v_h},1}},\dots,\A_1]{v_{h-\max\set{\xhat[\A_h]{v_h},1}},\dots,v_1}}\\
        &=\E\sqbr{\A_h(v_h)+R_{h-\max\set{\xhat[\A_h]{v_h},1}}}
    \end{align*}

    And in the same manner, then using the over-time cost function of $\A$:
    \begin{align}
        \E\sqbr{\A,\A_{h-1},\dots,\A_1}&=\E\sqbr{\A_h(v_h)+R_{h-\max\set{\xhat[\A]{v_h},1}}}\\
        &=\E\sqbr{\trew[\A]{v_h}}+R_{n-1}\label{eq:hi1}
    \end{align}

    Define an auction $\A_h'$ which follows $\A_h$ exactly, but its reward is with respect to the cost function is $c_{h}$ (i.e.\ same auction, different setting). Since it follows $\A_h$ we also have:
    \begin{align}\label{eq:hi3}
        &\E\sqbr{\A_h',\A_{h-1},\dots,\A_1}=\E\sqbr{\A_h,\A_{h-1},\dots,\A_1}>\E\sqbr{\A,\A_{h-1},\dots,\A_1}
    \end{align}
    and
    \begin{align}\label{eq:hi2}
        \E\sqbr{\A_h',\A_{h-1},\dots,\A_1}=\E\sqbr{\A_h'(v_h)+R_{h-\max\set{\xhat[\A_h']{v_h},1}}}=\E\sqbr{\A_h(v_h)+R_{h-\max\set{\xhat[\A_h]{v_h},1}}}.
    \end{align}

    Due to the optimality of $\A$, extending Eq.~\ref{eq:hi1} we have:
    \begin{align*}
        \E\sqbr{\A,\A_{h-1},\dots,\A_1}
        &\geq\E\sqbr{\trew[\A_h']{v_h}}+R_{n-1}\\
        &=\E\sqbr{\A_h'(v_h)-R_{n+1}+R_{h-\max\set{\xhat[\A_h']{v_h},1}}}+R_{n-1}\\
        &\underbrace{=}_{\ref{eq:hi2}}\E\sqbr{\A_h(v_h)+R_{h-\max\set{\xhat[\A_h]{v_h},1}}}\\
        &\underbrace{=}_{\ref{eq:hi2}}\E\sqbr{\A_h',\A_{h-1},\dots,\A_1},
    \end{align*}

    but this is a contradiction to Eq.~\ref{eq:hi3}, completing the proof.
\end{proof}

\section{Full Results and Proofs for Section~\ref{sec:tech-overview-rental-mechs}: Finding Optimal Rental Mechanisms}\label{app:rental-mechs}
\subsection{Virtual Value Ironing}
\label{sec:ironing}

We formally define the ironing procedure that is used in Section~\ref{sec:tech-overview-rental-mechs}.

    \begin{definition}[Ironed virtual valuations~\cite{myerson1981optimal,HartlineR08}]\label{def:ironing}
    Given a distribution function $F(\cdot)$ and any function $\theta:\V\rightarrow\reals$, the \textnormal{ironed virtual value function}, $\Bar{\theta}$, is constructed as follows:

    \begin{enumerate}
        \item For $q\in[0,1]$, define $h(q)=\theta\brackets{F^{-1}(q)}$.
        \item Define $H(q)=\int_0^q h(r)dr$.
        \item Define $\Psi$ as the convex hull of $H$ --- the largest convex function bounded above by $H$ for all $q\in[0,1]$.
        \item Define $\psi(q)$ as the derivative of $\Psi(q)$, where defined, and extend to all of $[0,1]$ by right-continuity.
        \item Finally, $\Bar{\theta}(v)=\psi(F(v))$.
    \end{enumerate}
    \end{definition}
\begin{lemma}[Generalized restatement of Lemma 2.8 from~\cite{HartlineR08}]\label{lemma:ironing-properties}
    Let $F$ be a distribution function, let $\theta:\V\rightarrow\reals$, and let $\x{v}$ be a monotone allocation rule. 
    Define $\Psi,H$ and $\Bar{\theta}$ as in Definition~\ref{def:ironing}.
    Then 
    \begin{equation}
        \E\sqbr{\theta(v)\x{v}}\leq \E\sqbr{\Bar{\theta}(v)\x{v}},
    \end{equation}
    with equality holding if and only if $\frac{d}{d v}\x{v}=0$ wherever $\Psi(F(v))<H(F(v))$.
\end{lemma}

The proof of~\cite{HartlineR08} can be applied as is: it does not rely on any properties of $\theta$ that are specific to their work.

\subsection{Proof of Lemma~\ref{lemma:irn-maximizer-maximizes-expected}}
\label{sec:proof-optimal-rental-mechs}
The following proof relies strongly on the ironing procedure definition.
\begin{proof}[Proof of Lemma~\ref{lemma:irn-maximizer-maximizes-expected}]
 Let $\S$ be a \swac setting, and $\theta:\V\rightarrow\reals$ be a function, and $\Bar{\theta}$ its ironing
    Let $\alloc$ be a monotone non-decreasing allocation rule that is pointwise maximizing of $\Bar{\theta}(v)\x{v}-c\brackets{\x{V}}$.
    
    Define $\Psi,H$ from the ironing procedure used to define $\Bar{\theta}(\cdot)$.
    At points $v$ where $\Psi(F(v))<H(F(v))$, $\Psi$ is locally linear as the convex hull of $H$, and hence $\Bar{\theta}(v)$ is locally constant. Thus from Lemma~\ref{lemma:ironing-properties} we are guaranteed that $\frac{d}{dv}\x{v}=0$ at all such points.
    It follows that
    \begin{equation}
         \E\sqbr{\theta(v)\x{v}}=\E\sqbr{\Bar{\theta}(v)\x{v}}.
    \end{equation}
    Therefore, also due to $\x{\cdot}$ being pointwise maximizing w.r.t. $\Bar{\theta}$, the following expression
    \begin{align}\label{eq:best-rew}
        \E\sqbr{\theta(v)\x{v}-c\brackets{\x{v}}}
    \end{align}
    is maximized among all monotone non-decreasing allocation rules.
\end{proof}

\end{document}